%% file: main.tex
\def\confversion{0}
\def\ifconf{\ifnum\confversion=1}
\def\ifnotconf{\ifnum\confversion=0}

\documentclass[11 pt]{article}

\def\showauthornotes{0}
\def\showkeys{0}
\def\showdraftbox{0}

\input{macros}

\input{abbrev} 

\newcommand{\vnote}{\Authornote{VA}}

\newcommand{\mnote}{\Authornote{MT}}

\usepackage[top=1in, bottom=1in, left=1.25in, right=1.25in]{geometry}
\begin{document}

\title{Approximating Constraint Satisfaction Problems on High-Dimensional Expanders}
\author{
Vedat Levi Alev \thanks{ Supported by NSERC Discovery Grant 2950-120715, NSERC
  Accelerator Supplement 2950-120719, and partially supported by NSF awards
  CCF-1254044 and CCF-1718820. {\tt vlalev@uwaterloo.ca}.} 
\and 
Fernando Granha Jeronimo\thanks{Supported in part by NSF grants CCF-1254044 and CCF-1816372.  {\tt granha@uchicago.edu}. } 
\and 
Madhur Tulsiani \thanks{Supported by NSF grants CCF-1254044 and CCF-1816372. \tt madhurt@ttic.edu} 
}

\setcounter{page}{0}

\date{}

\maketitle
\draftbox
\thispagestyle{empty}

\input{abstract}

\newpage

\ifnotconf
\pagenumbering{roman}
\tableofcontents
\clearpage
\fi

\pagenumbering{arabic}
\setcounter{page}{1}


\section{Introduction}\label{sec:intro}
\input{intro.tex}

\section{Preliminaries and Notation}\label{sec:prelims}
\input{prelims.tex}

\section{Proof Overview: Approximating \maxfourxor}\label{sec:example}
\input{example.tex}

\section{Walks}\label{sec:walks}
\input{walks.tex}

\section{Spectral Analysis of Swap Walks}\label{sec:asd}
\input{asd.tex}

\section{Approximating Max-$k$-CSP}\label{sec:csp}
\input{hdx_brs.tex}

\section{High-Dimensional Threshold Rank}\label{sec:trank}
\input{trank.tex}

\section{Quantum k-local Hamiltonian}\label{sec:quantum_hamiltonian}
\input{quantum_hamiltonian.tex}

\section*{Acknowledgements}
We thank Anand Louis for several enlightening discussions in the initial phases
of this work. We are also grateful to the anonymous reviewers for several helpful 
suggestions.

\bibliographystyle{alphaurl}
\bibliography{macros,madhur}

\appendix
\section{From Local to Global Correlation}\label{app:appenditis}
\input{appendix.tex}

\end{document}

%% file: macros.tex
\usepackage{xspace,xcolor,enumerate}
\usepackage{amsmath,amssymb}
\usepackage{amsthm}
\usepackage[toc,page]{appendix}
\usepackage{thmtools}
\usepackage{thm-restate}
\usepackage{color,graphicx}
\usepackage{boxedminipage}

\ifnum\showkeys=1
\usepackage[color]{showkeys}
\fi

\definecolor{darkred}{rgb}{0.5,0,0}
\definecolor{darkgreen}{rgb}{0,0.35,0}
\definecolor{darkblue}{rgb}{0,0,0.55}

\usepackage[pdfstartview=FitH,pdfpagemode=UseNone,colorlinks,linkcolor=darkblue,filecolor=darkred,citecolor=darkgreen,urlcolor=darkred,pagebackref]{hyperref}

\usepackage[capitalise,nameinlink]{cleveref}
\usepackage[T1]{fontenc}
\usepackage{mathtools,dsfont,bbm}
\usepackage{mathpazo}

\ifnum\showauthornotes=1
\newcommand{\Authornote}[2]{{\sf\small\color{red}{[#1: #2]}}}
\newcommand{\Authorcomment}[2]{{\sf \small\color{gray}{[#1: #2]}}}
\newcommand{\Authorfnote}[2]{\footnote{\color{red}{#1: #2}}}
\else
\newcommand{\Authornote}[2]{}
\newcommand{\Authorcomment}[2]{}
\newcommand{\Authorfnote}[2]{}
\fi

\ifnum\showdraftbox=1
\newcommand{\draftbox}{\begin{center}
  \fbox{%
    \begin{minipage}{2in}%
      \begin{center}%
        \begin{Large}%
          \textsc{Working Draft}%
        \end{Large}\\
        Please do not distribute%
      \end{center}%
    \end{minipage}%
  }%
\end{center}
\vspace{0.2cm}}
\else
\newcommand{\draftbox}{}
\fi

\newtheorem{theorem}{Theorem}[section]

\newtheorem{definition}[theorem]{Definition}
\newtheorem{lemma}[theorem]{Lemma}
\newtheorem{remark}[theorem]{Remark}
\newtheorem{proposition}[theorem]{Proposition}
\newtheorem{corollary}[theorem]{Corollary}
\newtheorem{claim}[theorem]{Claim}
\newtheorem{fact}[theorem]{Fact}

\newtheorem{algo}[theorem]{Algorithm}

\newenvironment{algorithm}[3]
        {\noindent\begin{boxedminipage}{\textwidth}\begin{algo}[#1]\ \par
        {\begin{tabular}{r l}
        \textbf{Intput} & #2\\
        \textbf{Output} & #3
        \end{tabular}\par\enskip}}
        {\end{algo}\end{boxedminipage}}

\def\FullBox{\hbox{\vrule width 6pt height 6pt depth 0pt}}

\def\qed{\ifmmode\qquad\FullBox\else{\unskip\nobreak\hfil
\penalty50\hskip1em\null\nobreak\hfil\FullBox
\parfillskip=0pt\finalhyphendemerits=0\endgraf}\fi}

\def\qedsketch{\ifmmode\Box\else{\unskip\nobreak\hfil
\penalty50\hskip1em\null\nobreak\hfil$\Box$
\parfillskip=0pt\finalhyphendemerits=0\endgraf}\fi}

\def\to{\rightarrow}
\def\eps{\varepsilon}
\def\epsilon{\varepsilon}

\def\eps{\epsilon}

\def\phi{\varphi}
\def\cal{\mathcal}

\newcommand{\ie}{i.e.,\xspace}

\newcommand{\etal}{et al.\xspace}

\newcommand{\mper}{\,.}
\newcommand{\mcom}{\,,}

\newcommand{\R}{{\mathbb R}}
\newcommand{\E}{{\mathbb E}}

\newcommand{\B}{\{0,1\}\xspace}

\usepackage{nicefrac}

\let\nfrac=\nicefrac

\newcommand{\abs}[1]{\ensuremath{\left\lvert #1 \right\rvert}}

\newcommand{\norm}[1]{\ensuremath{\left\lVert #1 \right\rVert}}
\newcommand{\smallnorm}[1]{\ensuremath{\lVert #1 \rVert}}

\newcommand{\mydot}[2]{\ensuremath{\left\langle #1, #2 \right\rangle}}

\newcommand{\ip}[2] {\ensuremath{\langle #1 , #2 \rangle}}

\newcommand{\one}{{\mathbf{1}}}

\newcommand{\qary}{[q]}

\newcommand{\Esymb}{\mathbb{E}}
\newcommand{\Psymb}{\mathbb{P}}

\newcommand{\Varsymb}{\mathrm{Var}}
\DeclareMathOperator*{\ExpOp}{\Esymb}

\makeatletter
\def\Pr#1{%
    \ProbabilityRender{\Psymb}{#1}%
}

\def\Ex#1{%
    \ProbabilityRender{\Esymb}{#1}%
}

\def\Var#1{%
    \ProbabilityRender{\Varsymb}{#1}%
}

\def\ProbabilityRender#1#2{
  \@ifnextchar\bgroup%
  {\renderwithdist{#1}{#2}}
   {\singlervrender{#1}{#2}}
}
\def\singlervrender#1#2{%
   \ensuremath{\mathchoice
       {{#1}\left[ #2 \right]}
       {{#1}[ #2 ]}
       {{#1}[ #2 ]}
       {{#1}[ #2 ]}
   }
}
\def\renderwithdist#1#2#3{%
   \@ifnextchar\bgroup
   {\superfancyrender{#1}{#2}{#3}}
   {\ensuremath{\mathchoice
      {\underset{#2}{#1}\left[ #3 \right]}
      {{#1}_{#2}[ #3 ]}
      {{#1}_{#2}[ #3 ]}
      {{#1}_{#2}[ #3 ]}
     }
   }
}
\def\superfancyrender#1#2#3#4#5{
   \ensuremath{\mathchoice
      {\underset{#1}{{#1}}\left#4 #3 \right#5}
      {{#1}_{#2}#4 #3 #5}
      {{#1}_{#2}#4 #3 #5}
      {{#1}_{#2}#4 #3 #5}
   }
}
\makeatother

\def\expop{\ExpOp}

\newfont{\inhead}{eufm10 scaled\magstep1}
\newcommand{\deffont}{\sf}

\newcommand{\calD}{{\cal D}}

\newcommand{\calU}{{\cal U}}

\DeclareMathOperator\supp{Supp}

\DeclareMathOperator{\cov}{\operatorname {Cov}}

\DeclareMathOperator{\im}{\operatorname{im}}

\newcommand{\ee}{\ensuremath{\mathrm e}}

\newcommand{\problemmacro}[1]{\textsf{#1}}

\newcommand{\maxkcsp}{\problemmacro{MAX k-CSP}\xspace}

\newcommand{\maxfourxor}{\problemmacro{MAX 4-XOR}\xspace}

\newcommand{\uniquegames}{\problemmacro{Unique Games}\xspace}
\newcommand{\ugc}{\problemmacro{Unique Games Conjecture}\xspace}

\newcommand{\inparen}[1]{\left(#1\right)}             
\newcommand{\inbraces}[1]{\left\{#1\right\}}           
\newcommand{\insquare}[1]{\left[#1\right]}             



%% file: abbrev.tex
\DeclareSymbolFont{extraup}{U}{zavm}{m}{n}
\DeclareMathSymbol{\varheart}{\mathalpha}{extraup}{86}
\DeclareMathSymbol{\vardiamond}{\mathalpha}{extraup}{87}

\def\RR{\mathbb R}

\def\tand{\quad\text{and}\quad}

\def\matr#1{\mathsf{#1}}
\def\vecc#1{\boldsymbol{#1}}
\def\rand#1{\mathbf{#1}}
\def\rv#1{\rand #1}

\def\one{\mathbf 1}

\def\Cc{\mathcal C}

\def\Ee{\mathcal E}

\def\OPT{\mathsf{OPT}}
\def\SAT{\mathsf{SAT}}

\def\SDP{\mathsf{SDP}}

\def\ee{\varepsilon}

\def\Ins{\mathfrak I}
\def\assn{\eta}
\def\tree{\mathcal T}
\def\Up{\matr U}
\def\Dee{\matr D}
\def\Aye{\matr A}
\def\Bee{\matr B}
\def\Why{\matr Y}
\def\Www{\matr W}
\def\Enn{\matr N}
\def\Ess{\matr S}
\def\Ide{\matr I}

\def\Emm{\matr M}

\def\eff{\vecc f}

\def\aye{\mathfrak a}
\def\ess{\mathfrak s}
\def\tee{\mathfrak t}
\def\bee{\mathfrak b}

\DeclareMathOperator{\SwapT}{Swap}

\DeclareMathOperator{\Tr}{Tr}

\DeclareMathOperator{\Rank}{rank}
\DeclarePairedDelimiter\set{\lbrace}{\rbrace}
\DeclarePairedDelimiter\parens{\lparen}{\rparen}

\DeclarePairedDelimiter\Abs{|}{|}

\newcommand{\hrswap}[4]{\Ess_{#1,#2}^{\left(#3,#4\right)}}

\newcommand{\rswap}[3]{\Ess_{#1,#2}^{\left(#3\right)}}
\newcommand{\rfswap}[2]{\Ess_{#1, #2}}

\newcommand{\rnwalk}[3]{\Enn_{#1,#2}^{\left(#3\right)}}

\newcommand{\nwalk}[2]{\Enn_{#1,#1}^{\left(#2\right)}}

%% file: abstract.tex
We consider the problem of approximately solving constraint satisfaction
problems with arity $k > 2$ ($k$-CSPs) on instances satisfying certain expansion
properties, when viewed as hypergraphs.
Random instances of $k$-CSPs, which are also highly expanding,  are well-known
to be hard to approximate using known algorithmic techniques (and are widely
believed to be hard to approximate in polynomial time). 
However, we show that this is not necessarily the case for instances where the
hypergraph is a \emph{high-dimensional expander}.

We consider the spectral definition of high-dimensional expansion used by Dinur
and Kaufman [FOCS 2017] to construct certain primitives related to PCPs. 
They measure the expansion in terms of a parameter $\gamma$ which is the
analogue of the second singular value for expanding graphs. 
Extending the results by Barak, Raghavendra and Steurer [FOCS 2011] for 2-CSPs,
we show that if an instance of \maxkcsp over alphabet $[q]$ is a
high-dimensional expander with parameter $\gamma$, then it is possible to
approximate the maximum fraction of satisfiable constraints up to an additive error
$\eps$ using $q^{O(k)} \cdot (k/\eps)^{O(1)}$ levels of the sum-of-squares SDP
hierarchy, provided $\gamma \leq \eps^{O(1)} \cdot (1/(kq))^{O(k)}$.

Based on our analysis, we also suggest a notion of threshold-rank for
hypergraphs, which can be used to extend the results for approximating 2-CSPs on
low threshold-rank graphs. 
We show that if an instance of \maxkcsp has threshold rank $r$ for a threshold
$\tau = (\eps/k)^{O(1)} \cdot (1/q)^{O(k)}$, then it is possible to approximately
solve the instance up to additive error $\eps$, using 
$r \cdot q^{O(k)} \cdot (k/\eps)^{O(1)}$ levels of the sum-of-squares hierarchy.
As in the case of graphs, high-dimensional expanders (with sufficiently small
$\gamma$) have threshold rank 1 according to our definition.
%


%% file: intro.tex
%
%
We consider the problem of approximately solving constraint satisfaction
problems (CSPs) on instances satisfying certain expansion properties. 
The role of expansion in understanding the approximability of CSPs with two
variables in each constraint (2-CSPs) has been extensively studied and has led to several
results, which can also be viewed as no-go results for PCP constructions (since PCPs
are hard instances of CSPs). 
It was shown by Arora \etal \cite{AKKSTV08} (and strengthened by Makarychev and
Makarychev \cite{MakarychevM11}) that the \uniquegames problem is easily
approximable on expanding instances, thus proving that the \ugc  of
Khot \cite{Khot02:unique} cannot be true for expanding instances. 
Their results were extended to all 2-CSPs and several partitioning problems in
works by Barak, Raghavendra and Steurer \cite{BarakRS11}, Guruswami and Sinop
\cite{GuruswamiS11}, and Oveis Gharan and Trevisan \cite{OGT15} under much
weaker notions of expansion.

We consider the following question: 
\begin{center}
\emph{When are expanding instances of $k$-CSPs easy for $k > 2$?}
\end{center}
At first glance, the question does not make much sense, since random instances
of $k$-CSPs (which are also highly expanding) are known to be hard for various
models of computation (see \cite{KothariMOW17} for an excellent survey).
However, while the kind of expansion exhibited by random instances of CSPs is
useful for constructing codes, it is not sufficient for constructing primitives for
PCPs, such as locally testable codes \cite{BHR05}.
On the other hand, objects such as high-dimensional expanders, which possess a
form of ``structured multi-scale expansion'' have been useful in constructing
derandomized direct-product and direct-sum tests (which can be viewed as locally
testable distance amplification codes) \cite{DinurK17}, lattices with large
distance \cite{KaufmanM18}, list-decodable direct product codes
\cite{DinurHKNT18}, and are thought to be intimately connected with PCPs
\cite{DinurK17}. Thus, from the PCP perspective, it is more relevant to ask if
this form of expansion can be used to efficiently approximate constraint
satisfaction problems.

\paragraph{Connections to coding theory.}
Algorithmic results related to expanding CSPs are also relevant for the problem of
\emph{decoding} locally testable codes.
Consider a code $C$ constructed via $k$-local operations (such as $k$-fold
direct-sum) on a base code $C_0$ with smaller distance. Then, a codeword in $C$
is simply an instance of a CSP, where each bit places a constraint on $k$ bits
(which is $k$-XOR in case of direct sum) of the relevant codeword in
$C_0$. 
The task of decoding a noisy codeword is then equivalent to 
finding an assignment in  $C_0$, satisfying the maximum number of 
constraints for the above instance. 
Thus, algorithms for solving CSPs on expanding instances may lead to new decoding 
algorithms for codes obtained by applying local operations to a base code.
\mnote{Changed the text here a little to highlight the coding motivation.}
In fact, the list decoding algorithm for direct-product codes by Dinur \etal
\cite{DinurHKNT18} also relied on algorithmic results for expanding unique games.
Since all constructions of locally testable codes need to have at least some
weak expansion \cite{DinurK12}, it is interesting to understand what notions of
expansion are amenable to algorithmic techniques.
\vspace{-15 pt}
\paragraph{High-dimensional expanders and our results.} 
A $d$-dimensional expander is a downward-closed hypergraph (simplicial complex),
say $X$,  with edges of size at most $d+1$, such that for every hyperedge $\aye
\in X$ (with $\abs{\aye} \leq d-1$), a certain ``neighborhood graph''
$G(X_{\aye})$ is a spectral expander\footnote{While there are several
  definitions of high-dimensional expanders, we consider the one by Dinur and
  Kaufman \cite{DinurK17}, which is most closely related to spectral expansion,
  and was also the one shown to be related to PCP applications. Our results also
  work for a weaker but more technical definition by Dikstein \etal
  \cite{DiksteinDFH18}, which we defer till later.}. 
Here, the graph $G(X_{\aye})$ is defined to have the vertex set $\set{i ~|~ \aye
  \cup \{i\} \in X}$ and edge-set $\set{{i,j} ~|~ \aye \cup \{i,j\} \in X}$. 
If the (normalized) second singular value of each of the neighborhood graphs is
bounded by $\gamma$, $X$ is said to be a  $\gamma$-high-dimensional expander
($\gamma$-HDX).

Note that (the downward closure of) a random sparse $(d+1)$-uniform hypergraph,
say with $n$ vertices and $c \cdot n$ edges, is very unlikely to be a
$d$-dimensional expander. With high probability, no two hyperedges share more
than one vertex and thus for any $i \in [n]$, the neighborhood graph $G_i$ is
simply a disjoint union of cliques of size $d$, which is very far from an
expander. While random hypergraphs do not yield high-dimensional expanders, such
objects are indeed known to exists via (surprising) algebraic constructions  
\cite{LubotzkySV05a, LubotzkySV05b, KaufmanO18Constr,ConlonTZ18} and are known
to have several interesting properties and applications
~\cite{KaufmanKL16,DinurHKNT18,KaufmanM17,KaufmanO18Walk,DiksteinDFH18,DinurK17,ParzanchevskiRT2016}.

Expander graphs can simply be thought of as the one-dimensional case of the
above definition. 
The results of Barak, Raghavendra and Steurer \cite{BarakRS11} for 2-CSPs 
yield that if the constraint graph of a 2-CSP instance (with
size $n$ and alphabet size $q$) is a sufficiently good (one dimensional)
spectral expander, then one can efficiently find solutions satisfying $\OPT -
\eps$ fraction of constraints, where $\OPT$ denotes the maximum fraction of
constraints satisfiable by any assignment.
Their algorithm is based on $(q/\eps)^{O(1)}$ levels of the Sum-of-Squares (SoS)
SDP hierarchy, and the expansion requirement on the constraint graph is that the
(normalized) second singular value should be at most $(\eps/q)^{O(1)}$.
We show a similar result for $k$-CSPs when the corresponding simplicial complex
$X_{\Ins}$, which is obtained by including one hyperedge for each constraint and
taking a downward closure, is a sufficiently good $(k-1)$-dimensional expander.
\begin{theorem}[Informal]
Let $\Ins$ be an instance of \maxkcsp on $n$ variables taking values over an
alphabet of size $q$, and let $\eps > 0$. Let the simplicial complex $X_{\Ins}$
be a $\gamma$-HDX with $\gamma = \eps^{O(1)} \cdot (1/(kq))^{O(k)}$. 
Then, there is an algorithm based on $(k/\eps)^{O(1)} \cdot q^{O(k)}$ levels of the
Sum-of-Squares hierarchy, which  produces an assignment satisfying $\OPT - \eps$
fraction of the constraints.
\end{theorem}
\begin{remark}
While the level-$t$ relaxation for \maxkcsp can be solved in time $(nq)^{O(t)}$
\cite{RW17:sos}, the rounding algorithms used by \cite{BarakRS11} and our
work do not need the full power of this relaxation.
Instead, they are captured by the ``local rounding'' framework of Guruswami and
Sinop \cite{GS12:faster} who show how to implement a local rounding algorithm
based on $t$ levels of the SoS hierarchy, in time $q^{O(t)} \cdot n^{O(1)}$
(where $q$ denotes the alphabet size).
\end{remark}
\vspace{-15 pt}
\paragraph{Our techniques.}
We start by using essentially the same argument for analyzing the SoS hierarchy
as was used by \cite{BarakRS11} (specialized to the case of expanders).
They viewed the SoS solution as giving a joint distribution on each pair of
variables forming a constraint, and proved that for sufficiently expanding
graphs, these distributions can be made close to product distributions, by
conditioning on a small number of variables (which governs the number of levels
required).
Similarly, we consider the conditions under which joint distributions on
$k$-tuples corresponding to constraints can be made close to product
distributions. Since the \cite{BarakRS11} argument shows how to split a joint
distribution into two marginals, we can use it to recursively split a set of size
$k$ into two smaller ones (one can think of all splitting operations as
forming a binary tree with $k$ leaves).

However, our arguments differ in the kind of expansion required to perform the
above splitting operations. 
In the case of the 2-CSP, one splits along the
edges of the constraint graph, and thus we only need the expansion of the
contraint graph (which is part of the assumption). 
However, in the case of $k$-CSPs, we may split a set of size 
$(\ell_1 + \ell_2)$ into disjoint sets of size $\ell_1$ and $\ell_2$.
This requires understanding the expansion of the following family of (weighted)
bipartite graphs arising from the complex $X_{\Ins}$: The vertices in the
graph are sets of variables of size $\ell_1$ and $\ell_2$ that occur in some
constraint, and the weight of an edge $\{\aye_1, \aye_2\}$ for $\aye_1 \cap
\aye_2 = \emptyset$, is proportional to the
probability that a random constraint contains $\aye_1 \sqcup \aye_2$. 
Note that this graph may be weighted even if the $k$-CSP instance $\Ins$ is
unweighted.

We view the above graphs as random walks, which we call ``swap walks'' on the
hyperedges (faces) in the complex $X$. 
While several random walks on high-dimensional expanders have been shown to have
rapid mixing \cite{KaufmanM17, KaufmanO18Walk, DinurK17, LLP17}, we need a
stronger condition. To apply the argument from \cite{BarakRS11}, we not only
need that the second singular value is bounded away from one, but require it to
be an arbitrarily small constant (as a function of $\eps$, $k$ and $q$). 
We show that this is indeed ensured by the condition that
$\aye_1 \cap \aye_2 = \emptyset$, and obtain a bound of $k^{O(k)} \cdot
\gamma$ on the second singular value. 
This bound, which constitutes much of the
technical work in the paper, is obtained by first expressing these walks in
terms of more canonical walks, and then using the beautiful machinery of
harmonic analysis on expanding posets  by Dikstein \etal \cite{DiksteinDFH18} 
to understand their spectra.

The swap walks analyzed above represent natural random walks on simplicial
complexes, and their properties may be of independent interest for other
applications. 
Just as the high-dimensional expanders are viewed as ``derandomized'' versions
of the complete complex (containing all sets of size at most $k$), one can view
the swap walks as derandomized versions of (bipartite) Kneser graphs, which have
vertex sets $\binom{[n]}{\ell_1}$ and $\binom{[n]}{\ell_2}$, and edges $(\aye,
\bee)$ iff $\aye \cap \bee = \emptyset$.
We provide a more detailed and technical overview in \cref{sec:example} after
discussing the relevant preliminaries in \cref{sec:prelims}.
\vspace{-10 pt}
\paragraph{High-dimensional threshold rank.} 
The correlation breaking method in \cite{BarakRS11} can be applied as long
as the graph has low threshold rank \ie the number of singular values above a
threshold $\tau = (\eps/q)^{O(1)}$ is bounded.
Similarly, the analysis described above can be applied, as long as all the swap
walks which arise when splitting the $k$-tuples have bounded threshold rank.
This suggests a notion of high-dimensional threshold rank for hypergraphs
(discussed in \cref{sec:trank}), which
can be defined in terms of the threshold ranks of the relevant swap walks. 
We remark that it is easy to show that dense hypergraphs (with $\Omega(n^k)$
hyperedges) have small-threshold rank according to this notion, and thus it can
be used to recover known algorithms for approximating $k$-CSPs on dense
instances \cite{FK96:focs} (as was true for threshold rank in graphs).
\vspace{-10 pt}
\paragraph{Other related work.}
While we extend the approach taken by \cite{BarakRS11} for 2-CSPs, somewhat
different approaches were considered by Guruswami and Sinop \cite{GuruswamiS11},
and Oveis-Gharan and Trevisan \cite{OGT15}.
The work by Guruswami and Sinop relied on the expansion of the label
extended graph, and used an analysis based on low-dimensional
approximations of the SDP solution.
Oveis-Gharan and Trevisan used low-threshold rank assumptions to obtain a
regularity lemma, which was then used to approximate the CSP.
For the case of $k$-CSPs, the Sherali-Adams hierarchy can be used to solve
instances with bounded treewidth \cite{WJ04:treewidth} and approximately dense
instances \cite{YZ14:dense, MR17}. 
Brandao and Harrow \cite{BrandaoH13} also extended the results by
\cite{BarakRS11} for 2-CSPs to the case of 2-local Hamiltonians. 
We show that their ideas can also be used to prove a similar extension of our
results to $k$-local Hamiltonians on high-dimensional expanders.

In case of high-dimensional expanders, in addition to canonical walks described
here, a ``non-lazy'' version of these walks (moving from $\ess$ to $\tee$ only if
$\ess \neq \tee$) was also considered by Kaufman and
Oppenheim \cite{KaufmanO18Walk}, Anari \etal \cite{ALOV18:log} and Dikstein \etal
\cite{DiksteinDFH18}. 
%
%
The swap walks studied in this paper were also considered independently in a 
very recent work of Dikstein and Dinur \cite{DD19} (under the name "complement walks").

In a recent follow-up work \cite{AJQST19}, the algorithms developed here were 
also used to obtain new unique and list decoding algorithms for direct sum and 
direct product codes, obtained by a ``lifting" a base code $C_0$ via $k$-local 
operations to amplify distance. This work also showed that the hypergraphs obtained 
by considering collections of length-$k$ walks on an expanding graph also satisfy 
(a slight variant of) splittability, and admit similar algorithms.
%


%% file: prelims.tex
\subsection{Linear Algebra}
Recall that for an operator $\Aye: V \to W$ between two finite-dimensional inner product spaces $V$
and $W$, the operator norm can be written as
\[
\norm{\Aye}_{\textup{op}} = \sup_{ f,  g \ne 0} \frac{\ip{\Aye  f}{ g}}{\norm{ f}\norm{ g}} \mper
\]
Also, for such an $\Aye$ the adjoint $\Aye^{\dag}: W \to V$ is defined as the (unique) operator
satisfying $\ip{\Aye f}{g} = \ip{f}{\Aye^{\dag}g}$ for all $f \in V, g \in W$.
For $\Aye: V \to W$, we take $\norm{\Aye}_{\textup{op}} = \sigma_1(\Aye) \geq \sigma_2(\Aye) \geq
\cdots \geq \sigma_r(\Aye) > 0$ to be its singular values in descending order. Note that for $\Aye: V \to
V$, $\sigma_2(\Aye)$ denotes its second largest eigenvalue in absolute value.

\subsection{High-Dimensional Expanders}\label{subsec:hdx}
A high-dimensional expander (HDX) is a particular kind of downward-closed hypergraph (simplicial
complex) satisfying an expansion requirement. 
We elaborate on these properties and define well known natural walks on HDXs below.
\subsubsection{Simplicial Complexes}
\begin{definition}
A {\deffont simplicial complex} $X$ with ground set $[n]$ is a downward-closed
collection of subsets of $[n]$ \ie for all sets $\ess \in X$ and $\tee \subseteq
\ess$, we also have $\tee \in X$. The sets in $X$ are also referred to as
{\deffont faces} of $X$.

We use the notation $X(i)$ to denote the collection of all faces $\ess$ in $X$ with
$\abs{\ess}=i$. When faces are of cardinality at most $d$, we also use the notation $X(\leq d)$ to
denote all the faces of $X$. By convention, we take $X(0) := \set{\varnothing}$.

A simplicial complex $X(\leq d)$ is said to be a {\deffont pure simplicial complex} if every face of
$X$ is contained in some face of size $d$. Note that in a pure simplicial complex $X(\leq d)$, the top slice
$X(d)$ completely determines the complex.
\end{definition}

Note that it is more common to associate a geometric representation to
simplicial complexes, with faces of cardinality $i$ being referred to as faces
of \emph{dimension} $i-1$ (and the collection being denoted by $X(i-1)$ instead
of $X(i)$). However, since we will only be treating these as hypergraphs, we
prefer to index faces by their cardinality, to improve readability of related expressions.

An important simplicial complex is the complete complex.

\begin{definition}[Complete Complex $\Delta_d(n)$]
  We denote by $\Delta_d(n)$ the complete complex with faces of size at most $d$ \ie
  $\Delta_d(n) ~:=~ \set{\ess \subseteq [n] ~|~ \abs{\ess} \leq d}$.
\end{definition}

\subsubsection{Walks and Measures on Simplicial Complexes}\label{subsec:natural_walks}
Let  $C^k$ denote the space of real valued functions on $X(k)$ \ie
$$
C^k ~:=~ \set*{ f~\vert~f \colon X(k) \to \R} \cong \R^{X(k)}.
$$
We describe natural walks on simplicial complexes considered
in~\cite{DinurK17,DiksteinDFH18,KaufmanO18Walk}, as stochastic operators, which map functions in
$C^i$ to $C^{i+1}$ and vice-versa. 

To define the stochastic operators associated with the walks, we first need to
describe a set of probability measures which serve as the stationary measures for these random
walks. For a pure simplicial complex $X(\leq d)$, we define a collection of probability measures
$(\Pi_1, \ldots \Pi_d)$, with $\Pi_i$ giving a distribution on faces in the slice $X(i)$.
\begin{definition}[Probability measures $(\Pi_1, \ldots, \Pi_d)$]
Let $X(\leq d)$ be a pure simplicial complex and let $\Pi_d$ be an arbitrary probability measure on
$X(d)$. We define a coupled array of random variables $(\ess^{(d)}, \ldots,
\ess^{(1)})$ as follows: sample $\ess^{(d)} \sim \Pi_d$ and (recursively) for
each $i \in [d]$, take $\ess^{(i-1)}$ to be a uniformly random subset of
$\ess^{(i)}$, of size $i-1$.

The distributions $\Pi_{d-1}, \ldots, \Pi_1$ are then defined to be the marginal
distributions of the random variables $\ess^{(d-1)}, \ldots, \ess^{(1)}$ as defined above.
\end{definition}
The following is immediate from the definition above.
\begin{proposition}\label{prop:volpres}
    Let $\aye \in X(\ell)$ be an arbitrary face. For all $j \ge 0$, one has
    \[ \sum_{\bee \in X(\ell + j):\atop\bee \supseteq \aye} \Pi_{\ell + j}(\bee)~=~\binom{\ell + j}{j} \cdot \Pi_{\ell}(\aye).\]  
\end{proposition}
For all $k$, we define the inner product of functions $f, g \in C^k$, according to associated measure $\Pi_k$
$$
\ip{ f}{ g} = \Ex{\ess \sim \Pi_k}{ f(\ess)  g(\ess)} = \sum_{\ess \in X(k)}  f(\ess)  g(\ess) \cdot
\Pi_k(\ess) \mper
$$
We now define the up and down operators $\Up_i: C^i \to C^{i+1}$ and $\Dee_{i+1}: C^{i+1} \to C^{i}$
as
\begin{align*}
  [\Up_i g](\ess)
  &~=~ \Ex{ \ess' \in X(i), ~\ess' \subseteq \ess }{g(\ess')}
    ~=~  \frac{1}{i + 1} \cdot \sum_{x \in \ess} g(\ess \backslash \set{x}) \\
  [\Dee_{i + 1} g](\ess)
  &  ~=~ \Ex{\ess' \sim \Pi_{i+1} | \ess' \supset \ess}{g(\ess')}  
      ~=~ \frac{1}{i + 1} \cdot \sum_{x \notin \ess} g(\ess \sqcup \set{x})
    \cdot \frac{\Pi_{i + 1}(\ess \sqcup \set{x})}{\Pi_i(\ess)}
\end{align*}
An important consequence of the above definition is that $\Up_i$ and $\Dee_{i+1}$ are adjoints with
respect to the inner products of $C^i$ and $C^{i+1}$.
\begin{fact}
  $\Up_i = \Dee_{i+1}^{\dag}$, i.e., $\ip{\Up_i f}{ g} =
  \ip{f}{\Dee_{i+1} g}$ for every $f \in C^i$ and $g \in C^{i+1}$.
\end{fact}
Note that the operators can be thought of as defining random walks in a simplicial complex $X(\leq
d)$. The operator $\Up_i$ moves \emph{down} from a face $\ess \in X(i+1)$ to a face $\ess' \in X(i)$,
but lifts a function $g \in C^{i}$ \emph{up} to a function $\Up g \in C^{i+1}$. Similarly, the
operator $\Dee_{i+1}$ can be thought of as defining a random walk which moves \emph{up} from $\ess
\in X(i)$ to $\ess' \in X(i+1)$. It is easy to verify that these walks respectively map the measure $\Pi_{i+1}$
to $\Pi_{i}$, and $\Pi_{i}$ to $\Pi_{i+1}$.

\subsubsection{High-Dimensional Expansion}

We recall the notion of high-dimensional expansion (defined via local spectral expansion) considered
by~\cite{DinurK17}.
We first need a few pieces of notation.

For a complex $X(\leq d)$ and $\ess \in X(i)$ for some $i \in [d]$, we denote by $X_{\ess}$ the
{\deffont link complex}
\[
  X_\ess ~:=~ \set*{ \tee \backslash \ess \mid \ess \subseteq \tee \in X } \mper 
\]
When $\abs{\ess} \leq d-2$, we also associate a natural weighted graph $G(X_{\ess})$ to a link
$X_{\ess}$, with vertex set $X_{\ess}(1)$ and edge-set $X_{\ess}(2)$. The edge-weights are taken to
be proportional to the measure $\Pi_2$ on the complex $X_{\ess}$, which is in turn proportional to
the measure $\Pi_{\abs{\ess}+2}$ on $X$. The graph $G(X_{\ess})$ is referred to as the {\deffont
  skeleton} of $X_{\ess}$.
Dinur and Kaufman~\cite{DinurK17} define high-dimensional expansion in
terms of spectral expansion of the skeletons of the links.
\begin{definition}[$\gamma$-HDX from~\cite{DinurK17}]\label{def:hdx_dk} 
  A simplicial complex $X(\leq d)$ is said to be {\deffont $\gamma$-High Dimensional Expander ($\gamma$-HDX)}
  if for every $0 \le i \le d-2$ and for every $\ess \in X(i)$, the graph $G(X_\ess)$ satisfies
  $\sigma_2(G(X_{\ess})) \leq \gamma$, where $\sigma_2(G(X_{\ess}))$ denotes the  second singular
  value of the (normalized) adjacency matrix of $G(X_{\ess})$.
\end{definition}

\subsection{Constraint Satisfaction Problems (CSPs)}
A $k$-CSP instance $\Ins = (H, \Cc, w)$ with alphabet size $q$
consists of a $k$-uniform hypergraph, a set of constraints
\[ \Cc = \set*{ \Cc_\aye \subseteq [q]^{\aye} : \aye \in H},\]
and a non-negative weight function $w \in \RR_+^H$ on the
constraints, satisfying $\sum_{\aye \in H}w(a) = 1$.

A constraint $\Cc_\aye$ is said to be satisfied by an assignment
$\assn$ if we have $\assn|_{\aye} \in \Cc_\aye$
\ie~the restriction of $\assn$ on $\aye$ is contained in
$\Cc_\aye$. We write, $\SAT_\Ins(\assn)$ for the (weighted fraction
of the constraints) satisfied by the assignment $\assn$ \ie
\[ 
\SAT_\Ins(\assn) 
~=~ \sum_{\aye \in H} w(\aye) \cdot \one[\assn|_\aye \in \Cc_\aye] 
~=~ \Ex{\aye \sim w}{\one[\assn|_\aye \in \Cc_\aye]}
\mper
\]
We denote by $\OPT(\Ins)$ the maximum of $\SAT_\Ins(\assn)$ over all $\assn \in [q]^{V(H)}$.

Any $k$-uniform hypergraph $H$ can be associated with a pure
simplicial complex in a canonical way by just setting $X_{\Ins} = \set*{\bee
  : \exists\ \aye \in H ~\text{and}~ \aye \supseteq \bee }$ -- notice that
$X_{\Ins}(k) = H$. 
We will refer to this complex as the {\deffont constraint complex} of the instance $\Ins$.
The probability distribution $\Pi_{k}$ on
$X_{\Ins}$ will be derived from the weights function $w$ of the
constraint, i.e~
\[ \Pi_{k}(\aye) = w(\aye) \quad\forall \aye \in X_{\Ins}(k) = H.\]

\newcommand{\V}[2]{{v}_{(#1,#2)}}
\newcommand{\W}[2]{{w}_{(#1,#2)}}
\newcommand{\Vempty}{\V{\emptyset}{\emptyset}}

\subsection{Sum-of-Squares Relaxations and $t$-local PSD Ensembles}\label{ssec:sos-relax}
The Sum-of-Squares (SoS) hierarchy gives a sequence of increasingly tight semidefinite programming 
relaxations for several optimization problems, including CSPs. Since we will use relatively 
few facts about the  SoS hierarchy, already developed in the analysis of Barak, Raghavendra and 
Steurer \cite{BarakRS11}, we will adapt their notation of {\deffont $t$-local distributions} to describe 
the relaxations. For a $k$-CSP instance $\Ins = (H,\Cc, w)$ on $n$ variables, we consider the following 
semidefinite  relaxation given by $t$-levels of the SoS hierarchy, with vectors $\V{S}{\alpha}$ 
for all $S \subseteq [n]$ with $\abs{S} \leq t$, and all $\alpha \in [q]^S$. Here, for $\alpha_1 \in
\qary^{S_1}$ and $\alpha_2 \in \qary^{S_2}$, $\alpha_1 \circ \alpha_2 \in \qary^{S_1 \cup S_2}$
denotes the partial assignment obtained by concatenating $\alpha_1$ and $\alpha_2$.
\begin{table}[h]
\hrule
\vline
\begin{minipage}[t]{0.99\linewidth}
\vspace{-5 pt}
{\small
\begin{align*}
\mbox{maximize}\quad ~~
\Ex{\aye \sim w}{\sum_{\alpha \in C_{\aye}} \smallnorm{\V{\aye}{\alpha}}^2}&~=:~ \SDP(\Ins)\\
\mbox{subject to}\quad \quad ~
 \mydot{\V{S_1}{\alpha_1}}{\V{S_2}{\alpha_2}} &~=~ 0 
  & \forall~ \alpha_1|_{S_1 \cap S_2} \neq \alpha_2|_{S_1 \cap S_2}\\
 \mydot{\V{S_1}{\alpha_1}}{\V{S_2}{\alpha_2}} 
  &~=~ \mydot{\V{S_3}{\alpha_3}}{\V{S_4}{\alpha_4}}
  & \forall~ S_1 \cup S_2 = S_3 \cup S_4, 
  ~\alpha_1 \circ \alpha_2 = \alpha_3 \circ \alpha_4 \\
 \sum_{j \in [q]} \smallnorm{\V{\{i\}}{j}}^2 &~=~ 1 &\forall i \in [n]\\
 \smallnorm{\Vempty} &~=~ 1 &
\end{align*}}
\vspace{-14 pt}
\end{minipage}
\hfill\vline
\hrule
\end{table}

For any set $S$ with $|S| \leq t$, the vectors $\V{S}{\alpha}$ induce a
probability distribution $\mu_S$ over $\qary^S$ such that the assignment $\alpha \in \qary^S$ appears with probability $\smallnorm{\V{S}{\alpha}}^2$.
Moreover, these distributions are consistent on intersections \ie for $T
\subseteq S \subseteq [n]$, we have $\mu_{S|T} = \mu_T$, where $\mu_{S|T}$
denotes the restriction of the distribution $\mu_S$ to the set $T$.
We use these distributions to define a collection of random variables $\rv Y_1,
\ldots, \rv Y_n$ taking values in $\qary$, such that for any set $S$ with
$\abs{S} \leq t$, the collection of variables $\inbraces{\rv Y_i}_{i \in S}$
have a joint distribution $\mu_S$. Note that the entire collection $(\rv Y_1,
\ldots, \rv Y_n)$ \emph{may not} have a joint distribution: this property is
only true for sub-collections of size $t$. We will refer to the collection $(\rv
Y_1, \ldots, \rv Y_n)$ as a {\deffont $t$-local ensemble} of random variables.

We also have that that for any $T \subseteq [n]$ with $\abs{T} \leq t-2$, and any 
$\beta \in \qary^T$, we can define a $(t-\abs{T})$-local ensemble $(\rv Y_1', \ldots, \rv Y_n')$ by
``conditioning'' the local distributions on the event $\rv Y_T = \beta$, where $Y_T$ is shorthand
for the collection $\inbraces{\rv Y_i}_{i \in T}$. For any $S$ with $\abs{S} \leq t-\abs{T}$, we define
the distribution of $\rv Y_S'$ as $\mu_S' := \mu_{S \cup T} | \set{\rv Y_T = \beta}$.
Finally, the semidefinite program also ensures that for any such conditioning,
the conditional covariance matrix
\[
\Emm_{(S_1, \alpha_1)(S_2,\alpha_2)} ~=~ \cov\inparen{\one[\rv Y_{S_1}' = \alpha_1],
  \one[\rv Y_{S_2}' = \alpha_2]}
\]
is positive semidefinite, where $\abs{S_1}, \abs{S_2} \leq (t-\abs{T})/2$.
Here, for each pair $S_1, S_2$ the covariance is computed using the joint distribution $\mu_{S_1
  \cup S_2}'$.
The PSD-ness be easily verified by noticing that the above matrix can be written as the Gram matrix of the vectors
\[
  \W{S}{\alpha} ~:=~ \frac{1}{\smallnorm{\V{T}{\beta}}} \cdot \V{T \cup S}{\beta \circ \alpha}
  ~-~ \frac{\smallnorm{\V{T \cup S}{\beta \circ \alpha}}^2}{\smallnorm{\V{T}{\beta}}^3} \cdot \V{T}{\beta}
\]
In this paper, we will only consider $t$-local ensembles such that for every conditioning on a set
of size at most $t-2$, the conditional covariance matrix is PSD. We will refer to these as {\deffont
$t$-local PSD ensembles}.
We will also need a simple corollary of the above definitions.
\begin{fact}\label{fact:set-ensemble}
Let  $(\rv Y_1, \ldots, \rv Y_n)$ be a $t$-local PSD ensemble, and let $X$ be any simplicial complex with $X(1)=[n]$. 
Then, for all $s \leq t/2$, the collection $\inbraces{\rv Y_{\aye}}_{\aye \in X(\le s)}$ is a $(t/s)$-local PSD ensemble,
    where $X(\le s) = \bigcup_{i = 1}^s X(i)$.
\end{fact}
For random variables $\rv Y_S$ in a $t$-local PSD ensemble, we use the notation
$\inbraces{\rv Y_S}$ to denote the distribution of $\rv Y_S$ (which exists when $\abs{S} \leq
t$). We also define $\Var{\rv Y_S}$ as $\sum_{\alpha \in {\qary^S}} \Var{\one\insquare{\rv Y_S = \alpha}}$.
%


%% file: example.tex
We consider a simple example of a specific $k$-CSP, which captures most of the
key ideas in our proof. Let $\Ins$ be an unweighted instance of $4$-XOR on $n$
Boolean variables. Let $H$ be a 4-uniform hypergraph on vertex set $[n]$, with a
hyperedge corresponding to each constraint \ie each $\aye =
\set{i_1, i_2, i_3, i_4} \in H$ corresponds to a constraint in $\Ins$ of the form
\[
x_{i_1} + x_{i_2} + x_{i_3} + x_{i_4} ~=~ b_{\aye} \pmod 2 \mcom
\]
for some $b_{\aye} \in \B$. 
Let $X$ denote the constraint complex for the instance $\Ins$ such that $X(1)=[n]$, $X(4)=H$ and let
$\Pi_1, \ldots, \Pi_4$ be the associated distributions (with $\Pi_4$ being
uniform on $H$). 

\vspace{-10 pt}
\paragraph{Local vs global correlation: the BRS strategy.}
We first recall the strategy used by \cite{BarakRS11}, which also suggests a natural first step for
our proof. Given a 2-CSP instance with an associated graph $G$, and a $t$-local PSD
ensemble $\rv Y_1, \ldots, \rv Y_n$ obtained from the SoS relaxation, they consider if the ``local
correlation" of the ensemble is small across the edges of $G$ (which correspond to constraints) \ie
\[
\Ex{\set{i,j} \sim G}{\norm{\inbraces{\rv Y_{i} \rv Y_{j}} - \inbraces{\rv Y_{i}}\inbraces{\rv Y_{j}} }_1} ~\leq~ \eps \mper
\]
If the local correlation is indeed small, we easily produce an assignment achieving a value
$\SDP - \eps$ in expectation, simply by rounding each variable $x_i$ independently according to the
distribution $\inbraces{\rv Y_i}$. 
On the other hand, if this is not satisfied, they show (as a special case of
their proof) that if $G$ is an expander with second eigenvalue $\lambda \leq c
\cdot (\eps^2/q^2)$, then variables also have a high ``global correlation",
between a typical pair $(i,j) \in [n]^2$. Here, $q$ is the alphabet size and $c$
is a fixed constant.
They use this to show that for $(\rv Y_1', \ldots, \rv Y_n')$ obtained by conditioning on the value
of a randomly chosen $\rv Y_{i_0}$, we have 
\[
\Ex{i}{\Var{\rv Y_i}} - \expop_{i_0, \rv Y_{i_0}}\Ex{i}{\Var{\rv Y_i'}} ~\geq~ \Omega(\eps^2/q^2) \mcom
\] 
where the expectations over $i$ and $i_0$ are both according to the stationary
distribution on the vertices of $G$.
Since the variance is bounded between 0 and 1, this essentially shows that the
local correlation must be at most $\eps$ after conditioning on a set of size
$O(q^2/\eps^2)$ (although the actual argument requires a bit more care and needs
to condition on a somewhat larger set).

\vspace{-10 pt}
\paragraph{Extension to 4-XOR.}
As in \cite{BarakRS11}, we check if the $t$-local PSD ensemble $(\rv Y_1, \ldots, \rv Y_n)$
obtained from the SDP solution satisfies
\[
\Ex{\set{i_1, i_2, i_3, i_4} \in H}{\norm{\inbraces{\rv Y_{i_1} \rv Y_{i_2} \rv Y_{i_3} \rv Y_{i_4}}
    - \inbraces{\rv Y_{i_1}}\inbraces{\rv Y_{i_2}} \inbraces{\rv Y_{i_3}} \inbraces{\rv Y_{i_4}}}_1} ~\leq~ \eps \mper
\]
As before, independently sampling each $x_i$ from $\inbraces{\rv Y_i}$ gives an expected value at least
$\SDP - \eps$ in this case. If the above inequality is not satisfied, an application of triangle
inequality gives
\[
  \Ex{\set{i_1, i_2, i_3, i_4} \in H}{
    \begin{array}{l}
    \norm{\inbraces{\rv Y_{i_1} \rv Y_{i_2} \rv Y_{i_3} \rv Y_{i_4}} - \inbraces{\rv Y_{i_1} \rv
      Y_{i_2}} \inbraces{\rv Y_{i_3} \rv Y_{i_4}}}_1 ~+~ \\[5pt]
   \norm{\inbraces{\rv Y_{i_1} \rv Y_{i_2}} - \inbraces{\rv Y_{i_1}} \inbraces{\rv Y_{i_2}}}_1 ~+~
      \norm{\inbraces{\rv Y_{i_3} \rv Y_{i_4}} - \inbraces{\rv Y_{i_3}} \inbraces{\rv Y_{i_4}}}_1
    \end{array}
  }
  ~>~ \eps \mper
\]
Symmetrizing over all orderings of
$\set{i_1, i_2, i_3, i_4}$, we can write the above as
\[
\eps_2 + 2 \cdot \eps_1 ~>~ \eps \mcom
\]
which gives $\max\inbraces{\eps_1, \eps_2} \geq \eps/3$. Here,
\begin{align*}
  \eps_1 &~:=~ \Ex{\set{i_1, i_2} \sim \Pi_2}{\norm{\inbraces{\rv Y_{i_1} \rv Y_{i_2}} -
           \inbraces{\rv Y_{i_1}} \inbraces{\rv Y_{i_2}}}_1} \mcom \quad \text{and} \\
  \eps_2 &~:=~ \Ex{\set{i_1, i_2, i_3, i_4} \sim \Pi_4}{
    \norm{\inbraces{\rv Y_{i_1} \rv Y_{i_2} \rv Y_{i_3} \rv Y_{i_4}} - \inbraces{\rv Y_{i_1} \rv
           Y_{i_2}} \inbraces{\rv Y_{i_3} \rv Y_{i_4}}}_1} \\ &~=~ \Ex{\set{i_1, i_2, i_3, i_4} \sim \Pi_4}{
    \norm{\inbraces{\rv Y_{\set{i_1, i_2}} \rv Y_{\set{i_3, i_4}}} - \inbraces{\rv Y_{\set{i_1,
                                                                i_2}}} \inbraces{\rv Y_{\set{i_3,
                                                                i_4}} }}_1} \mper
\end{align*}
As before, $\eps_1$ measures the local correlation across edges of a weighted graph $G_1$ with
vertex set $X(1) = [n]$ and edge-weights given by $\Pi_2$. Also, $\eps_2$ measures the analogous
quantity for a graph $G_2$ with vertex set $X(2)$ (pairs of variables occurring in constraints) and
edge-weights given by $\Pi_4$.

Recall that the result from \cite{BarakRS11} can be applied to \emph{any}
graph $G$ over variables in a 2-local PSD ensemble, as long as the
$\sigma_2(G)$ is small. 
Since $\inbraces{\rv Y_i}_{i \in [n]}$ and
$\inbraces{\rv Y_{\ess}}_{\ess \in X(2)}$ are both $(t/2)$-local PSD ensembles (by
\cref{fact:set-ensemble}), we will apply the result to the graph $G_1$ on the
first ensemble and $G_2$ on the second ensemble.
%
%
We consider the potential
\[
\Phi(\rv Y_1, \ldots, \rv Y_n) ~:=~ \Ex{i \sim \Pi_1}{\Var{\rv Y_i}} ~+~ \Ex{\ess \sim
  \Pi_2}{\Var{\rv Y_{\ess}}} \mper
\]
Since local correlation is large along at least one of the graphs $G_1$ and $G_2$,
using the above arguments (and the non-decreasing nature of variance under conditioning) it is easy
to show that in expectation over the choice of  $\set{i_0, j_0} \sim \Pi_2$ and $\beta \in \qary^2$
chosen from $\inbraces{\rv Y_{\set{i_0, j_0}}}$, the conditional ensemble $(\rv Y_1', \ldots, \rv Y_n')$
satisfies
\[
\Phi(\rv Y_1, \ldots, \rv Y_n) - \Ex{i_0, j_0, \beta}{\Phi(\rv Y_1', \ldots, \rv Y_n')} ~=~
\Omega(\eps^2) \mcom
\]
\emph{provided $G_1$ and $G_2$ satisfy $\sigma_2(G_1), \sigma_2(G_2) ~\leq~ c \cdot \eps^2$} for an
appropriate constant $c$.

The bound on the eigenvalue of $G_1$ follows simply from the fact that 
it is the skeleton of $X$, which is a $\gamma$-HDX.
Obtaining bounds on the eigenvalues of $G_2$ and similar higher-order graphs,
constitutes much of the technical part of this paper.
Note that for a random sparse instance of \maxfourxor, the graph $G_2$ will be a
matching with high probability (since $\{i_1, i_2\}$ in a constraint will only
be connected to $\{i_3, i_4\}$ in the same constraint). However, we show that in
case of a $\gamma$-HDX, this graph has second eigenvalue $O(\gamma)$.
We analyze these graphs in terms of modified high-dimensional random walks, which we
call ``swap walks''.

We remark that our potential and choice of a ``seed set'' of variables to
condition on, is slightly different from \cite{BarakRS11}.
To decrease the potential function above, we need that for each level $X(i)$ ($i = 1,2$ in the
example above) the seed set must contain sufficiently many independent samples
from $X(i)$ sampled according to $\Pi_i$.
This can be ensured by drawing independent samples from the top level $X(k)$
(though $X(2)$ suffices in the above example). In contrast, the seed set in
\cite{BarakRS11} consists of random samples from $\Pi_1$.
\vspace{-10 pt}
\paragraph{Analyzing Swap Walks.}
The graph $G_2$ defined above can be thought of as a random walk on $X(2)$, which starts
at a face $\ess \in X(2)$, moves up to a face (constraint) $\ess' \in X(4)$ containing it, and then
descends to a face $\tee \in X(2)$ such that $\tee \subset \ess'$ and $\ess \cap \tee
= \emptyset$ \ie the walk ``swaps out'' the elements in $\ess$ for other
elements in $\ess'$.
Several walks considered on simplicial complexes allow for the possibility of a non-trivial 
intersection, and hence have second eigenvalue lower bounded by a constant. On the other hand, 
swap walks completely avoid any laziness and thus turn out to have eigenvalues which can be made 
arbitrarily small. 
To understand the eigenvalues for this walk, we will express it in terms of other
canonical walks defined on simplicial complexes.

Recall that the up and down operators can be used to define random walks on simplicial
complexes. The up operator $\Up_i : C^i \to C^{i+1}$ defines a walk that moves \emph{down} from a
face $\ess \in X(i+1)$ to a random face $\tee \in X(i), \tee \subset \ess$ (the operator thus
``lifts'' a function in $C^i$ to a function in $C^{i+1}$). Similarly, the down operator $\Dee_i: C^i
\to C^{i-1}$ moves \emph{up} from a face $\ess \in X(i-1)$ to $\tee \in X(i), \tee \supset s$, with
probability $\Pi_{i}(\tee)/(i \cdot \Pi_{i-1}(\ess))$.
These can be used to define a canonical random walk
\[
\nwalk{2}{u} ~:=~ \Dee_3 \cdots \Dee_{u+2} \Up_{u+1} \cdots \Up_2 \mcom \qquad
\nwalk{2}{u}: C^2 \to C^2 \mcom
\]
which moves from up for $u$ steps $\ess \in X(2)$ to $\ess' \in X(u+2)$, and then descends back to
$\tee \in X(2)$. Such walks were analyzed optimally by Dinur and Kaufman \cite{DinurK17}, who proved that
$\lambda_2\inparen{\nwalk{2}{u}} =  2/(u+2) \pm O_u(\gamma)$ when $X$ is a $\gamma$-HDX. 
Thus, while this walk gives an expanding graph with vertex set $X(2)$, the second eigenvalue cannot
be made arbitrarily small for a fixed $u$ (recall that we are interested in showing that
$\sigma_2(G_2) \leq c\cdot \eps^2$).
However, note that we are only interested in $\nwalk{2}{2}$ \emph{conditioned on
  the event} that the two elements from $\ess$ are ``swapped out'' with new
elements in the final set $\tee$ \ie $\ess \cap \tee = \emptyset$. We define
\[
  \hrswap{2}{2}{u}{j}(\ess, \tee) ~:=~
  \begin{cases}
    \frac{\binom{u+2}{2}}{\binom{u}{j} \cdot \binom{2}{2-j}} \cdot \nwalk{2}{u} & \text{if}~ \abs{\tee \setminus \ess} = j
    \\[3 pt]
    0 & \text{otherwise}
  \end{cases}
  \mcom
\]
where the normalization is to ensure stochasticity of the matrix. In this notation, the graph $G_2$
corresponds to the random-walk matrix $\hrswap{2}{2}{2}{2}$. We show that while
$\sigma_2(\nwalk{2}{2}) \approx 1/2$, we have that $\sigma_2(\hrswap{2}{2}{2}{2}) = O(\gamma)$.
We first write the canonical walks in terms of the swap walks. Note that
\[
\nwalk{2}{2} ~=~ \frac16 \cdot \Ide ~+~ \frac23 \cdot \hrswap{2}{2}{2}{1} ~+~ \frac16 \cdot
\hrswap{2}{2}{2}{2} \mcom
\]
since the ``descent'' step from $\ess' \in X(4)$ containing $\ess \in X(2)$, produces a $\tee \in
X(2)$ which ``swaps out'' $0, 1$ and $2$ elements with probabilities
$\nfrac{1}{6}, \nfrac{2}{3}$ and $\nfrac{1}{6}$
respectively. Similarly, 
\[
\nwalk{2}{1} ~=~ \frac13 \cdot \Ide ~+~ \frac23 \cdot \hrswap{2}{2}{1}{1} \mper
\]
Finally, we use the fact (proved in \cref{sec:walks}) that while the canonical walks do depend on
the ``height'' $u$ (\ie $\nwalk{2}{u} \neq \nwalk{2}{u'}$) the swap walks (for a fixed number of swaps $j$)
are independent of the height to which they ascend! In particular, we have
\[
\hrswap{2}{2}{2}{1} ~=~ \hrswap{2}{2}{1}{1} \mper
\]
Using these, we can derive an expression for the swap walk $\hrswap{2}{2}{2}{2}$ as
\[
  \hrswap{2}{2}{2}{2}
  ~=~ \Ide ~+~ 6 \cdot \nwalk{2}{2} - 6 \cdot \nwalk{2}{1}
  ~=~ \Ide ~+~ 6\cdot \inparen{\Dee_3\Dee_4\Up_3\Up_2 - \Dee_3\Up_2}
\]
%
%
To understand the spectrum of operators such as the ones given by the above
expression, we use the beautiful machinery for harmonic analysis over HDXs (and
more generally over expanding posets) developed by Dikstein \etal
\cite{DiksteinDFH18}. They show how to decompose the spaces $C^k$ into
approximate eigenfunctions for operators of the form $\Dee \Up$. 
Using these decompositions and the properties of expanding posets, we can show
that distinct eigenvalues of the above operator are approximately the same (up
to $O(\gamma)$ errors) when analyzing the walks on the complete complex. 
%
%
Finally, we use the fact that swap walks in a complete complex correspond to Kneser graphs (for
which the eigenvectors and eigenvalues are well-known) to show that $\lambda_2(\hrswap{2}{2}{2}{2}) =
  O(\gamma)$.

\paragraph{Splittable CSPs and high-dimensional threshold rank.} We note that the ideas used above
can be generalized (at least) in two ways. 
In the analysis of distance from product distribution for
a 4-tuple of random variables forming a contraint, we split it in 2-tuples. In general, we can
choose to split tuples in a $k$-CSP instance along \emph{any} binary tree $\mathcal{T}$ with $k$ leaves,
with each parent node corresponding to a swap walk between tuples forming its children. 
%
Finally, the analysis from \cite{BarakRS11} also works if the each of the swap walks in some $\cal
T$ have a bounded number (say $r$) of eigenvalues above some threshold $\tau$, which provide a
notion of high-dimensional threshold rank for hypergraphs.
We refer to such an instance as a
{\deffont $(\tree, \tau, r)$-splittable}.

The arguments sketched above show that
high-dimensional expanders are $(\tree, O(\gamma), 1)$-splittable for all $\tree$.
Since the knowledge of $\tree$ is only required in our analysis and not
in the algorithm, we say that $\textrm{rank}_{\tau}(\Ins) \leq r$ (or that $\Ins$ is
$(\tau,r)$-splittable) if $\Ins$ is $(\tree, \tau, r)$-splittable for any $\tree$. 
We defer the precise statement of results for $(\tau,r)$-splittable instances
to \cref{sec:trank}. 


%% file: walks.tex
It is important to note that both $\Up_i$ and $\Dee_{i + 1}$ can be
thought of as row-stochastic matrices, i.e.~we can think of them as
the probability matrices describing the movement of a walk from $X(i +
1)$ to $X(i)$; and from $X(i)$ to $X(i + 1)$ respectively. More
concretely, we will think
\[ [\Dee_{i+1}^\top e_\ess](\tee) = \Pr{\text{the walk moves \underline{up} from $\ess \in X(i)$ to $\tee \in X(i + 1)$}}\]
and similarly
\[ [\Up_i^\top e_\tee](\ess) = \Pr{\text{the walk moves \underline{down} from $\tee \in X(i + 1)$ to $\ess \in X(i)$}}.\]

By referring to the definition of the up and down operators in~\cref{sec:prelims}, it is easy to verify that
\[ [\Dee_{i + 1}^\top e_\ess](\tee) = \one[\tee \supseteq \ess] \cdot \frac{1}{i + 1} \frac{\Pi_{i + 1}(\tee)}{\Pi_i(\ess)} \tand
   [\Up_i^\top e_\tee](\ess) = \one[\ess \subseteq \tee] \cdot \frac{1}{i + 1}.\]
It is easy to see that our notion of random walk respects the probability distributions $\Pi_j$, i.e.~we have
\[ \Up_i^\top \Pi_{i + 1} = \Pi_i \tand \Dee_{i + 1}^\top \Pi_i = \Pi_{i+1},\]
\ie randomly moving up from a sample of $\Pi_j$ gives a sample of $\Pi_{j + 1}$ and similarly, moving down from
a sample of $\Pi_{j + 1}$ results in a sample of $\Pi_j$.

Instead of going up and down by one dimension, one can try going up or
down by multiple dimensions since
$\parens*{\Dee_{i + 1} \cdots \Dee_{i + \ell}}$ and $\parens*{\Up_{i + \ell} \cdots \Up_i}$
are still row-stochastic matrices. Further, the corresponding
probability vectors still have intuitive explanations in terms of the
distributions $\Pi_j$. For a face $\ess \in X(k)$, we introduce the
notation
\[ p_\ess^{(u)} = \parens*{\Dee_{k + 1} \cdots \Dee_{k + u}}^\top e_\ess\]
where we take $p_\ess^{(0)} = e_\ess$. This notation will be used to denote the 
probability distribution of the \underline{up-walk} starting from $\ess \in X(k)$ and ending in a random face $\tee \in X(k+u)$ satisfying $\tee \supseteq \ess$.

Note that the following Lemma together with \cref{prop:volpres} implies that  $p_\ess^{(u)}$ is indeed a probability distribution.
\begin{proposition}\label{prop:probform}
    For $\ess \in X(k)$ and $\aye \in X(k + u)$ one has,
    \[ p^{(u)}_\ess(\aye) = \one[\aye \supseteq \ess] \cdot \frac{1}{\binom{k+u}{u}} \cdot \frac{\Pi_{k+u}(\aye)}{\Pi_k(\ess)}.\]
\end{proposition}
\begin{proof}
    Notice that for $u = 0$, the statement holds trivially. We assume that there exists some $u \ge 0$
    that satisfies
    \[ p_\ess^{(u)}(\aye) = \one[\aye \supseteq \ess] \cdot \frac{1}{\binom{k+u}{u}}\cdot \frac{\Pi_{k + u}(\aye)}{\Pi_k(\ess)} \]
    for all $\aye \in X(k + u)$.

    For $\bee \in X(k + (u + 1))$ one has,
    \[ p_\ess^{(u + 1)}(\bee) = [\Dee_{k + u + 1}^\top p_\ess^{(u)}](\bee) = \frac{1}{k + u + 2} 
    \cdot \sum_{x \in \bee} \frac{\Pi_{k + u + 1}(\bee)}{\Pi_{k + u}(\bee \backslash \set*{x})} \cdot p_\ess^{(u)}(\bee \backslash \set*{x}). \]
    Plugging in the induction assumption, this implies
    \begin{eqnarray*}
        p_\ess^{(u + 1)}(\bee) & = & \frac{1}{(k + u + 1)} \cdot \sum_{x \in \bee} \frac{\Pi_{k + u + 1}(\bee)}{\Pi_{k + u}(\bee \backslash \set{x})} \cdot \parens*{ \one[(\bee \backslash \set{x)}) \supseteq \ess] \cdot \frac{1}{\binom{k+u}{u}} \cdot \frac{\Pi_{k + u}(\bee \backslash \set{x})}{\Pi_k(\ess)}},\\
        & = & \frac{1}{(k + u + 1)} \cdot \frac{1}{\binom{k+u}{u}} \cdot \sum_{x \in \bee} \one[ \bee \backslash \set{x} \supseteq \ess] \cdot \frac{\Pi_{k + u + 1}(\bee)}{\Pi_k(\ess)}.
    \end{eqnarray*}
    First, note that the up-walk only hits the faces that contain
    $\ess$, otherwise $\one[ \bee \backslash \set{x} \supseteq \ess] =
    0$.

    Suppose therefore $\bee \in X(k + u + 1)$ satisfies $\bee
    \supseteq \ess$. Since there are precisely $(u + 1)$
    indices whose deletion still preserves the containment of
    $\ess$, we can write
    \begin{eqnarray*}
        p_\ess^{(u + 1)}(\bee) & = & \one[\bee \supseteq \ess] \cdot \frac{u+1}{k + u + 1} \cdot \frac{1}{\binom{k+u}{u}} \frac{\Pi_{k + u + 1}(\bee)}{\Pi_k(\ess)},\\
        & = & \one[\bee \supseteq \ess] \cdot \frac{1}{\binom{k+u+1}{u+1}} \cdot \frac{\Pi_{k + u + 1}(\bee)}{\Pi_k(\ess)}.
    \end{eqnarray*}
    Thus, proving the proposition.
\end{proof}

Similarly, we introduce the notation 
$q_{\aye}^{(u)}$, as
\[ q_{\aye}^{(u)}(\ess) = \parens*{ \Up_{k + u - 1} \cdots \Up_k }^\top e_\ess, \]
i.e.~for the probability distribution of the \underline{down-walk} starting from $\aye \in X(k+u)$ and ending
in a random face of $X(k)$ contained in $\aye$.
The following can be verified using Proposition \ref{prop:probform}, and the fact that
$\parens*{ \Up_{k+u - 1} \cdots \Up_k}^\dagger = \Dee_{k+u} \cdots \Dee_{k+1}$.
\begin{corollary}
    Let $X(\le d)$ be a simplicial complex, and $k, u \ge 0$ be parameters satisfying $k + u \le d$. For $\aye \in X(k+u)$ and $\ess \in X(k)$, one has
    \[ q_{\aye}^{(u)}(\ess) = \frac{1}{\binom{k+u}{u}} \cdot \one[ \ess \subseteq \aye ].\]
\end{corollary}

In the remainder of this section, we will try to construct more
intricate walks on $X$ from $X(k)$ to $X(l)$.

\subsection{The Canonical and the Swap Walks on a Simplicial Complex}
\begin{definition}[Canonical and Swap $u$-Walks]
    Let $d \ge 0$, $X(\le d)$ be a simplicial complex, and $k, l, u \ge 0$ be parameters satisfying
    $l \le k$, $u \le l$ and $d \ge k + u$; where the constraints on these parameters are to ensure well-definedness.
    We will define the following random walks,
    \begin{itemize}
        \item \textbf{canonical $u$-walk from $X(k)$ to $X(l)$}. Let $\rnwalk{k}{l}{u}$ be the (row-stochastic) Markov operator that represents
            the following random walk: Starting from a face $\ess \in X(k)$,
            \begin{itemize}
                \item {\em (random ascent/up-walk)} randomly move up a face $\ess'' \in X(k + u)$ that contains $\ess$, where $\ess''$
                    is picked with probability
                    \[ p_\ess^{(u)}(\ess'') = [\parens*{ \Dee_{k + 1} \cdots \Dee_{k + u} }^\top e_\ess](\ess''). \]
                \item {\em  (random descent/down-walk)} go to a face $\ess' \in X(l)$ picked uniformly among all the 
                    $l$-dimensional faces that are contained in $\ess''$, \ie~the set $\ess'$ is picked with probability
                    \[ q_{\ess''}(\ess') = \one[\ess' \subseteq \ess''] \cdot \frac{1}{\binom{k + u}{l}} = [\parens*{\Up_{k + u  -1} \cdots \Up_l}^\top e_{\ess''}](\ess').\]
            \end{itemize}
            The operator $\rnwalk{k}{l}{u} \colon C^l \to C^k$ satisfies the following equation,
            \[ \rnwalk{k}{l}{u} = \Dee_{k + 1} \cdots \Dee_{k + u} \Up_{k + u - 1} \cdot \Up_{k} \cdots \Up_{l}.\]
            Notice that we have $\rnwalk{k}{k}{0} = \Ide$, and $\rnwalk{k}{l}{0} = (\Up_{k-1}\ldots \Up_{l})$ for $l < k$.
        \item \textbf{swapping walk from $X(k)$ to $X(l)$}. Let $\rfswap{k}{l}$ be the Markov operator that represents the
        following random walk: Starting from a face $\ess \in X(k)$,
        \begin{itemize}
            \item { \em (random ascent/up-walk)} randomly move up to a face $\ess'' \in X(k + l)$ that contains $\ess$, where 
                as before $\ess''$ is picked with probability
                \[ p_\ess^{(l)}(\ess'') = [\parens*{ \Dee_{k + 1} \cdots \Dee_{k + l + 1} }^\top e_\ess](\ess''). \]
                \item {\em (deterministic descent)} \underline{deterministically} go to $\ess'=\ess'' \backslash \ess \in X(l)$. 
        \end{itemize}
\end{itemize}
\end{definition}

For our applications, we will need to show that the walk $\rfswap{k}{l}$ has good spectral expansion whenever $X$ is 
a $d$-dimensional $\gamma$-expander, for $\gamma$ sufficiently small. To show this, we will relate the swapping walk
operator $\rfswap{k}{l}$ on $X$ to the canonical random walk operators $\rnwalk{k}{l}{u}$ (q.v.~\cref{lemma:formula_swapping_show}).

By the machinery of expanding posets
(q.v.~\cref{sec:asd}) it is possible to argue that the spectral
expansion of the random walk operator $\rnwalk{k}{l}{u}$ on a high
dimensional expander will be close to that of the complete complex. This will allow us to conclude using the relation
between the swapping walks and the canonical walks (q.v.~\cref{lemma:formula_swapping_show}) that the spectral expansion of the swapping walk on $X$, will be 
comparable with the spectral expansion of the swap walk on the complete complex. More precisely, we will show
\begin{lemma}[\cref{cor:complete_rec_kneser}]\label{lemma:formula_swapping_show}
    For any $d, k, l \ge 0$, and the complete simplicial simplicial complex $X(\le d)$, one has the following:
    If $k \ge l \ge 0$ and $d \ge k + l$, we have
    \[ \sigma_2 (\rfswap{k}{l} ) = O_{k,l}\parens*{\frac{1}{n}}. \]
\end{lemma}

Using these two, and the expanding poset machinery, we will conclude
\begin{theorem}[\cref{thm:swap-eig-bd} simplified]\label{thm:swap-eig-bd-show}
    Let $X$ be a $d$-dimensional $\gamma$ expander. If $k \ge l \ge 0$ satisfy $d \ge l + k$ we have,
    \[ \sigma_2(\rfswap{k}{l}) = O_{k,l}(\gamma)\]
    where $\rfswap{k}{l}$ is the swapping walk on $X$ from $X(k)$ to $X(l)$.
\end{theorem}

To prove \cref{thm:swap-eig-bd-show} we will need to
define an intermediate random walk that we will call the $j$-swapping
$u$-walk from $X(k)$ to $X(l)$:
\begin{definition}[$j$-swapping $u$-walk from $X(k)$ to $X(l)$] Given $d, u, j, k, l \ge 0$ satisfying $l \le k$, $j \le u$, $u \le l$, and $d \ge k + u$.
Let $\hrswap{k}{l}{u}{j}$ be the Markov operator that represents the following random walk from $X(k)$ to $X(l)$ on a
$d$-dimensional simplicial complex $X$:
Starting from $\ess \in X(k)$
\begin{itemize}
    \item {\em (random ascent/up-walk) } randomly move up to a face $\ess'' \in X(k + u)$ that contains $\ess$, where $\ess''$ is
        picked with probability
        \[ p_\ess^{(u)}(\ess'') = [\parens*{ \Dee_{k + 1} \cdots \Dee_{k + u} }^\top e_\ess ](\ess'').\] 
    \item {\em (conditioned descent) } go to a face $\ess' \in X(l)$ sampled uniformly among all the subsets
        of $\ess'' \in X(k+u)$ that have intersection $j$ with $\ess'' \backslash \ess$, i.e.~$|\ess' \cap (\ess'' \backslash \ess)| = j$. 
\end{itemize}
Notice that $\rfswap{k}{l} = \hrswap{k}{l}{l}{l}$ for any $k$ and $\Ide =
\hrswap{k}{k}{u}{0}$ for any $u$.
\end{definition}
\begin{remark}
We will prove that the parameter $u$ does not effect the swapping walk
$\hrswap{k}{l}{u}{j}$ so long as $u \ge j$, i.e. for all $u, u' \ge j$ we have
$\hrswap{k}{l}{u'}{j} = \hrswap{k}{l}{u}{j}$.
Thus, we will often write $\rswap{k}{l}{j}$ for $\hrswap{k}{l}{j}{j}$.
\end{remark}

\subsection{Swap Walks are Height Independent}

Recall that the swap walk $\hrswap{k}{l}{u}{j}$ is the conditional
walk defined in terms of $\rnwalk{k}{l}{u}$ where $\ess \in X(k)$ is
connected to $\tee \in X(l)$ only if $\Abs*{ \tee \setminus \ess } =
j$. The parameter $u$ is called the \underline{height} of the walk, namely the
number of times it moves up. Since up and down operators have second
singular value bounded away from $1$, the second singular value of
$\rnwalk{k}{l}{u}$ shrinks as $u$ increases.  In other words, the
operator $\rnwalk{k}{l}{u}$ depends on the height $u$. Surprisingly,
the walk $\hrswap{k}{l}{u}{j}$ which is defined in terms of
$\rnwalk{k}{l}{u}$ does not depend on the particular choice of $u$ as
long as it is well defined. More precisely, we have the following
result.

\begin{lemma}\label{lemma:heigh_indep}
  If $X$ is a $d$-dimensional simplicial complex, $0 \le l \le k$,
    and $u, u' \in [j, d-k]$, then
  \[ \hrswap{k}{l}{u}{j} = \hrswap{k}{l}{u'}{j}.\]
\end{lemma}

In order to obtain~\cref{lemma:heigh_indep}, we will need a simple
proposition:

\begin{proposition}\label{prop:swapformul_rect}
    Let $\ess \in X(k)$, $\ess' \subseteq \ess$ and $|t'|= j$. Suppose
    $\ess' \sqcup \tee' \in X(l)$. Then, we have
    \[ \hrswap{k}{l}{u}{j}(\ess, \ess' \sqcup \tee') = \frac{1}{\binom{k}{l - j} \cdot \binom{u}{j}} \cdot \sum_{\aye \in X(k + u) : \atop \aye \supseteq (\ess \sqcup \tee')} p_\ess^{(u)}(\aye).\]
\end{proposition}
\begin{proof}
    The only way of picking $\ess' \sqcup \tee'$ at the descent step
    is picking some $\aye \in X(k + u)$ that contains $\ess' \sqcup
    \tee'$ in the ascent step. The probability of this happening is precisely,
    \[ p_1 = \sum_{\aye \in X(k + u) :\atop \aye \supseteq (\ess \sqcup \tee')} p^{(u)}_\ess(\aye).\]
    Suppose we are at a set $\aye = \ess \sqcup \tee$, such that $\tee
    \supseteq \tee'$ and $\ess \cap \tee = \varnothing$. Now, the
    probability of the descent step ending at $\ess' \sqcup \tee'$ is the probability
    of a randomly sampled $(l - j)$-elemented subset of $\ess$
    being $\ess'$ and the probability of a randomly sampled
    $j$-elemented subset of $\tee$ being $\tee'$. It can be verified
    that this probability is
    \[ p_2 = \frac{1}{\binom{k}{l- j} \cdot \binom{u}{j}}.\]    
    By law of total probability we establish that
    \[ \hrswap{k}{l}{u}{j}(\ess, \ess' \sqcup \tee') = p_1 \cdot p_2  = \frac{1}{\binom{k}{l - j} \cdot \binom{u}{j}} \cdot \sum_{\aye \in X(k + u) :\atop \aye \sqcup (\ess \sqcup \tee')} p_\ess^{(u)}(\aye).\]
\end{proof}

\begin{lemma}[Height Independence]\label{theo:heigh_indep}
    Let $u \in [j, d-k]$. For any $\ess \in X(k)$,  $\ess' \subseteq \ess$ and $\tee' \in X(j)$ satisfying $\ess' \sqcup \tee' \in X(l)$ we have the following,
    \[ \hrswap{k}{l}{u}{j}(\ess, \ess' \sqcup \tee') = \frac{1}{\binom{k}{l - j} \binom{k +j}{j}} \cdot \frac{\Pi_{k + j}(\ess \sqcup \tee')}{\Pi_k(\ess)}. \]
    In particular, the choice of $u \in [j, d -k]$ does not affect the swap walk.
\end{lemma}
\begin{proof}   
    We have,
    \begin{eqnarray*}
        \sum_{\aye \in X(k + u) : \aye \supseteq \ess \sqcup \tee'} p_\ess^{(k + u)}(\aye) &  =  &\frac{1}{\binom{k + u}{u}} \cdot \frac{1}{\Pi_k(\ess)} \cdot \sum_{\aye \in X(k + u) : \aye \supseteq \ess \sqcup \tee} \Pi_{k + u}(\aye),\\
        & = & \frac{\binom{k + u}{u - j}}{\binom{k + u}{u}}\cdot \frac{\Pi_{k + j}(\ess \sqcup \tee')}{\Pi_k(\ess)}
    \end{eqnarray*}
    where the first equality is due to~\cref{prop:probform} and the
    second is due to~\cref{prop:volpres} together with the observation
    that $\ess \sqcup \tee' \in X(k + j)$.

    Thus, by~\cref{prop:swapformul_rect} we get,
    \[ \hrswap{k}{l}{u}{j}(\ess, \tee) = \frac{1}{\binom{u}{j} \cdot \binom{k}{l - j}}\frac{\binom{k + u}{u - j}}{\binom{k + u }{u}} \cdot \frac{\Pi_{k + j}(\ess \sqcup \tee')}{\Pi_k(\ess)}.\]
    We complete the proof by noting that,
    \[ \frac{\binom{k + u }{u - j}}{\binom{k + u }{u}} =  \frac{\binom{u}{j}}{\binom{k+j}{j}},\]
    and thus
    \[ \hrswap{k}{l}{u}{j}(\ess, \tee) = \frac{1}{\binom{k}{l - j} \cdot \binom{k+j}{j}} \cdot \frac{\Pi_{k + j}(\ess \sqcup \tee')}{\Pi_k(\ess)}\]
    which proves the formula.
\end{proof}
Since the choice of $u$ does not affect the formula, we
obtain~\cref{lemma:heigh_indep}.

\subsection{Canonical Walks in Terms of the Swap Walks}

We show that the canonical walks are given by an average of swap walks with respect to the hypergeometric distribution.

\begin{lemma}\label{lemma:formula_normal_rect}
  Let $u, l, k, d \ge 0$ be given satisfying $l \le k$ and $u \le l$. Then, we have the following formula
    for the canonical $u$-walk on any $X(\le d)$ satisfying $d \ge k + u$
    \[ \rnwalk{k}{l}{u} = \sum_{j = 0}^{u} \frac{\binom{u}{j} \binom{k}{l - j}}{\binom{k + u}{l}} 
    \cdot \rswap{k}{l}{j}.\]
\end{lemma}

\begin{proof}
    Suppose the canonical $u$-walk starting
    from $\ess \in X(k)$ picks $\ess'' \in X(k + u)$ in the second step. Write $\Ee_j(\ess'')$ for the
    event that the random face $\ess'$ the canonical $u$-walk picks in the descent step satisfies
    \[ \Abs*{ \ess' \setminus \ess } = j. \] 
    By elementary combinatorics,
    \[ \Pr{\ess' \subseteq \ess''}{\Ee_j(\ess'') \mid \ess'' } = \frac{\binom{u}{j} \binom{k}{l - j} }{\binom{k + u}{l}} \]
    where the draw of the probability is over the subsets $\ess' \in X(l)$ of $\ess''$.
    Further, let $\tee'_j$ be the random variable that stands for the face picked in the descent step
    of the $j$-swapping $u$-walk from $X(k)$ to $X(l)$.

    By the definition of the $j$-swapping walk from $X(k)$ to $X(l)$, conditioning that the ascent step picks the same 
    $\ess'' \in X(k + u)$ we have
    \begin{equation}
        \Pr{ \tee'_j = \tee \mid \ess'' } = \Pr{\ess' = \tee \mid \ess'' \textrm{ and }\Ee_j(\ess'')}.\label{eq:swap_transfer}
    \end{equation}
    Now, by the law of total probability we have
    \begin{eqnarray*}
        \rnwalk{k}{l}{u}(\ess, \tee) = \Pr{\rand S' = \tee} & = & \sum_{j = 0}^{u} \sum_{\ess'' \in X(k + u)} \Pr{\ess''} \cdot \Pr{\Ee_j(\ess'') \mid \ess''} \cdot \Pr{\ess' = \tee \mid \rand \ess'' \textrm{ and }\Ee_j(\ess'')},\\
        & = & \sum_{j = 0}^{u} \frac{ \binom{u}{j} \binom{k}{l - j} }{\binom{k +u }{l}} \cdot
        \Ex{\ess'' \supseteq \ess}{\Pr{\ess' = \tee \mid \rand \ess'' \textrm{ and } \Ee_j(\ess'') }},\\
        & = & \sum_{j= 0}^{u} \frac{ \binom{u}{j} \binom{k + u}{l  - j}}{\binom{k  + u}{l}} \cdot
        \Ex{\ess'' \supseteq \ess}{ \Pr{\tee'_j = \tee \mid \ess''} }
    \end{eqnarray*}
    where we used Equation \eqref{eq:swap_transfer} to get the last equality. Another application of the
    law of total probability gives us 
    \[ \Ex{\ess'' \supseteq \ess}{\Pr{\tee'_j = \tee \mid \ess''}} = \Pr{\tee'_j = \tee}.\]
    This allows us to write, 
    \begin{eqnarray*}
        \rnwalk{k}{l}{u}(\ess, \tee) & = & \sum_{j = 0}^{u} \frac{ \binom{u}{j} \binom{k}{l - j}}{\binom{k+ u}{l} } \cdot \Pr{\tee'_j = \tee},\\
        & = &\sum_{j = 0}^{u} \frac{ \binom{u}{j} \binom{k}{l - j}}{\binom{k + u}{l} } \cdot \hrswap{k}{l}{u}{j}(\ess, \tee), 
    \end{eqnarray*}
    The statement follows using height independence, i.e.~\cref{lemma:heigh_indep}
\end{proof}

\subsection{Inversion: Swap Walks in Terms of Canonical Walks}

We show how the swap walks can be obtained as a signed sum of canonical 
walks. This result follows from binomial inversion which we recall
next.

\begin{fact}[Binomial Inversion, \cite{Bernstein02}]\label{fac:binomial_inversion}
    Let $\parens*{a_n}_{n \ge 0}, \parens*{b_n}_{n \ge 0}$ be arbitrary sequences. Suppose for all $n \ge 0$ we have,
    \[ b_n = \sum_{j = 0}^n \binom{n}{j}\cdot (-1)^j\cdot a_j.\]
    Then, we also have
    \[ a_n = \sum_{j = 0}^n \binom{n}{j}\cdot (-1)^j \cdot b_j.\]
\end{fact}

\begin{corollary}\label{cor:swap_rectangular_inverse}
    Let $k, l, d \ge 0$ be given parameters such that $k + l \le d$ and $k \ge l$. For any simplicial complex $X(\le d)$, one has the following formula for the $u$-swapping walk from $X(k)$
    to $X(l)$ in terms of the canonical $j$-walks:
    \[ \binom{k }{l - u} \rswap{k}{l}{u} = \sum_{j = 0}^u (-1)^{u - j} \cdot \binom{k + j}{l} \cdot \binom{u}{j} \cdot \rnwalk{k}{l}{j}.\]
\end{corollary}
\begin{proof}
    Fix faces $\ess \in X(k)$ and $\tee \in X(l)$ and set for all $j \in [0, u]$
    \[ a_j := \binom{k}{l - j} \cdot (-1)^j \cdot \rswap{k}{l}{j}(\ess, \tee).\]
    Notice that we have by~\cref{lemma:formula_normal_rect}
    \[\binom{k + u}{l} \cdot \rnwalk{k}{l}{u}(\ess, \tee) = \sum_{j = 0}^u \binom{u}{j}\cdot (-1)^j \cdot a_j = \sum_{j = 0}^u \binom{u}{j} \cdot \binom{k}{l-j} \cdot \cdot \rswap{k}{l}{j}(\ess, \tee). \]
    i.e.~if we set
    \[ b_u = \binom{k + u}{l} \cdot \rnwalk{k}{l}{u}(\ess, \tee),\]
    we can apply~\cref{fac:binomial_inversion} to obtain
    \begin{eqnarray*}
        \binom{k}{l - u} \cdot (-1)^u \cdot \rswap{k}{l}{u}(\ess, \tee) & = & a_u\\
        & = & \sum_{j = 0}^u \binom{u}{j}\cdot (-1)^j \cdot b_j\\
        & = & \sum_{j = 0}^u \binom{u}{j} \cdot \binom{k + j}{l}\cdot (-1)^j \cdot \rnwalk{k}{l}{j}(\ess, \tee).
    \end{eqnarray*}
    Dividing both sides of this equation by $(-1)^u$ yields the desired result.
\end{proof}


%% file: asd.tex
Swap walks arise naturally in our $k$-CSPs approximation scheme on
HDXs where the running time and the quality of approximation depend on
the expansion of these walks. For this reason, we analyze the spectra
of swap walks. We show that swap walks $\rfswap{k}{k}$ of
\text{$\gamma$-HDXs} are indeed expanding for $\gamma$ sufficiently
small. More precisely, the first main result of this section is the
following.
\begin{theorem}[Swap Walk Spectral Bound]\label{theo:sq_swap_gap}
  Let $X(\le d)$ be a $\gamma$-HDX with $d \ge 2k$. Then the second
  largest singular value $\sigma_2(\rfswap{k}{k})$ of the swap
  operator $\rfswap{k}{k}$ is
  \[
  \sigma_2(\rfswap{k}{k}) ~\le~ \gamma \cdot \left(2^7 \cdot k^4 \cdot 2^{3k} \cdot k^{k}\right).
  \]
\end{theorem}
\cref{theo:sq_swap_gap} is enough for the analysis of our $k$-CSP
approximation scheme when $k$ is a power of two. However, to analyze
general $k$-CSPs on HDXs we need to understand the spectra of general
swap walks $\rfswap{k}{l}$ where $k$ may differ from $l$. Therefore,
we generalize the spectral analysis of $\rfswap{k}{k}$ above to
$\rfswap{k}{l}$ obtaining~\cref{thm:swap-eig-bd}, our second main
result of this section.
\begin{theorem}[Rectangular Swap Walk Spectral Bound]\label{thm:swap-eig-bd}
  Suppose $X(\le d)$ is a $\gamma$-HDX with $d \ge k + l$ and $k \le
  l$. Then the largest non-trivial singular value
  $\sigma_2(\rfswap{k}{l})$ of the swap operator $\rfswap{k}{l}$ is
  \[
  \sigma_2(\rfswap{k}{l}) ~\le~ \sqrt{\gamma \cdot \left(2^8 \cdot k^2 \ell^2 \cdot 2^{2k+4l} \cdot k^{k} \right)}.
  \]
\end{theorem}

\subsection{Square Swap Walks $\rfswap{k}{k}$}\label{subsec:proof_strategy_swap_k_k}

We prove~\cref{theo:sq_swap_gap} by connecting the spectral structure
of $\rfswap{k}{k}$ of general $\gamma$-HDXs to the well behaved case
of complete simplicial complexes. To distinguish these two cases we
denote by $\rfswap{k}{k}^{\Delta}$ the swap $\rfswap{k}{k}$ of
complete complexes~\footnote{The precise parameters of the complete
  complex $\Delta_d(n)$ where $\rfswap{k}{k}^{\Delta}$ lives will not
  be important. We only require that $\rfswap{k}{k}^{\Delta}$ is well
  defined in the sense that $d \ge 2k$ and $n > d$.}. In fact,
$\rfswap{k}{k}^{\Delta}$ is the random walk operator of the well known
Kneser graph $K(n,k)$ (see~\cref{def:kneser}).
\begin{definition}[Kneser Graph $K(n,k)$~\cite{GodsilM15}]\label{def:kneser}
  The Kneser graph $K(n,k)$ is the graph $G=(V,E)$ where $V =
  \binom{[n]}{k}$ and $E = \{\{\ess,\tee\}~\vert~\ess \cap \tee =
  \emptyset\}$.
\end{definition}
Then at least for complete complexes we know that
$\rfswap{k}{k}^{\Delta}$ is expanding. This is a direct consequence
of~\cref{fact:kneser_singular_values}.
\begin{fact}[Kneser Graph~\cite{GodsilM15}]\label{fact:kneser_singular_values}
  The singular values~\footnote{The precise eigenvalues are also well
    known, but singular values are enough in our analysis.} of the
  Kneser graph $K(n,k)$ are
  \[
  {n-k-i \choose k-i},
  \]
  for $i = 0,\dots,k$.
\end{fact}
This means that $\sigma_2(\rfswap{k}{k}^{\Delta}) = O_k(1/n)$ as shown
in~\cref{claim:square_knesere_gap}.
\begin{claim}\label{claim:square_knesere_gap}
  Let $d \ge 2k$ and $\Delta_d(n)$ be the complete complex. The second
  largest singular value $\sigma_2(\rfswap{k}{k}^{\Delta})$ of the
  swap operator $\rfswap{k}{k}^{\Delta}$ on $\Delta_d(n)$ is
  \[
  \sigma_2(\rfswap{k}{k}^{\Delta}) ~=~ \frac{k}{n-k},
  \]
  provided $n \ge M_k$ where $M_k \in \mathbb{N}$ only depends on $k$.
\end{claim}

\begin{proof}
  First note that for the complete complex $\Delta_d(n)$, the operator
  $\rfswap{k}{k}^{\Delta}$ is the walk matrix of the Kneser graph
  $K(n,k)$. Since the degree of $K(n,k)$ is $\binom{n-k}{k}$, the
  result follows from~\cref{fact:kneser_singular_values}.
\end{proof}

Therefore, if we could claim that $\sigma_2(\rfswap{k}{k})$ of an
arbitrary $\gamma$-HDX is close to $\sigma_2(\rfswap{k}{k}^{\Delta})$
(provided $\gamma$ is sufficiently small), we would conclude that
general $\rfswap{k}{k}$ walks are also expanding. A priori there is no
reason why this claim should hold since a general $d$-sized
$\gamma$-HDX may have much fewer hyperedges ($O_d(n)$ versus
$\binom{n}{d}$ in the complete $\Delta_d(n)$). Fortunately, it turns
out that this claim is indeed true (up to $O_k(\gamma)$ errors).

To prove~\cref{theo:sq_swap_gap} we employ the beautiful expanding
poset (EPoset) machinery of Dikstein et
al.~\cite{DiksteinDFH18}. Before we delve into the full technical
analysis, it might be instructive to see how~\cref{theo:sq_swap_gap}
is obtained from understanding the quadratic form $\ip{\rfswap{k}{k}
  f}{f}$ where $f \in C^k$.

First we informally recall the decomposition $C^k = \sum_{i=0}^k
C^k_i$ from the EPoset machinery where $C^k_i$ can be thought of as
the space of approximate eigenfunctions of \textit{degree} $i$ of
$C^k$ (the precise definitions are deferred
to~\ref{subsec:hdx_harmonic_analysis}). In this decomposition, $C^k_0$
is defined as the space of constant functions of $C^k$.

We prove the stronger result that the $\rfswap{k}{k}$ operators
of any $\gamma$-HDX has an an approximate spectrum that only depends
on $k$ provided $\gamma$ is small enough. More precisely, we
prove~\cref{lemma:swap_quadratic_form}.
\begin{lemma}[Swap Quadratic Form]\label{lemma:swap_quadratic_form}
  Let $f = \sum_{i=0}^k f_i$ with $f_i \in C^k_i$. Suppose $X(\le d)$ is a
  $\gamma$-HDX with $d \ge 2k$. If $\gamma \le \epsilon \left(64
  k^{k+4}2^{3k+1} \right)^{-1}$, then
  \[
  \ip{\rfswap{k}{k} f}{f} ~=~ \sum_{i=0}^k \lambda_{k}(i) \cdot \ip{f_i}{f_i} ~\pm~ \epsilon,
  \]
  where $\lambda_{k}(i)$ depends only on $k$ and $i$, i.e.,
  $\lambda_{k}(i)$ is an approximate eigenvalue of $\rfswap{k}{k}$
  associated to space $C^k_i$.
\end{lemma}

\begin{remark}
  From~\cref{lemma:swap_quadratic_form}, it might seem that we are
  done since there exist approximate eigenvalues $\lambda_{k}(i)$ that
  only depend on $k$ and $i$. However, giving an explicit expression
  for these approximate eigenvalues is tricky. For this reason, we
  rely on the expansion of Kneser graphs as will be clear later.
\end{remark}

Towards showing~\cref{lemma:swap_quadratic_form}, we introduce the
notion of \textit{balanced} operators which in particular captures
canonical and swap walks and we show that the quadratic form
expression of~\cref{lemma:swap_quadratic_form} is a particular case of
a general result for $\ip{\Bee f}{f}$ where $\Bee$ is a general
\textit{balanced} operator. A \textit{balanced} operator in $C^k$ is
any operator that can be obtained as linear combination of
\textit{pure balanced} operators, the later being operators that are a
formal product of an equal number of up and down operators.

\begin{lemma}[General Quadratic Form]\label{lemma:asd_general_decomposition}
  Let $\epsilon \in (0,1)$ and let $\mathcal{Y} \subseteq
  \{\Why~\vert~\Why \colon C^k \to C^k\}$ be a collection of formal
  operators that are product of an equal number of up and down walks
  (i.e., pure balanced operators) not exceeding $\ell$ walks. Let
  $\Bee = \sum_{\Why \in \mathcal{Y}} \alpha^{\Why} \Why$ where
  $\alpha^{\Why} \in \mathbb{R}$ and let $f = \sum_{i=0}^k f_i$ with
  $f_i \in C^k_i$. If $\gamma \le \epsilon \left(16 k^{k+2}\ell^2
  \sum_{\Why \in \mathcal{Y}} \vert \alpha^{\Why} \vert \right)^{-1}$,
  then
  \[
  \ip{\Bee f}{f} ~=~ \sum_{i=0}^k \left(\sum_{\Why \in \mathcal{Y}} \alpha^{\Why} \lambda^{\Why}_{k}(i)\right) \cdot \ip{f_i}{f_i} ~\pm~ \epsilon,
  \]
  where $\lambda^{\Why}_{k}(i)$ depends only on the operators
  appearing in the formal expression of $\Why$, $i$ and $k$, i.e.,
  $\lambda^{\Why}_{k}(i)$ is the approximate eigenvalue of $\Why$
  associated to $C^k_i$.
\end{lemma}

\begin{remark}
  Note that our result generalizes the analysis
  of~\cite{DiksteinDFH18} for expanding posets of HDXs which
  considered the particular case $\Bee = \Dee_{k+1}\Up_k$. Moreover,
  their error term analysis treated all the parameters not depending
  on the number of vertices $n$ as constants. In this work we make the
  dependence on the parameters explicit since this dependence is
  important in understanding the limits of our $k$-CSPs approximation
  scheme on HDXs. The beautiful EPoset machinery~\cite{DiksteinDFH18}
  is instrumental in our analysis.
\end{remark}

Now, we are ready to prove~\cref{theo:sq_swap_gap}. For convenience we
restate it below.
\begin{theorem}[Swap Walk Spectral Bound (restatement of~\cref{theo:sq_swap_gap})]
  Let $X(\le d)$ be a $\gamma$-HDX with $d \ge 2k$. For every $\sigma \in
  (0,1)$, if $\gamma \le \sigma \cdot \left(64 k^{k+4}2^{3k+1}
  \right)^{-1}$, then the second largest singular value
  $\sigma_2(\rfswap{k}{k})$ of the swap operator $\rfswap{k}{k}$ is
  \[
  \sigma_2(\rfswap{k}{k}) ~\le~ \sigma.
  \]
\end{theorem}

\begin{proof}
  First we show that for $i \in [k]$ the $i$-th approximate eigenvalue
  $\lambda_k(i)$ of the swap operator $\rfswap{k}{k}$ is actually
  zero. Note that for $i \in [k]$ the space $C^k_i$ is a non-trivial
  eigenspace (i.e., $C_i^k$ is not the space of constant
  functions). Let $\rfswap{k}{k}^{\Delta}$ be the swap operator of the
  complete complex $\Delta_{d}(n)$. On one
  hand~\cref{claim:square_knesere_gap} gives
  \[
  \sigma_2(\rfswap{k}{k}^{\Delta}) ~=~ \max_{f \in C^k \colon f \perp 1, \norm{f}=1} ~\left\vert \ip{\rfswap{k}{k}^{\Delta}f}{f} \right\vert ~=~ O_k\left(\frac{1}{n}\right).
  \]
  On the other hand since $\Delta_{d}(n)$ is a $\gamma^{\Delta}$-HDX
  where $\gamma^{\Delta} = O_k(1/n)$, if $n$ is sufficiently large
  we have $\gamma^{\Delta} \le \gamma$ and
  thus~\cref{lemma:asd_general_decomposition} can be applied to give
  \[
  \sigma_2(\rfswap{k}{k}^{\Delta}) ~\ge~ \max_{f_i \in C^k_i \colon i \in [k], \norm{f_i}=1} ~\left\vert \ip{\rfswap{k}{k}^{\Delta}f_i}{f_i} \right\vert ~=~  \left\vert \lambda_{k}(i) \right\vert \cdot \ip{f_i}{f_i} ~\pm~ O_k\left(\frac{1}{n}\right).
  \]  
  Since $n$ is arbitrary and $\lambda_{k}(i)$ depends only on $k$ and
  $i$, we obtain $\lambda_{k}(i) = 0$ as claimed.  Now
  applying~\cref{lemma:asd_general_decomposition} to the swap operator
  $\rfswap{k}{k}$ of the $\gamma$-HDX $X(\le d)$ yields
  \[
  \sigma_2(\rfswap{k}{k}) ~=~ \max_{f \in C^k \colon f \perp 1, \norm{f}=1} ~\left\vert \ip{\rfswap{k}{k}f}{f} \right\vert ~\le~ \max_{i \in [k]}~\left\lvert \lambda_{k}(i) \right\vert ~+~ \sigma ~=~ \sigma,
  \]
  concluding the proof.
\end{proof}

\subsection{Expanding Posets and Balanced Operators}
We state the definitions used in our technical proofs starting with $\gamma$-EPoset from~\cite{DiksteinDFH18}.
\begin{definition}[$\gamma$-EPoset adapted from~\cite{DiksteinDFH18}]\label{def:gamma_hdx_op_def}
  A complex $X(\le d)$ with operators $\Up_{0},\dots,\Up_{d-1}$,
  $\Dee_{1},\dots,\Dee_{d}$ is said to be a
  $\gamma$-EPoset~\footnote{We tailor their general EPoset definition
    to HDXs. In fact, what they call $\gamma$-HDX we call
    $\gamma$-EPoset. Moreover, what they call $\gamma$-HD expander we
    call $\gamma$-HDX.} provided
  \begin{equation}\label{eq:gamma_hdx_approx_op}
    \norm{\matr M^+_i - \Up_{i-1}\Dee_i}_{\textup{op}} \le \gamma,
  \end{equation}
  for every $i=1,\dots,d-1$, where
  \[
  \matr M^+_i ~\coloneqq~ \frac{i+1}{i} \left(\Dee_{i+1}\Up_i - \frac{1}{i+1}\Ide\right),
  \]
  i.e., $\matr M^+_i$ is the non-lazy version of the random walk
  $\nwalk{i}{1} = \Dee_{i+1}\Up_i$.
\end{definition}

\cref{def:gamma_hdx_op_def} can be directly used as an operational
definition of high-dimension expansion as done
in~\cite{DiksteinDFH18}. To us it is important that $\gamma$-HDXs are
also $\gamma$-EPosets as established in~\cref{lemma:hdx_vs_eposet}. In
fact, these two notions are known to be closely related.

\begin{lemma}[From~\cite{DiksteinDFH18}]\label{lemma:hdx_vs_eposet}
  Let $X$ be a $d$-sized simplicial complex.
  \begin{itemize}
     \item If $X$ is a $\gamma$-HDX, then $X$ is a $\gamma$-EPoset.
     \item If $X$ is a $\gamma$-EPoset, then $X$ is a $3 d\gamma$-HDX.
  \end{itemize}
\end{lemma}

Naturally the complete complex $\Delta_d(n)$ is a $\gamma$-EPoset
since it is a $\gamma$-HDX. Moreover, in this particular case $\gamma$
vanishes as $n$ grows.
\begin{lemma}[From~\cite{DiksteinDFH18}]
  The complete complex $\Delta_d(n)$ is a $\gamma$-EPoset with $\gamma
  = O_d\left(1/n\right)$.
\end{lemma}

\subsubsection*{Harmonic Analysis on Simplicial Complexes}\label{subsec:hdx_harmonic_analysis}

The space $C^k$ defined in~\cref{subsec:natural_walks} can be
decomposed into subspaces $C^k_i$ of functions of \textit{degree} $i$
for $0 \le i \le k$ where
\[
C^k_i ~\coloneqq~ \{\Up^{k-i}  h_i~\vert~h_i \in H_i \},
\]
with $H_i ~\coloneqq~ \ker{\left(\Dee_i\right)}$, and
\[
C^{k}_{0} ~\coloneqq~ \{ f \colon X(k) \to \mathbb{R}~|~\textrm{ $ f$ is constant}\}.
\]
More precisely, we have the following.
\begin{lemma}[From~\cite{DiksteinDFH18}]\label{lemma:a_hdx_space_decomp}
  \[
  C^k ~=~ \sum_{i=0}^{k} C^{k}_i.
  \]
\end{lemma}
\cref{lemma:a_hdx_space_decomp} is proven in~\cref{app:harmonic_hdx}
as~\cref{lemma:hdx_space_decomp}.

For convenience set $\vec{\delta} \in \mathbb{R}^{d-1}$ such that
$\delta_i = 1/(i+1)$ for $i \in [d-1]$. It will also be convenient to
work with the following equivalent version
of~\cref{eq:gamma_hdx_approx_op}
\begin{equation}\label{eq:complete_skewed_commutation}
  \norm{\Dee_{i+1}\Up_i -(1-\delta_i) \Up_{i-1}\Dee_i - \delta_i \Ide}_{\textup{op}} ~\le~ \frac{i}{i+1} \gamma.
\end{equation}

Towards our goal of understanding quadratic forms of swap operators we
study the approximate spectrum of operators of the form $\Why
=\Why_{\ell}\dots\Why_1$ where each $\Why_i$ is either an up or down
operator, namely, $\Why$ is a generalized random walk of $\ell$
steps. We regard the expression $\Why_{\ell}\dots\Why_1$ defining
$\Why$ as a formal product.
\begin{definition}[Pure Balanced Operator]
  We call $\Why \colon C^k \to C^k$ a pure balanced operator if $\Why$
  can be defined as product $\Why_{\ell}\dots\Why_1$~\footnote{For the
    analysis it is convenient to order the indices appearing in
    $\Why_{\ell}\dots\Why_1$ in decreasing order from left to right.}
  where each $\Why_i$ is either an up or down operator. When we say
  that the spectrum of $\Why$ depends on $\Why$ we mean that it
  depends on $k$ and on the formal expression $\Why_{\ell}\dots\Why_1$
  (i.e., pattern of up and down operators).
\end{definition}
\begin{remark}
  By definition canonical walks $\rnwalk{k}{k}{u}$ are \textit{pure
    balanced} operators.
\end{remark}
Taking linear combinations of \textit{pure balanced} operators leads
to the notion of \textit{balanced} operators.
\begin{definition}[Balanced Operator]
  We call $\Bee \colon C^k \to C^k$ a balanced operator provided there
  exists a set of pure balanced operators $\mathcal{Y}$ such that
  \[
  \Bee ~=~ \sum_{\Why \in \mathcal{Y}} \alpha^{\Why} \cdot \Why,
  \]
  where $\alpha^{\Why} \in \mathbb{R}$.
\end{definition}
\begin{remark}
  \cref{cor:swap_rectangular_inverse} establishes that
  $\rswap{k}{k}{u}$ are \textit{balanced} operators. In particular,
  $\rfswap{k}{k}$ is a \textit{balanced} operator.
\end{remark}

It turns out that at a more crude level the behavior of $\Why$ is
governed by how the number of up operators compares to the number of
down operators. For this reason, it is convenient to
define~$\mathcal{U}(\Why)= \{ \Why_i~\vert~\Why_i \textup{ is an up
  operator} \}$ and~$\calD(\Why) = \{\Why_i~\vert~\Why_i \textup{ is a
  down operator} \}$ where $\Why$ is a \textit{pure balanced}
operator. When $\Why$ is clear in the context we use $\mathcal{U} =
\mathcal{U}(\Why)$ and~$\calD = \calD(\Why)$.

Henceforth we assume $h_i \in H_i = \ker\left(\Dee_i\right)$, $f_i \in
C^k_i$ and $g \in C^k$. This convention will make the statements of
the technical results of~\cref{subsec:asd_tech_result} cleaner.

\subsection{Quadratic Forms over Balanced Operators}\label{subsec:asd_tech_result}

Now we establish all the technical results leading to and including
the analysis of quadratic forms over \textit{balanced operators}. By
considering this general class of operators our analysis generalizes
the analysis given in~\cite{DiksteinDFH18}. At the same time we refine
their error terms analysis by making the dependence on the EPoset
parameters explicit. Recall that an explicit dependence on these
parameters is important in understanding the limits of our $k$-CSP
approximation scheme.

\begin{lemma}[General Quadratic Form (restatement of~\cref{lemma:asd_general_decomposition})]
  Let $\epsilon \in (0,1)$ and let $\mathcal{Y} \subseteq
  \{\Why~\vert~\Why \colon C^k \to C^k\}$ be a collection of formal
  operators that are product of an equal number of up and down walks
  (i.e., pure balanced operators) not exceeding $\ell$ walks. Let
  $\Bee = \sum_{\Why \in \mathcal{Y}} \alpha^{\Why} \Why$ where
  $\alpha^{\Why} \in \mathbb{R}$ and let $f = \sum_{i=0}^k f_i$ with
  $f_i \in C^k_i$. If $\gamma \le \epsilon \left(16 k^{k+2}\ell^2
  \sum_{\Why \in \mathcal{Y}} \vert \alpha^{\Why} \vert \right)^{-1}$,
  then
  \[
  \ip{\Bee f}{f} ~=~ \sum_{i=0}^k \left(\sum_{\Why \in \mathcal{Y}} \alpha^{\Why} \lambda^{\Why}_{k}(i)\right) \cdot \ip{f_i}{f_i} ~\pm~ \epsilon,
  \]
  where $\lambda^{\Why}_{k}(i)$ depends only on the operators
  appearing in the formal expression of $\Why$, $i$ and $k$, i.e.,
  $\lambda^{\Why}_{k}(i)$ is the approximate eigenvalue of $\Why$
  associated to $C^k_i$.
\end{lemma}
Since swap walks are \textit{balanced operators}, we will deduced the
following (as proven later).
\begin{lemma}[Swap Quadratic Form (restatement of~\cref{lemma:swap_quadratic_form})]
  Let $f = \sum_{i=0}^k f_i$ with $f_i \in C^k_i$. Suppose $X(\le d)$ is a
  $\gamma$-HDX with $d \ge 2k$. If $\gamma \le \epsilon \left(64
  k^{k+4}2^{3k+1} \right)^{-1}$, then
  \[
  \ip{\rfswap{k}{k} f}{f} ~=~ \sum_{i=0}^k \lambda_{k}(i) \cdot \ip{f_i}{f_i}~\pm~\epsilon,
  \]
  where $\lambda_{k}(i)$ depends on only on $k$ an $i$, i.e.,
  $\lambda_{k}(i)$ is an approximate eigenvalue of $\rfswap{k}{k}$
  associated to space $C^k_i$.
\end{lemma}

The next result,~\cref{lemma:asd_exact_u_operator}, (implicit
in~\cite{DiksteinDFH18}) will be key in establishing that the spectral
structure of $\gamma$-EPosets is fully determined by the parameters in
$\vec{\delta}$ provided $\gamma$ is small enough. Note that the
Eposet~\cref{def:gamma_hdx_op_def} provides a ``calculus'' for
rearranging a single pair of up and down $\Dee \Up$. The next result
treats the more general case of $\Dee\Up\cdots \Up$.
\begin{lemma}[Structure Lemma]\label{lemma:asd_exact_u_operator}
  Suppose $\lvert \calD \rvert = 1$. Let $\Why_{c} \in \calD$ be the
  unique down operator in $\Why_{\ell}\dots\Why_1$. If
  $\norm{\Aye}_{\textup{op}} \le 1$, then
  \[
  \ip{\Aye \Why_{\ell} \dots \Why_1 h_i}{g} ~=~ \begin{cases} 0 & \textup{if $\ell = 1$ or $c = 1$}\\
    Q_{c,i}(\vec{\delta}) \cdot \ip{\Aye \Up^{\ell-2} h_i}{g}~\pm~\left(c-1\right) \cdot \gamma \norm{h_i} \norm{g} & \textup{otherwise},
    \end{cases}
  \]
  where $Q_{c,i}$ is a polynomial in the variables $\vec{\delta}$
  depending on $c,i$ such that $Q_{c,i}(\vec{\delta}) \le 1$.
\end{lemma}

\begin{proof}
  We induct on $(\ell, c)$. If $\ell = 1$ or $c = 1$, we have $\Why_1
  h_i = \Dee_i h_i = 0$ so the result trivially holds. Otherwise, we
  have $\Why_{c}\Why_{c-1} = \Dee_{j+1}\Up_j$ where $j=i+c-2$. Then
  \begin{align*}
    \ip{\Aye \Why_{\ell} \dots \Why_{c+1}(\Why_{c}\Why_{c-1})  \Why_{c-2} \dots \Why_1 h_i}{g},
  \end{align*}
  becomes
  {\small
    \begin{align*}
      & (1-\delta_j) \cdot \ip{\Aye \Why_{\ell} \dots \Why_{c+1}\Up_{j-1}\Dee_j \Why_{c-2} \dots \Why_1h_i}{g}~+~
       \delta_j\cdot \ip{\Aye \Why_{\ell} \dots \Why_{c+1} \Why_{c-2} \dots \Why_1 h_i}{g}~\pm~\gamma \norm{h_i}\norm{g} && \textup{(\cref{eq:gamma_hdx_approx_op})}\\
      & ~=~  (1-\delta_j) \cdot \ip{\Aye \Why_1\dots \Why_{c-1} \Up_{j-1}\Dee_j \Why_{c+2}\dots\Why_{\ell}h_i}{g}~+~
       \delta_j \cdot \ip{\Aye \Up^{\ell - 2} h_i}{g}~\pm~\gamma \norm{h_i}\norm{g}\\
      & ~=~  (1 -\delta_j)\cdot Q_{c-1,i}(\vec{\delta}) \cdot \ip{\Aye \Up^{\ell - 2} h_i}{g}~\pm~(1 -\delta_j)\cdot \left(c-2\right) \gamma \norm{h_i}\norm{g}
       ~+~\delta_j \cdot \ip{\Aye \Up^{\ell - 2} h_i}{g}~\pm~\gamma \norm{h_i}\norm{g} && \textup{(I.H.)}\\
      & ~=~  Q_{c,i}(\vec{\delta}) \cdot \ip{\Aye \Up^{\ell - 2} h_i}{g}~\pm~\left(c - 1 \right) \cdot \gamma \norm{h_i}\norm{g}.    
  \end{align*}}
\end{proof}

With~\cref{lemma:asd_exact_u_operator} we are close to recover the
approximate spectrum of $\Dee_{k+1} \Up_k$
from~\cite{DiksteinDFH18}. However, in our application we will need to
analyze more general operators, namely, \textit{pure balanced} and
\textit{balanced} operators.

\begin{lemma}[Refinement of~\cite{DiksteinDFH18}]\label{lemma:eigenvalue_single_app}
  If $\norm{\Aye}_{\textup{op}} \le 1$, then
  \[
  \ip{\Aye \Dee_{k+1} \Up_k f_i}{g} ~=~ \lambda_i \cdot \ip{\Aye f_i}{g}~\pm~(k-i+1) \cdot \gamma \norm{h_i} \norm{g},
  \]
  where $\lambda_i = Q_{k-i+2,i}(\vec{\delta})$.
\end{lemma}

\begin{proof}
  Recall that $f_i = \Up^{k-i} h_i$ where $h_i \in
  \ker\left(\Dee_i\right)$. Set $\Why = \Dee_{k+1} \Up_k
  \Up^{k-i}$.~\cref{lemma:asd_exact_u_operator} yields
  \[
  \ip{\Aye \Dee_{k+1} \Up_k f_i}{g} ~=~ \lambda_i \cdot \ip{\Aye f_i}{g}~\pm~(k-i+1) \cdot \gamma \norm{h_i} \norm{g},
  \]
  where $\lambda_i = Q_{k-i+2,i}(\vec{\delta})$.
\end{proof}

%
%
Then powers of the operator $\Dee_{k+1}\Up_k$ behave as expected.
\begin{lemma}[Exponentiation Lemma]\label{cor:asd_exponentiation}
  \[
  \ip{\left(\Dee_{k+1}\Up_k\right)^{s}f_i}{f_i} ~=~ \lambda_i^{s} \cdot \norm{f_i}^2~\pm~s \cdot (k-i+1) \cdot \gamma \norm{h_i}\norm{f_i},
  \]
  where $\lambda_i$ is given in~\cref{lemma:eigenvalue_single_app}.
\end{lemma}

\begin{proof}
  Follows immediately from the foregoing and the fact that
  $\norm{\Dee_{k+1}\Up_k}_{\textup{op}} = 1$.
\end{proof}

In case $\lvert \calD \rvert > \lvert \calU \rvert$, $\Why \colon C^i
\to C^j$ is an operator whose kernel approximately contains
$\ker(\matr D_i)$ as the following lemma makes precise.
\begin{lemma}[Refinement of~\cite{DiksteinDFH18}]\label{lemma:asd_extra_d_operator}
  If $\lvert \calD \rvert > \lvert \calU \rvert$ and $h_i \in
  \ker\left(\Dee_i\right)$, then
  \[
  \ip{\Aye \Why_{\ell} \dots \Why_1 h_i}{g} ~=~ \pm \ell^2 \cdot \gamma \norm{h_i} \norm{g},
  \]
  provided $\norm{\Aye}_{\textup{op}} \le 1$.
\end{lemma}

\begin{proof}
  Let $c \in [\ell]$ be the smallest index for which $\Why_{c}$ is a
  down operator. Observe that $c < \ell/2$ since $\lvert \calD \rvert
  > \lvert \calU \rvert$. We induct on $m = \lvert \calD \rvert$. If
  $c = 1$, then $\ip{\Aye \matr D_i h_i}{g} = 0$. Hence assume $c,m > 1$
  implying $\Why_{c}\Why_{c-1} = \Dee_{i+c}\Up_{i+c-1}$.
  Applying~\cref{lemma:asd_exact_u_operator} we obtain
  \begin{align*}
    \ip{\Aye \Why_{\ell}\dots\Why_1h_i}{g} &~=~ \ip{\left(\Aye \Why_{\ell} \dots \Why_{c+1}\right) \Dee \Up \Up^{c-2} h_i}{g}\\
                                         &~=~ Q_{c,i}(\vec{\delta}) \cdot \ip{\left(\Aye \Why_{\ell} \dots \Why_{c+1}\right) \Up^{c-2} h_i}{g}~\pm~\frac{\ell}{2} \cdot \gamma \norm{h_i} \norm{g}\\
                                         &~=~ \pm Q_{c,i}(\vec{\delta}) \cdot \left(\ell-2\right)^2 \cdot \gamma \norm{h_i} \norm{g}~\pm~\frac{\ell}{2} \cdot \gamma \norm{h_i} \norm{g} && \text{(Induction)}\\
                                         &~=~ \pm \ell^2 \cdot \gamma \norm{h_i} \norm{g},
  \end{align*}
  where in the last derivation we used $Q_{c,i}(\vec{\delta}) \le 1$.
\end{proof}

We turn to an important particular case of $\lvert \calD \rvert =
\lvert \calU \rvert$, namely, the canonical walks. We show that
$\nwalk{k}{u}$ is approximately a polynomial in the operator
$\Dee_{k+1}\Up_k$. As a warm up consider the case $\nwalk{k}{2}=
\Dee_{k+1}\Dee_{k+2} \Up_{k+1} \Up_k$. Using
the~\cref{eq:complete_skewed_commutation}, we get
\begin{align*}
  \nwalk{k}{2} &~\approx~ (1-\delta_{k+1}) \cdot \Dee_{k+1}\Up_{k} \Dee_{k+1}  \Up_k ~+~ \delta_{k+1} \cdot \Dee_{k+1} \Up_k \\
               &~=~ (1-\delta_{k+1}) \cdot \left(\Dee_{k+1}\Up_{k}\right)^2 ~+~ \delta_{k+1} \cdot \Dee_{k+1}\Up_{k}.
\end{align*}
Inspecting this polynomial more carefully we see that that its
coefficients form a probability distribution. This property holds in
general as the following~\cref{lemma:asd_equal_d_u_operators}
shows. This gives an alternative (approximate) random walk
interpretation of $\nwalk{k}{u}$ as the walk that first selects the
power $s$ according to the distribution encoded in the polynomial and
then moves according to $\left(\Dee_{k+1}\Up_{k}\right)^s$.

\begin{lemma}[Canonical Polynomials]\label{lemma:asd_equal_d_u_operators}
  For $k,u \ge 0$ there exists a degree $u$ univariate polynomial
  $F_{u,k,\vec{\delta}}^N$ depending only on
  $u,k,\vec{\delta}$ such that
  \[
  \norm{\nwalk{k}{u} - F_{u,k,\vec{\delta}}^N(\Dee_{k + 1}\Up_{k})}_{\textup{op}} ~\le~ (u-1)^2 \cdot \gamma.
  \]
  Moreover, the coefficients of this polynomial form a probability
  distribution, i.e., $F_{u,k,\vec{\delta}}^N(x) =
  \sum_{i=0}^u c_i x^i$ where $\sum_{i=0}^u c_i = 1$ and $c_i \ge 0$
  for $i=0,\dots,u$.
\end{lemma}

\begin{proof}
    For $u=0$, $\nwalk{k}{0} =\Ide$ and the lemma trivially
    follows. Similarly, if $u=1$, $\nwalk{k}{1} = \Dee_{k+1}\Up_{k}$. Now
    suppose $u \ge 2$. Set $\Why = \nwalk{k}{u}$, i.e.,
    \[
    \Why ~=~ \Dee_{k + 1}\dots \left(\Dee_{k+u}\Up_{k+u-1}\right)\dots \Up_{k}.
    \]
    For convenience let $j = k+u-1$. Using
    the~\cref{eq:complete_skewed_commutation} we can replace
    $\Dee_{j+1}\Up_{j}$ in $\Why$ by $(1-\delta_j)\Up_{j-1}\Dee_{j} +
    \delta_j \Ide$ incurring an error of $\gamma$ (in spectral norm) and
    yielding
    \[
    \Why ~\approx~ (1-\delta_j)\cdot \Why' ~+~ \delta_j \cdot \nwalk{k}{u-1},
    \]
    where $\Why'$ was obtained from $\Why$ by moving the rightmost
    occurence of a down operator (in this case $\Dee_{j+1}$) one
    position to right. We continue this process of moving the
    rightmost occurrence of a down operator until the resulting
    operator is up to $(u-1) \cdot \gamma$ error
    \[
    \alpha \cdot \nwalk{k}{u-1} \left(\Dee_{k + 1}\Up_{k}\right) ~+~ \beta \cdot \nwalk{k}{u-1},
    \]
    where $\alpha = \prod_{i=k+1}^j \left(1-\delta_i \right)$ and
    $\beta = \sum_{i=k+1}^{j} \delta_i\prod_{i=k+1}^j \left(1-\delta_i
    \right)$.  Since $\delta_i = \delta_i > 0$, $\alpha,\beta$ are non
    negative and form a probability distribution.  Now the result
    follows from the induction hypothesis applied to $\nwalk{k}{u-1}$.
\end{proof}
\begin{remark}
  Having a polynomial expression $F_{u,k,\vec{\delta}}^N(\Dee_{k +
    1}\Up_{k}) \approx \nwalk{k}{u}$ and knowing that $\rfswap{k}{k}$
  can be written as linear combination of canonical walks, we could
  deduce that $\rfswap{k}{k}$ is also approximately a polynomial in
  $\Dee_{k + 1}\Up_{k}$. Using an error refined version of
  the~\cref{cor:asd_exponentiation} (showing that exponentiation of
  $\Dee_{k + 1}\Up_{k}$ behaves naturally), we could deduce the
  approximate spectrum of $\rfswap{k}{k}$. We avoid this approach
  since it analysis introduces unnecessary error terms and we can
  understand quadratic forms of \textit{pure balanced} operators
  directly.
\end{remark}
\begin{remark}
  The canonical polynomial $F_{u,k,\vec{\delta}}^N(\Dee_{k +
    1}\Up_{k})$ is used later in the error analysis that relates the
  norms $\norm{h_i}$ and $\norm{f_i}$ (\cref{lemma:asd_fnorm}).
\end{remark}
  
Now we consider $\Why$ where $\lvert \calD \rvert = \lvert \calU
\rvert$ in full generality. We show how the quadratic form of $\Why$
behaves in terms of the approximate eigenspace decomposition $C^k =
\sum_{i=0}^k C^k_i$.

\begin{lemma}[Pure Balanced Walks]\label{lemma:asd_equal_d_u_operators_general}
  Suppose $\Why = \Why_{\ell}\dots \Why_1$ is a product of an equal
  number of up and down operators, i.e., $\lvert \calD \rvert = \lvert
  \calU \rvert$. Then for $f_i \in C^k_i$
  \[
  \ip{\Why f_i }{f_i} ~=~ \lambda_{k,i}^{\Why} \cdot \ip{f_i }{f_i} ~\pm~ \gamma \cdot (\ell^2+\ell(k-i-1)) \norm{h_i}\norm{f_i},
  \]
  where $\lambda_{k,i}^{\Why}$ is an approximate eigenvalue depending only on
  $\Why$, $k$ and $i$.
\end{lemma}

\begin{proof}
  We induct on even $\ell$. For $\ell = 0$, the result trivially
  follows so assume $\ell \ge 2$. Let $c \in [\ell]$ be the smallest
  index of a down operator. Set $\Aye = \Why_{\ell}\dots \Why_{c+1}$ and
  let $\Why' = \Why_{c}\dots \Why_1 = \Dee \Up\dots\Up$. Observe that
  \[
  \ip{\Aye \Why' f_i }{f_i} ~=~ \ip{\Aye \Dee \Up^{c-1 + k-i} h_i }{f_i}.
  \]
  Applying~\cref{lemma:asd_exact_u_operator} to the RHS above gives
  \[
  \ip{\Aye \Dee \Up^{c-1 + k-i} h_i }{f_i} = Q_{c-1 + k-i,i}(\vec{\delta}) \cdot \ip{\Aye \Up^{c-2} f_i }{f_i}  ~\pm~ \left(c+k-i-2\right) \cdot \gamma \norm{h_i} \norm{f_i}.
  \]
  Applying the induction hypothesis to $\Why'' = A\Up^{c-2}$ in the
  above RHS yields
  \begin{align*}
    & Q_{c-1 + k-i,i}(\vec{\delta}) \cdot \lambda^{\Why''}_{k,i} \ip{f_i }{f_i} \\
    & ~\pm~ Q_{c-1 + k-i,i}(\vec{\delta}) \cdot \gamma \cdot ((\ell-1)^2+(\ell-1)(k-i-1)) \norm{h_i}\norm{f_i}\\
    & ~\pm~ \left(c+k-i-2\right) \cdot \gamma \norm{h_i} \norm{f_i} \\
    & ~=~ \lambda^{\Why}_{k,i} \cdot \ip{f_i }{f_i} ~\pm~ \gamma \cdot (\ell^2+\ell(k-i-1)) \norm{h_i}\norm{f_i},
  \end{align*}
  where $\lambda^{\Why}_{k,i} = Q_{c-1 + k-i,i}(\vec{\delta}) \cdot
  \lambda^{\Why''}_{k,i}$ and the last equality follows from $Q_{c-1 +
    k-i,i}(\vec{\delta}) \le 1$ and $c \le \ell$.
\end{proof}

To understand all errors in the analysis
in~\cref{lemma:asd_equal_d_u_operators_general} we need to derive the
approximate orthogonality of $f_i$ and $f_j$ for $i \ne j$
from~\cite{DiksteinDFH18} in more detail. We start with the following
bound in terms of $h_i,h_j$.
\begin{lemma}[Refinement of~\cite{DiksteinDFH18}]\label{lemma:asd_crude_orthogonality}
  For $i \ne j$,
  \[
  \ip{f_i}{f_j} ~=~ \pm (2k-i-j)^2 \cdot \gamma \norm{h_i} \norm{h_j}.
  \]
\end{lemma}

\begin{proof}
  Recall that $f_i = \Up^{k-i} h_i$, $f_j = \Up^{k-j} h_j$ where $h_i
  \in \ker\left(\Dee_i\right)$, $h_j \in
  \ker\left(\Dee_j\right)$. Without loss of generality suppose $i >
  j$. We have
  \[
  \ip{\Up^{k-i}h_i}{\Up^{k-j}h_j} ~=~ \ip{\Dee^{k-j}\Up^{k-i}h_i}{h_j}.
  \] Since $k-j > k- i$, the result follows from~\cref{lemma:asd_extra_d_operator}.
\end{proof}

To give a bound for~\cref{lemma:asd_crude_orthogonality} only in terms
of the eigenfunction norms $\norm{f_i}$ and not in terms of
$\norm{h_i}$, we need to understand how the norms of $h_i$ and $f_i$
are related.
\begin{lemma}[Refinement of~\cite{DiksteinDFH18}]\label{lemma:asd_fnorm}
  Let $\eta_{k,i} = (k-i)^2+1$ and let $\beta_i = \sqrt{\abs{
      F_{k-i,i,\vec{\delta}}^N(\delta_i) \pm \gamma \cdot
      \eta_{k,i}}}$ where $F_{k-i,k,\vec{\delta}}^N$ is a canonical
  polynomial of degree $k-i$
  from~\cref{lemma:asd_equal_d_u_operators}. Then
  \[
  \ip{f_i}{f_i} ~=~ \beta_i^2 \cdot \ip{h_i}{h_i}.
  \]
  Let $\theta_{k,i} = \left(i+1\right)^{k-i}$.
  Furthermore, if $\gamma \le 1/(2 \cdot \eta_{k,i} \cdot \theta_{k,i})$, then $\beta_i \ge
  \frac{1}{2\theta_{k,i}}$.
\end{lemma}

\begin{proof}
  Recall that $f_i = \Up^{k-i} h_i$ where $h_i \in
  \ker\left(\Dee_i\right)$. For $i = k$ the result trivially follows
  so assume $k > i$. First consider the case $k=i+1$. We have
  \begin{equation}\label{eq:f_i_f_i_one_step}
    \ip{\Up_ih_i}{\Up_ih_i} ~=~ \ip{\Dee_{i+1}\Up_i h_i}{h_i} ~=~ \delta_i \cdot \ip{h_i}{h_i} ~\pm~  \gamma \cdot \ip{h_i}{h_i}.
  \end{equation}
  For general $k > i$ we have
  \[
  \ip{\Up^{k-i}h_i}{\Up^{k-i}h_i} ~=~ \ip{\Dee^{k-i}\Up^{k-i}h_i}{h_i}.
  \]
  Applying~\cref{lemma:asd_equal_d_u_operators} to $\Dee^{k-i}\Up^{k-i}$ yields
  \[
  \ip{\Dee^{k-i}\Up^{k-i}h_i}{h_i} ~=~ \ip{F_{k-i,i,\vec{\delta}}^N(\Dee_{i+1} \Up_i) h_i}{h_i} ~\pm~ \gamma \cdot (k-i-1)^2.
  \]
  Combining~\cref{eq:f_i_f_i_one_step} and~\cref{cor:asd_exponentiation} gives
  \[
  \ip{F_{k-i,i,\vec{\delta}}^N(\Dee_{i+1} \Up_i) h_i}{h_i} ~\pm~ \gamma \cdot (k-i-1)^2 ~=~ \ip{F_{k-i,i,\vec{\delta}}^N(\delta_i) h_i}{h_i} ~\pm~ \gamma \cdot ((k-i)^2+1).
  \]
  Since $F_{k-i,i,\vec{\delta}}^N(x) = \sum_{i=0}^{k-i} c_i
  x^i$ where the coefficients $c_i$ form a probability distribution,
  we get
  \[
  F_{k-i,i,\vec{\delta}}^N(\delta_i) ~\ge ~\delta_i^{k-i} ~=~ \left(\frac{1}{i+1}\right)^{k-i}.
  \]
\end{proof}
Now, we can state the approximate
orthogonality~\cref{lemma:precise_orthogonality} in terms of the
eigenfunction norms.
\begin{lemma}[Approximate Orthogonality (refinement of~\cite{DiksteinDFH18})]\label{lemma:precise_orthogonality}
   Let $\eta_{k,s},\theta_{k,s}, \beta_s$ for $s \in \{i,j\}$ be given
   as in~\cref{lemma:asd_fnorm}. If $i \ne j$ and $\beta_i, \beta_j >
   0$, then
  \[
  \ip{f_i}{f_j} ~= ~\pm~ \frac{\gamma \cdot (2k-i-j)^2}{\beta_i \beta_j} \norm{f_i} \norm{f_j}.
  \]
  Furthermore, if $\gamma \le \min\left(1/(2 \cdot \eta_{k,i} \cdot \theta_{k,i}),
  1/(2 \cdot \eta_{k,j} \cdot \theta_{k,j})\right)$, then $\beta_i, \beta_j > 0$
  and
  \[
  \ip{f_i}{f_j} ~=~ \pm \gamma \cdot \theta_{k,i} \cdot \theta_{k,j} \cdot (2k-i-j)^2 \norm{f_i} \norm{f_j}.
  \]
\end{lemma}

\begin{proof}
  Follows directly from~\cref{lemma:asd_fnorm}.
\end{proof}

We generalize the quadratic form
of~\cref{lemma:asd_equal_d_u_operators_general} to linear combinations
of general \textit{pure balanced} operators $\Why$, namely, to
\textit{balanced} operators.
\begin{lemma}[General Quadratic Form (restatement of~\cref{lemma:asd_general_decomposition})]
  Let $\epsilon \in (0,1)$ and let $\mathcal{Y} \subseteq
  \{\Why~\vert~\Why \colon C^k \to C^k\}$ be a collection of formal
  operators that are product of an equal number of up and down walks
  (i.e., pure balanced operators) not exceeding $\ell$ walks. Let
  $\Bee = \sum_{\Why \in \mathcal{Y}} \alpha^{\Why} \Why$ where
  $\alpha^{\Why} \in \mathbb{R}$ and let $f = \sum_{i=0}^k f_i$ with
  $f_i \in C^k_i$. If $\gamma \le \epsilon \left(16 k^{k+2}\ell^2
  \sum_{\Why \in \mathcal{Y}} \vert \alpha^{\Why} \vert \right)^{-1}$,
  then
  \[
  \ip{\Bee f}{f} ~=~ \sum_{i=0}^k \left(\sum_{\Why \in \mathcal{Y}} \alpha^{\Why} \lambda^{\Why}_{k}(i)\right) \cdot \ip{f_i}{f_i} ~\pm~ \epsilon,
  \]
  where $\lambda^{\Why}_{k}(i)$ depends only on the operators
  appearing in the formal expression of $\Why$, $i$ and $k$, i.e.,
  $\lambda^{\Why}_{k}(i)$ is the approximate eigenvalue of $\Why$
  associated to $C^k_i$.
\end{lemma}

\begin{proof}
  Using~\cref{lemma:asd_equal_d_u_operators_general} and the assumption on $\gamma$ gives
  \begin{align*}
    \ip{\Bee f}{f} & ~=~ \sum_{i=0}^k \sum_{\Why \in \mathcal{Y}} \alpha^{\Why} \lambda^{\Why}_{k}(i) \cdot \ip{f_i}{f_i} \\
                   &  \qquad\quad ~+~ \sum_{i\ne j} \sum_{\Why \in \mathcal{Y}} \left( \alpha^{\Why} \lambda^{\Why}_{k}(i) \cdot \ip{f_i}{f_j} ~\pm~ \gamma \cdot \alpha^{\Why} (\ell^2 + \ell(k-i-1)) \ip{h_i}{f_j} \right)\\
                   & ~=~ \sum_{i=0}^k \sum_{\Why \in \mathcal{Y}} \alpha^{\Why} \lambda^{\Why}_{k}(i) \cdot \ip{f_i}{f_i} ~+~ \sum_{i\ne j} \sum_{\Why \in \mathcal{Y}} \alpha^{\Why} \lambda^{\Why}_{k}(i) \cdot \ip{f_i}{f_j} ~\pm~ \frac{\epsilon}{2}.
  \end{align*}
  Next we use~\cref{lemma:precise_orthogonality} to bound the second
  double summation and conclude the proof.
\end{proof}

We instantiate~\cref{lemma:precise_orthogonality} for swap walks with
their specific parameters.  First, we introduce some
notation. Using~\cref{cor:swap_rectangular_inverse}, we have
\begin{align*}
  \rfswap{k}{k} ~=~ \sum_{j = 0}^k (-1)^{k - j} \cdot \binom{k + j}{k}
  \cdot \binom{k}{j} \cdot \rnwalk{k}{k}{j} ~=~ \sum_{j = 0}^k  \alpha_j \cdot \rnwalk{k}{k}{j},
\end{align*}
where $\alpha_j = (-1)^{k - j} \cdot \binom{k + j}{k} \cdot
\binom{k}{j}$.

Finally, we have all the pieces to
prove~\cref{lemma:swap_quadratic_form} restated below.
\begin{lemma}[Swap Quadratic Form (restatement of~\cref{lemma:swap_quadratic_form})]
  Let $f = \sum_{i=0}^k f_i$ with $f_i \in C^k_i$. Suppose $X(\le d)$ is a
  $\gamma$-HDX with $d \ge 2k$. If $\gamma \le \epsilon \left(64
  k^{k+4}2^{3k+1} \right)^{-1}$, then
  \[
  \ip{\rfswap{k}{k} f}{f} ~=~ \sum_{i=0}^k \lambda_{k}(i) \cdot \ip{f_i}{f_i}~\pm~\epsilon,
  \]
  where $\lambda_{k}(i)$ depends on only on $k$ an $i$, i.e.,
  $\lambda_{k}(i)$ is an approximate eigenvalue of $\rfswap{k}{k}$
  associated to space $C^k_i$.
\end{lemma}

\begin{proof}
  First note that~\cref{lemma:asd_equal_d_u_operators_general}
  establishes the existence of approximate eigenvalues
  $\lambda_{k,j}(i)$ of $\rnwalk{k}{k}{j}$ corresponding to space
  $C^k_i$ for $i=0,\dots,k$ such that $\lambda_{k,j}(i)$ depends only
  on $k$, $i$ and $j$. To apply~\cref{lemma:asd_general_decomposition}
  we need to bound $\sum_{j= 0}^{k} \vert \alpha_{j} \vert$. Since
  \begin{align*}
    \sum_{j = 0}^{k} \vert \alpha_{j} \vert~=~ \sum_{j = 0}^{k} \binom{k + j}{k} \cdot \binom{k}{j}~\le~ 2^{k} \cdot \sum_{j=0}^{k} \binom{k + j}{k}  ~\le~ 2^{3k+1},
  \end{align*}
  we are done.
\end{proof}


\subsection{Rectangular Swap Walks $\rfswap{k}{l}$}\label{subsec:proof_strategy_swap_k_l}

We turn to the spectral analysis of rectangular swap walks, i.e.,
$\rfswap{k}{l}$ where $k \ne l$. Recall that to bound
$\sigma_2(\rfswap{k}{k})$ in~\cref{subsec:proof_strategy_swap_k_k} we
proved that the spectrum of $\rfswap{k}{k}$ for a $\gamma$-HDX is
close to the spectrum of $\rfswap{k}{k}^{\Delta}$ using the analysis
of quadratic forms over \textit{balanced} operators
from~\cref{subsec:asd_tech_result}. Then we appealed to the fact that
$\rfswap{k}{k}^{\Delta}$ is expanding since it is the walk operator of
the well known Kneser graph. In this rectangular case, we do not have
a classical result establishing that $\rfswap{k}{l}^{\Delta}$ is
expanding, but we were able to establish
it~\cref{cor:complete_rec_kneser}.
\begin{lemma}\label{cor:complete_rec_kneser}
  Let $d \ge k+l$ and $\Delta_d(n)$ be the complete complex. The
  second largest singular value $\sigma_2(\rfswap{k}{l}^{\Delta})$ of
  the swap operator $\rfswap{k}{l}^{\Delta}$ on $\Delta_d(n)$ is
  \[
  \sigma_2(\rfswap{k}{l}^{\Delta}) ~\le~\max\left(\frac{k}{n-k}, \frac{l}{n-l} \right),
  \]
  provided $n \ge M_{k,l}$ where $M_{k,l} \in \mathbb{N}$ only depends
  on $k$ and $l$.
\end{lemma}

Towards proving~\cref{cor:complete_rec_kneser} we first introduce a
generalization of Kneser graphs which we denote \textit{bipartite
  Kneser} graphs defined as follows.
\begin{definition}[General Bipartite Kneser Graph]
  Let $X(\le d)$ where $d \ge k + l$. We denote by $K^X(n,k,l)$ the
  bipartite graph on (vertex) partition $\left(X(k),X(l)\right)$ where
  $\ess \in X(k)$ is adjacent to $\tee \in X(l)$ if and only if
  $\ess\cap\tee$ is empty. We also refer to graphs of the form
  $K^X(n,k,l)$ as \textit{bipartite Kneser} graphs.
\end{definition}
It will be convenient to distinguish \textit{bipartite Kneser} graphs
coming from general $\gamma$-HDX and the complete complex
$\Delta_d(n)$.
\begin{definition}[Complete Bipartite Kneser Graph]
  Let $X(\le d)$ where $d \ge k + l$. If $X$ is the complete complex,
  i.e., $X = \Delta_d(n)$, then we denote $K^X(n,k,l)$ as simply as
  $K(n,k,l)$ and we refer to it as \textit{complete bipartite Kneser}.
\end{definition}

We obtain the spectra of \textit{bipartite Kneser} graphs
generalizing~\footnote{Note that the singular values of $K(n,k)$ can
  be deduced from the bipartite case.} the classical result
of~\cref{fact:kneser_singular_values}. More precisely, we
prove~\cref{lemma:rec_kneser_spectrum}.
\begin{lemma}[Bipartite Kneser Spectrum]\label{lemma:rec_kneser_spectrum}
  The non-zero eigenvalues of the (normalized) walk operator of
  $K(n,k,l)$ are $\pm \lambda_i$ where
  \[
  \lambda_i ~=~ \frac{ {n-k-i \choose l-i} {n-l-i \choose k-i}}{{n-k \choose l}{n-l \choose k}},
  \]
  for $i = 0,\dots,\min(k,l)$.  
\end{lemma}

Now the proof follows a similar strategy to the $\rfswap{k}{k}$,
namely, we analyze quadratic forms over $\rfswap{k}{k}$ using the
results from~\cref{subsec:asd_tech_result}
 
Let $X(\le d)$ where $d \ge k+l$. Let $\Aye_{k,l}$ be the (normalized)
walk operator of $K^X(n,k,l)$, i.e.,
\[
\Aye_{k,l} ~=~ \begin{pmatrix}
  0   & \rswap{k}{l}{l}\\
  \left(\rswap{k}{l}{l}\right)^{\dag} &0
\end{pmatrix}.
\]
To determine the spectrum of $\Aye_{k,l}$ it is enough to consider the
spectrum of $\Bee =
\rswap{k}{l}{l}\left(\rswap{k}{l}{l}\right)^{\dag}$. Using~\cref{cor:swap_rectangular_inverse},
we have {\small
\begin{align*}
  \Bee ~&=~ \left(\sum_{j = 0}^{l} (-1)^{l - j} \binom{k + j}{l} \cdot \binom{l}{j} \cdot \rnwalk{k}{l}{j} \right) \\
  & \qquad \qquad \left( \sum_{j' = 0}^{l} (-1)^{l - j'} \binom{k + j'}{l} \cdot \binom{l}{j'} \cdot \left(\rnwalk{k}{l}{j'}\right)^{\dag} \right) ~=~
  \sum_{j ,j'= 0}^{l} \alpha_{k,l,j,j'} \rnwalk{k}{l}{j}\rnwalk{l}{k}{j'+k-l},
\end{align*}
} for some coefficients $\alpha_{k,l,j,j'}$ depending only on $k$,
$l$, $i$, $j$ and $j'$. Since we have not yet used any specific
property of HDXs, these coefficients are the same for the complete
complex and general HDXs.

\begin{lemma}
  Let $X(\le d)$ be a $\gamma$-HDX with $d \ge k + l$. Let $f =
  \sum_{i=0}^k f_i$ with $f_i \in C^k_i$. For $\epsilon \in (0,1)$, if
  $\gamma \le \epsilon \left(64 k^{k+2}\ell^2 2^{2k+4l+2}
  \right)^{-1}$, then
  \[
  \ip{\Bee f}{f} ~=~ \sum_{i=0}^k \left(\sum_{j ,j'= 0}^{l} \alpha_{k,l,j,j'} \lambda_{k,l,j,j'}(i) \right) \cdot \ip{f_i}{f_i} ~+~ \epsilon,
  \]
  where $\lambda_{k,l,j,j'}(i)$ is the approximate eigenvalues of
  $\rnwalk{k}{l}{j}\rnwalk{l}{k}{j'+k-l}$ corresponding to space
  $C^k_i$. Furthermore, $\lambda_{k,l,j,j'}(i)$ depends only on
  $k$, $l$, $i$, $j$ and $j'$.
\end{lemma}

\begin{proof}
  First observe that each $\rnwalk{k}{l}{j}\rnwalk{l}{k}{j'+k-l}$ maps
  $C^k$ to itself, so it is a product of the same number of up and
  down operators. Now to apply~\cref{lemma:asd_general_decomposition}
  it only remains to bound $\sum_{j ,j'= 0}^{l} \vert \alpha_{k,l,j,j'}
  \vert$. Since
  \begin{align*}
    \sum_{j ,j'= 0}^{l} \vert \alpha_{k,l,j,j'} \vert & ~=~ \sum_{j ,j'= 0}^{l} \binom{k + j}{l} \cdot \binom{l}{j} \binom{k + j'}{l} \cdot \binom{l}{j'}\\
                                                 & ~\le~ 2^{2l} \left(\sum_{j=0}^{l} \binom{k + j}{l} \right) \cdot \left(\sum_{j'= 0}^{l} \binom{k + j'}{l}\right) ~\le~ 2^{2k+4l+2},
  \end{align*}
  we are done.
\end{proof}

Let $\Bee$ and $\Bee^{\Delta}$ stand for the $B$ operator for general
$\gamma$-HDX and the complete complex, respectively.

\begin{lemma}\label{lemma:rectangular_swap}
  Suppose $X(\le d)$ is a $\gamma$-HDX with $d \ge k + l$. For
  $\epsilon \in (0,1)$, if $\gamma \le \epsilon^2 \left(64
  k^{k+2}\ell^2 2^{2k+4l+2} \right)^{-1}$, then the second largest
  singular value $\sigma_2(\Bee)$ of $\Bee$ is
  \[
  \sigma_2(\Bee) ~\le~ \epsilon^2.
  \]
  Furthermore, the second largest non-trivial eigenvalue
  $\lambda(\Aye_{k,l})$ of the walk matrix of
  $K(n,k,l)$ is
  \[
  \lambda(\Aye_{k,l}) ~\le~ \epsilon.
  \]
\end{lemma}

\begin{proof}
  The proof follows the same strategy of~\cref{theo:sq_swap_gap},
  namely, we first consider $\Bee^{\Delta}$ and show that $\sum_{j
    ,j'= 0}^{l} \alpha_{k,l,j,j'} \lambda_{k,l,j,j'}(i) =
  0$. Using~\cref{cor:complete_rec_kneser}, we deduce that
  \[
  \left\lvert \sum_{j ,j'= 0}^{l} \alpha_{k,l,j,j'} \lambda_{k,l,j,j'}(i) \right\vert ~=~ O_{k,l}\left(\frac{1}{n^2}\right)
  \]
  for $i \in [k]$ where in this range each $C^k_i$ is not the trivial
  approximate eigenspace (associated with eigenvalue $1$).  Since
  $\alpha_{k,l,j,j'}$ and $\lambda_{k,l,j,j'}(i)$ do not depend on $n$ and
  $n$ is arbitrary, the LHS above is actually zero. Then our choice of
  $\gamma$~\cref{lemma:asd_general_decomposition} gives
  \[
  \max_{f \in C^k \colon f \perp 1, \norm{f}=1} ~\left\lvert \ip{\Bee f}{f} \right\vert ~\le~ \max_{i \in [k]} ~\left\lvert \sum_{j ,j'= 0}^{l} \alpha_{k,l,j,j'} \lambda_{k,l,j,j'}(i) \right\vert ~+~ \epsilon^2 ~=~ \epsilon^2.
  \]
\end{proof}

Now the proof of~\cref{thm:swap-eig-bd} follows. For convenience, we
restate it.
\begin{theorem}[Rectangular Swap Walk Spectral Bound (restatement of~\cref{thm:swap-eig-bd})]
  Suppose $X(\le d)$ is a $\gamma$-HDX with $d \ge k + l$ and $k \le
  l$. For $\sigma \in (0,1)$, if $\gamma \le \sigma^2 \cdot \left(64
  k^{k+2}\ell^2 2^{2k+4l+2} \right)^{-1}$, then the largest
  non-trivial singular value $\sigma_2(\rfswap{k}{l})$ of the swap
  operator $\rfswap{k}{l}$ is
  \[
  \sigma_2(\rfswap{k}{l}) ~\le~ \sigma.
  \]
\end{theorem}

\begin{proof}
  Follows directly from~\cref{lemma:rectangular_swap}.
\end{proof}

\subsection{Bipartite Kneser Graphs - Complete Complex}\label{sec:bipartite_kneser_complete}

Now we determine the spectrum of the \textit{complete bipartite
  Kneser} graph $K(n,k,l)$. More precisely, we prove the following.
\begin{lemma}[Bipartite Kneser Spectrum (restatement of~\cref{lemma:rec_kneser_spectrum})]
  The non-zero eigenvalues of the normalized walk operator of
  $K(n,k,l)$ are $\pm \lambda_i$ where
  \[
  \lambda_i ~=~ \frac{ {n-k-i \choose l-i} {n-l-i \choose k-i}}{{n-k \choose l}{n-l \choose k}},
  \]
  for $i = 0,\dots,\min(k,l)$.  
\end{lemma}

Henceforth, set $X = \Delta_d(n)$. To
prove~\cref{lemma:rec_kneser_spectrum} we work with the natural
rectangular matrix associated with $K(n,k,l)$, namely, the matrix $W
\in \mathbb{R}^{X(k) \times X(l)}$ such that
\[
\Www(\ess,\tee) ~=~ \mathds{1}_{[\ess\cap \tee = \emptyset]}
\]
for every $\ess \in X(k)$ and $\tee \in X(l)$.

Observe that the entries of $\Www\Www^\top$ and $\Www^\top \Www$ only
depend on the size of the intersection of the sets indexing the row
and columns. Hence, these matrices belong to the Johnson
scheme~\cite{GodsilM15} $J(n,k)$ and $J(n,l)$, respectively. Moreover,
the left and right singular vectors of $W$ are eigenvectors of these
schemes.

We adopt the eigenvectors used in Filmus' work~\cite{Filmus16}, i.e.,
natural basis vectors coming from some irreducible representation of
$S_n$ (see~\cite{sagan13}). First we introduce some notation. Let $\mu
= (n-i,i)$ be a partition of $n$ and let $\tau_{\mu}$ be a standard
tableau of shape $\mu$. Suppose the first row $\tau_{\mu}$ contains
$a_1 < \cdots < a_{n-i}$ whereas the second contains $b_1 < \cdots <
b_i$. To $\tau_{\mu}$ we associate the function $\phi_{\tau_{\mu}} \in
\mathbb{R}^{{[n] \choose k}}$ as follows
\[
 \phi_{\tau_{\mu}} ~=~ (\mathds{1}_{a_1} - \mathds{1}_{b_1})\dots (\mathds{1}_{a_i}-\mathds{1}_{b_i}),
\]
where $\mathds{1}_a \in \mathbb{R}^{{n \choose k}}$ is the
containment indicator of element $a$, i.e., $\mathds{1}_a(\ess) =
1$ if and only if $a \in \ess$. Filmus proved that
\[
\left\{ \phi_{\tau_{\mu}}~\vert~0 \le i \le k, \mu \vdash (n-i,i), \tau_{\mu} \textrm{ standard}\right\}
\]
is an eigenbasis of $\mathcal{J}(n,k)$. We abuse the notation by considering $\phi_{\tau_{\mu}}$
as both a function in $\mathbb{R}^{{n \choose k}}$ and $\mathbb{R}^{{n \choose l}}$ as long
as these functions are well defined. 

\begin{claim}
  If $\mu = (n-i,i)$ and $k,l\ge i$, then
  \[
  \Www \phi_{\tau_{\mu}} ~=~ (-1)^{i} \cdot {n-k-i \choose l-i} \cdot  \phi_{\tau_{\mu}}.
  \]
\end{claim}

\begin{proof}
  We follow a similar strategy of Filmus. For convenience suppose
  $\phi_{\tau_{\mu}} = (\mathds{1}_{1} - \mathds{1}_{2})\dots
  (\mathds{1}_{2i-1}-\mathds{1}_{2i})$. For $i=0$ the claim follows
  immediately so assume $i \ge 1$. Consider $\left(\Www
   \phi_{\tau_{\mu}}\right)(\ess)$ where $\ess \in {[n] \choose
    k}$. Note that
  \[
  \left(\Www  \phi_{\tau_{\mu}}\right)(\ess) ~=~ \sum_{\tee \in Y \colon \ess \cap \tee = \emptyset}  \phi_{\tau_{\mu}}(\tee).
  \]
  If $2j-1,2j \in \ess$ for some $j \in [i]$, then $2j-1,2j \not\in \tee$ so $\phi_{\tau_{\mu}}(\ess) = 0 = \left(W \phi_{\tau_{\mu}}\right)(\ess)$.
  If $2j-1,2j \not\in \ess$ for some $j \in [i]$, for each $\tee$ adjacent to $\ess$ there four cases: $2j-1,2j \in \tee$, $2j-1,2j \not\in \tee$, 
  $2j-1 \in \tee$ and $2j \not\in \tee$ or vice-versa. The first two cases yield $\phi_{\tau_{\mu}}(\tee) =0$ while the last two cases cancel each
  other in the summation and again $\phi_{\tau_{\mu}}(\ess) = 0 = \left(W \phi_{\tau_{\mu}}\right)(\ess)$. Now suppose that $\ess$ contains exactly
  one element of each pair $2j-1,2j$. For any adjacent $\tee$ to yield $\phi_{\tau_{\mu}}(\tee) \ne 0$, $\tee$ must contain $[2i]\setminus \ess$.
  Since there are ${n-k-i \choose l-i}$ such possibilities for $\tee$ we obtain
  \[
  \Www  \phi_{\tau_{\mu}} ~=~ (-1)^{i} \cdot {n-k-i \choose l-i} \cdot  \phi_{\tau_{\mu}},
  \]
  where the sign $(-1)^{i}$ follows from the product of the signs of each the $i$ pairs and the fact that $\ess$ and $\tee$ partition the elements
  in each pair.
\end{proof}

Since we are working with singular vectors, we need to be careful with
their normalization when deriving the singular values. We stress that
the norm of $\phi_{\tau_{\mu}}$ depends on the space where
$\phi_{\tau_{\mu}}$ lies.

\begin{claim}
  If $\mu = (n-i,i)$ and $\phi_{\tau_{\mu}} \in \mathbb{R}^{{n \choose k}}$, then
  \[
  \norm{ \phi_{\tau_{\mu}}}_2 ~=~ \sqrt{2^i {n-2i \choose k - i}}.
  \]
\end{claim}

\begin{proof}
  Since $\phi_{\tau_{\mu}}$ assumes values in $\{-1,0,1\}$ so its
  enough to count the number of sets $\ess \in {[n] \choose k}$ such
  that $\phi_{\tau_{\mu}}(\ess) \ne 0$. To have
  $\phi_{\tau_{\mu}}(\ess) \ne 0$, $\ess$ must contain exactly one
  element in each pair and the remaining $k-i$ elements of $\ess$ can
  be chosen arbitrarily among the elements avoiding the $2i$ elements
  appearing in the indicators defining $\phi_{\tau_{\mu}}$.
\end{proof}

Now the singular values of $W$ follow.
\begin{corollary}[Singular Values]
  The singular values of $W$ are
  \[
  \sigma_i ~=~ {n-k-i \choose l-i} \cdot \frac{\norm{ \phi_{\tau_{\mu}}^k}_2}{\norm{ \phi_{\tau_{\mu}}^l}_2},
  \]
  for $i=0,\dots,\min(k,l)$.
\end{corollary}
Note that for $k = l$ we recover the well know result
of~\cref{fact:kneser_singular_values}.

Finally we compute the eigenvalues of the bipartite graph
$K(n,k,l)$. Let $\Aye_{n,k,l}$ be its normalized adjacency matrix,
i.e.,
\[
\Aye_{n,k,l} ~=~ \begin{pmatrix}
                0   & \frac{1}{{n-k \choose l}} \Www\\
                \frac{1}{{n-l \choose k}} \Www^\top &0
              \end{pmatrix}.
\]

\begin{lemma}[Bipartite Kneser Spectrum (restatement of~\cref{lemma:rec_kneser_spectrum})]
  The non-zero eigenvalues of the normalized walk operator of
  $K(n,k,l)$ are $\pm \lambda_i$ where
  \[
  \lambda_i ~=~ \frac{ {n-k-i \choose l-i} {n-l-i \choose k-i}}{{n-k \choose l}{n-l \choose k}},
  \]
  for $i = 0,\dots,\min(k,l)$.  
\end{lemma}

\begin{proof}
  Since the spectrum of a bipartite graph is symmetric around zero, it
  is enough to compute the eigenvalues of $A_{n,k,l}^2$. Set $\alpha =
  1/{n-k \choose l}{n-l \choose k}$. Moreover, we consider $\alpha \cdot
  \Www\Www^\top$ since $\alpha \cdot \Www^\top\Www$ has the same non-zero
  eigenvalues. The non-zero eigenvalues of $\alpha \cdot \Www\Www^\top$ are
  \[
  \lambda_i ~=~ \frac{ {n-k-i \choose l-i} {n-l-i \choose k-i}}{{n-k \choose l}{n-l \choose k}},
  \]
  for $i = 0,\dots,\min(k,l)$.
\end{proof}


%% file: hdx_brs.tex
In the following, we will show that $k$-CSP instances $\Ins$ whose
constraint complex $X_\Ins(\le k)$ is a suitable expander admit an efficient
approximation algorithm. We will assume throughout that $X_\Ins(1) = [n]$, and
drop the subscript $\Ins$.

This was shown for 2-CSPs in \cite{BarakRS11}. In 
extending this result to $k$-CSPs
we will rely on a central Lemma of their paper. Before, we explain our algorithm we give a basic outline of our idea:

We will work with the SDP relaxation for the $k$-CSP
problem given by $L$-levels of SoS hierarchy, as defined in \cref{ssec:sos-relax} (for $L$ to be
specified later). This will give us an $L$-local PSD ensemble $\set{\rv Y_1, \ldots, \rv Y_n}$, which attains some value 
$\SDP(\Ins) \ge \OPT(\Ins)$. Since $\set*{\rv Y_1, \ldots, \rv Y_n}$, is a local PSD ensemble,
and not necessarily a probability distribution, we cannot sample from it directly. Nevertheless, since $\set*{\rv Y_j}$ will be actual probability distributions for all $j \in [n]$, one can independently sample $\assn_j \sim \set*{\rv Y_j}$ and
use $\assn = (\assn_1, \ldots, \assn_n)$ as the assignment for the $k$-CSP instance $\Ins$.

Unfortunately, while we know that the local distributions $\set{\rv Y_\aye}_{\aye \in X(k)}$ induced by 
$\set*{\rv Y_1, \ldots, \rv Y_n}$ will satisfy the constraints of $\Ins$ with good probability, \ie
\[ \ExpOp_{\aye \sim \Pi_k}{\Ex{\set{\rv Y_\aye}}{ \one[ \underbrace{\rv Y_a \textrm{ satisfies the constraint on } \aye}_{\Longleftrightarrow \rv Y_\aye \in~\mathcal C_\aye}]}} = \SDP(\Ins) \ge \OPT(\Ins),\] 
this might not be the case for the assignment $\assn$ sampled as before. It might be that the random variables
$\rv Y_{a_1}, \ldots, \rv Y_{a_k}$ are highly correlated for $\aye \in X(k)$, \ie~
\[ \ExpOp_{\aye \sim \Pi_k}{ \norm{\set{\rv Y_\aye} - \set{\rv Y_{a_1}} \cdots \set{\rv Y_{a_k}}}_1  }\]
is large.
One strategy employed by \cite{BarakRS11} to ensure that the quantity above is small, is
making the local PSD ensemble $\set{\rv Y_1, \ldots, \rv Y_n}$ 
be consistent with a randomly sampled partial assignment for a small subset of variables (q.v.~\cref{ssec:sos-relax}). We will show that this strategy is succesful if $X(\le k)$ is a 
$\gamma$-HDX (for $\gamma$ sufficiently small). Our final algorithm will be the following,

\begin{algorithm}{Propagation Rounding Algorithm}{An $L$-local PSD ensemble $\set{\rv Y_1, \ldots, \rv Y_n}$ and some 
    distribution $\Pi$ on $X(\le \ell)$.}{A random assignment $\assn: [n] \to [q]$.}\label{algo:prop-rd}
    \begin{enumerate}
        \item Choose $m \in \set*{1, \ldots, L/\ell}$ uniformly at random.
        \item Independently sample $m$ $\ell$-faces, $\ess_j \sim \Pi$ for $j = 1, \ldots, m$.
        \item Write $S = \bigcup_{j = 1}^m \ess_j$, for the set of the seed vertices.
        \item Sample assignment $\assn_{S}: S \to [q]$ according to the local distribution, $\set{\rv Y_{S}}$.
        \item Set $\rv Y' = \set{\rv Y_1, \ldots \rv Y_n | \rv Y_S = \assn_S}$, i.e.~the local ensemble
            $\rv Y$ conditioned on agreeing with $\assn_S$.
        \item For all $j \in [n]$, sample independently $\assn_j \sim \set{\rv Y'_j}$.
        \item Output $\assn = (\assn_1, \ldots, \assn_n)$.
    \end{enumerate}
\end{algorithm}
In our setting, we will apply \cref{algo:prop-rd} with the distribution $\Pi_k$ and the $L$-local PSD ensemble $\set{\rv Y_1, \ldots, \rv Y_n}$.
Notice that in expectation, the marginals of $\rv Y'$ on faces $\aye \in X(k)$ -- which are actual distributions -- will agree with the marginals of $\rv Y$, i.e.~$\ExpOp_{S, \eta_S} \ExpOp \rv Y'_\aye =  \ExpOp \rv Y_\aye$. In particular, the
approximation quality of \cref{algo:prop-rd} will depend on the average correlation of
$\rv Y'_{a_1}, \ldots, \rv Y'_{a_k}$ on the constraints $\aye \in X(k)$, where $\rv Y'$ is the local PSD ensemble
obtained at the end of the first phase of \cref{algo:prop-rd}.

In the case where $k=2$, the following is known \vnote{Madhur might fix this}
\begin{theorem}[Theorem 5.6 from \cite{BarakRS11}]\label{thm:brs_exp}
    Suppose a weighted undirected graph $G = ([n], E, \Pi_2)$ and an $L$-local PSD ensemble
    $Y = \set*{\rv Y_1, \ldots, \rv Y_n}$ are given. There exists absolute constants $c \ge 0$ and $C \ge 0$ satisfying the following: If $L \ge c \cdot \frac{q}{\ee^4}$, $\supp( \rv Y_i ) \le q$ for all $i \in V$, and $\lambda_2(G) \le C \cdot \ee^2/ q^2$ then 
    we have
    \[ \ExpOp_{ \set{i,j} \sim \Pi_2 }{ \norm{\set{\rv Y'_i, \rv Y'_j} - \set{\rv Y'_i}\set{\rv Y'_j}}_1} \le \ee,\]
    where $\rv Y'$ is as defined in \cref{algo:prop-rd} on the input of $\set{\rv Y_1, \ldots, \rv Y_n}$ and $\Pi_1$. 
\end{theorem}

To approximate $k$-CSPs well, we will show the following generalization of \cref{thm:brs_exp} for $k$-CSP instances
$\Ins$, whose constraint complex $X(\le k)$ is $\gamma$-HDX, for $\gamma$ sufficiently small.

\begin{theorem}\label{thm:brs_hdx}
    Suppose a simplicial complex $X(\le k)$ with $X(1) = [n]$ and an $L$-local PSD ensemble
    $\rv Y = \set{\rv Y_1, \ldots, \rv Y_n}$ are given. 
    
    There exists some universal constants $c' \ge 0$ and $C' \ge 0$ satisfying the following:
    If $L \ge c' \cdot (q^{k} \cdot k^5/\ee^4)$, $\supp(\rv Y_j) \le q$
    for all $j \in [n]$, and $X$ is a $\gamma$-HDX for $\gamma \le C' \cdot \ee^4/(k^{8 + k} \cdot 2^{6k} \cdot q^{2k})$. Then, we have
    \begin{equation}
        \ExpOp_{\aye \sim \Pi_k}{ \norm{\set{\rv Y'_{\aye}} - \set*{\rv Y'_{a_1}}\cdots \set*{\rv Y'_{a_k}}}_1 } \le \ee,
        \label{eq:brs_hdx}
    \end{equation}
    where $\rv Y'$ is as defined in \cref{algo:prop-rd} on the input of $\set{\rv Y_1, \ldots, \rv Y_n}$ and $\Pi_k$. 
\end{theorem}
Indeed, using \cref{thm:brs_hdx}, it will be straightforward to prove the following,
\begin{corollary}\label{cor:alg-cooked}
    Suppose $\Ins$ is a $q$-ary $k$-CSP instance 
    whose constraint complex $X(\le k)$ is a $\gamma$-HDX.
    
    There exists absolute constants $C' \ge 0$ and $c' \ge 0$, satisfying the following:
    If $\gamma \le C' \cdot \ee^4/(k^{8 + k} \cdot 2^{6k} \cdot q^{2k})$, there is an algorithm that runs in time
    $n^{O(k^5 \cdot q^{2k} \cdot \ee^{-4})}$ based on
    $(\frac{c' \cdot k^5 \cdot q^{k}}{\ee^4})$-levels of SoS-hierarchy and \cref{algo:prop-rd} that outputs a random
    assignment $\assn: [n] \to [q]$ that in expectation ensures $\SAT_\Ins(\assn) =
    \OPT(\Ins) - \ee$. 
\end{corollary}
\begin{proof}[Proof of \cref{cor:alg-cooked}]
    The algorithm will just run \cref{algo:prop-rd} on the local PSD-ensemble $\set{\rv Y_1, \ldots, \rv Y_n}$ 
    given by the SDP relaxation of $\Ins$ 
    strengthened by $L = c' \cdot \frac{k^5 \cdot q^{2k}}{\ee^4}$-levels of SoS-hierarchy and $\Pi_k$ -- where
    $c' \ge 0$ is the constant from \cref{thm:brs_hdx}. $\rv Y$ satisfies,
    \begin{equation}
        \SDP(\Ins) = \Ex{\aye \sim \Pi_k}{\Ex{\set{\rv Y_\aye}}{\one[\rv Y_\aye \in \Cc_\aye]}} \ge \OPT(\Ins). \label{eq:ysatisfiesgood} 
    \end{equation}
    Let $S$, $\assn_S$, and $\rv Y'$ be defined as in \cref{algo:prop-rd} 
    on the input of $\rv Y$ and $\Pi_k$. Since the conditioning done on $\set{\rv Y'}$ is consistent with
    the local distribution, by law of total expectation and \cref{eq:ysatisfiesgood} one has
    \begin{equation}
        \ExpOp_{S}\ExpOp_{ \assn_S \sim \set{\rv Y_S}}{ \ExpOp_{\aye \sim \Pi_k}{\Ex{\set{\rv Y'_\aye}}{ \one[ \rv Y'_\aye \in \Cc_\aye]}}} = \SDP(\Ins) \ge \OPT(\Ins).\label{eq:yprimesatgood}
    \end{equation}
    
    By~\cref{thm:brs_hdx} we know that
    \begin{equation}
        \ExpOp_{S}{\ExpOp_{\assn_S \sim \set{\rv Y_S}}{\ExpOp_{\aye \sim \Pi_k}{ \norm{\set{\rv Y'_\aye} - \set{\rv Y'_{a_1}}\cdots \set{\rv Y'_{a_k}}}_1}}} \le \ee\label{eq:distsclose}
    \end{equation}
    Now, the fraction of constraints satisfied by the algorithm in expectation is
    \[ \ExpOp_{\assn}[\SAT_\Ins(\assn)]= \ExpOp_{S}{\ExpOp_{\assn_S \sim \set{\rv Y_S}}{\ExpOp_{\aye \sim \Pi_k}{\Ex{(\assn_1, \ldots, \assn_n) \sim \set{\rv Y_1'} \cdots \set{\rv Y_n'}}{\one[ \assn|_\aye \in \Cc_\aye]}}}}.\]
    By using \cref{eq:distsclose}, we can obtain 
    \[ \ExpOp_{\assn}[\SAT_\Ins(\assn)] \ge \Ex{S}{\ExpOp_{\assn_S \sim \set{\rv Y_S}}{ \ExpOp_{\set{\rv Y_\aye}}{ \one[ \rv Y'_\aye \textrm{ satisfies the constraint on } \aye]}}} - \ee. \]
    Using \cref{eq:yprimesatgood}, we can conclude
    \[ \ExpOp_{\assn}[\SAT_\Ins(\assn)] \ge \SDP(\Ins) - \ee = \OPT(\Ins) - \ee.\]
\end{proof}
Our proof of \cref{thm:brs_hdx} will hinge on the fact that we can upper-bound the expected correlation of a face of large cardinality $\ell$, in terms of expected correlation over faces of smaller cardinality and expected correlations along the edges of a swap
graph. The swap graph here is defined as a weighted graph $G_{\ell_1, \ell_2} = \parens*{X(\ell_1) \sqcup X(\ell_2), E(\ell_1, \ell_2), w_{\ell_1, \ell_2}}$, where 
\[E(\ell_1, \ell_2) = \set*{ \set*{\aye, \bee} : \aye \in X(\ell_1), \bee \in X(\ell_2), \textrm{ and } \aye \sqcup \bee \in X(\ell_1 + \ell_2) }.\]
We will assume $\ell_1 \ge \ell_2$, and if $\ell_1 = \ell_2$ we are going to identify the two copies of every vertex. We will endow $E(\ell_1, \ell_2)$ with the weight function,
\[ w_{\ell_1, \ell_2}(\aye, \bee) =  \frac{\Pi_{\ell_1 + \ell_2}(\aye \sqcup \bee) }{ \binom{\ell_1 + \ell_2}{\ell_1} },\]
which can easily be verified to be a probability distribution on $E(\ell_1, \ell_2)$
Notice that in the case where $\ell_1 \ne \ell_2$ the random walk matrix of $G_{\ell_1, \ell_2}$ is given by
\[ \Aye_{\ell_1, \ell_2}= \begin{pmatrix} 0 & \rfswap{\ell_1}{\ell_2}\\
\rfswap{\ell_1}{\ell_2}^\dagger & 0\end{pmatrix},\]
and if $\ell_1 = \ell_2$ we have $\Aye_{\ell_1, \ell_1} = \rfswap{\ell_1}{\ell_1}$.
The stationary distribution of $\Aye_{\ell_1, \ell_2}$ is $\Pi_{\ell_1, \ell_2}$ defined by,
\begin{equation}
    \Pi_{\ell_1, \ell_2}(\bee) = \one[\bee \in X(\ell_1)] \cdot \frac{1}{2} \cdot \Pi_{\ell_1}(\bee) + \one[\bee\in X({\ell_2})] \cdot \frac{1}{2} \cdot \Pi_{\ell_2}(\bee).
    \label{eq:stat_meas}
\end{equation}
When we write an expectation of $\eff(\bullet, \bullet)$ over the edges in $E(\ell_1, \ell_2)$ with respect to
$w_{\ell_1, \ell_2}$, it is important to note,
\begin{equation}
    \Ex{\set*{\ess, \tee} \sim w_{\ell_1, \ell_2}}{ \eff(\ess, \tee) } = \sum_{\set{\ess,\tee} \in E(\ell_1, \ell_2)} \frac{1}{\binom{\ell_1+ \ell_2}{\ell_1}} \cdot \eff(\ess, \tee) \cdot \Pi_{\ell_1 + \ell_2}(\ess \sqcup \tee) = \frac{1}{\binom{\ell}{\ell_1}}\Ex{\aye \sim \Pi_k}{\sum_{\set*{\ess, \tee}\sim \aye} \eff(\ess, \tee) }\label{eq:edg_exp},
\end{equation}
where sum within the expectation in the RHS runs over the $\binom{\ell_1 + \ell_2}{\ell_1}$ possible ways of splitting $\aye$ into $\ess \sqcup \tee$ such that $\ess \in X(\ell_1)$
and $\tee \in X(\ell_2)$.
    When we are speaking about the spectral expansion of $G_{\ell_1, \ell_2}$, we will be speaking with regards to $\lambda_2(G_{\ell_1, \ell_2})$ and not with regards to $\sigma_2(G_{\ell_1, \ell_2})$.
\begin{remark}\label{rem:sing-to-eig} By simple linear algebra, we have
    \[ \lambda_2(G_{\ell_1, \ell_2}) := \lambda_2(\Aye_{\ell_1, \ell_2}) \le \sigma_2(\rfswap{\ell_1}{\ell_2}),\]
where we employ the notation $\lambda_2(\Emm)$ to denote the second largest eigenvalue (signed) of the matrix $\Emm$. 
\end{remark}

With this, we will show
\begin{lemma}[Glorified Triangle Inequality]\label{lem:split_rectangular}
    For a simplicial complex $X(\le k)$,  $\ell_1 \ge \ell_2 \ge 0$, $\ell = \ell_1+ \ell_2$, $\ell \le  k$, and 
    an $\ell$-local ensemble $\set*{\rv Y_1, \ldots, \rv Y_n}$, one has
    \begin{multline}
        \Ex{\aye \in \Pi_\ell}{\norm{\set*{\rv Y_\aye} - \prod_{i = 1}^{\ell} \set*{\rv Y_{a_i}}}_1} \le \Ex{\set{\ess, \tee} \sim w_{\ell_1, \ell_2}}{ \norm{\set*{\rv Y_\ess, \rv Y_\tee } - \set*{\rv Y_\ess}\set*{\rv Y_\tee}}_1}\\ + \Ex{\ess \sim \Pi_{\ell_1}}{
            \norm{\set*{\rv Y_\ess} - \prod_{i = 1}^{\ell_1} \set*{\rv Y_{s_i}}}_1}
        + \Ex{\tee \sim \Pi_{\ell_2}}{\norm{\set{\rv Y_\tee} - \prod_{i = 1}^{\ell_2} \set{\rv Y_{t_i}}}_1}\label{eq:swap_split}
    \end{multline}
\end{lemma}
One useful observation, is that by using \cref{lem:split_rectangular} repeatedly, we can reduce the problem of bounding
$\ExpOp_{\aye \in \Pi_\ell}{ \norm{\set{\rv Y_\aye} - \prod_{i = 1}^\ell \set{\rv Y_{a_i}}}_1 }$ to a problem of bounding
\[
\ExpOp_{\set{\ess, \tee} \sim w_{\ell_1, \ell_2}}{\norm{\set{\rv Y_{\ess}, \rv Y_\tee} - \set{\rv Y_\ess}\set{\rv Y_\tee}}_1},
\]
for $\ell_1 + \ell_2 \le k$. Though it is not a direct implication, it
is heavily suggested by \cref{fact:set-ensemble} and
\cref{thm:brs_exp}, that if $G_{\ell_1, \ell_2}$ is a good spectral
expander, after an application of \cref{algo:prop-rd} with our chosen
parameters, we should be able to bound these expressions. Using a key
lemma used from \cite{BarakRS11}, we will prove that this is indeed
the case. The only thing we need to make sure after this point, is
that the second eigenvalue $\lambda_2(G_{\ell_1, \ell_2})$ of the swap
graphs $G_{\ell_1, \ell_2}$ we will be using are small enough for our
purposes. Indeed, our choice of $\gamma$ in \cref{thm:brs_hdx} and
\cref{cor:alg-cooked} is to make sure that the bound we get on
$\lambda_2(G_{\ell_1, \ell_2})$ from \cref{thm:swap-eig-bd} (together
with \cref{rem:sing-to-eig}) is good enough for our purposes.

\subsection{Breaking Correlations for Expanding CSPs: Proof of \cref{thm:brs_hdx} }

Throughout this section, we will use the somewhat non-standard definition of variance introduced in \cite{BarakRS11},
\[ \Var{\rv Y_\aye} = \sum_{\assn \in [q]^\aye}\Var{\one[ \rv Y_\aye = \assn ] }.\]

We will use the following central lemma from \cite{BarakRS11} in our proof of \cref{thm:brs_hdx}:
\begin{lemma}[Lemma 5.4 from \cite{BarakRS11}]\label{lem:brs-dec}
    Let $G = (V, E, \Pi_2)$ be a weighted graph, $\set{\rv Y_1, \ldots, \rv Y_n}$ a local PSD ensemble, where we have
    $\supp(\rv Y_i) \le q$ for every $i \in V$,  and $q \ge 0$. Suppose $\ee \ge 0$ is a lower bound on 
    the expected statistical difference between independent
    and correlated sampling along the edges,\ie
    \[ \ee \le \Ex{\set{i,j} \sim \Pi_2}{\norm{\set{\rv Y_{ij} } - \set{\rv Y_i}\set{\rv Y_j}  }_1}.\]
    There exists absolute constants $c_0 \ge 0 $ and $c_1 \ge 0$ that satisfy the following:
    If $\lambda_2(G) \le c_0 \cdot \frac{\ee^2}{q^2}$. Then, conditioning on a random vertex decreases the variances,
    \[ \ExpOp_{i \sim \Pi_1}\ExpOp_{j \sim \Pi_2}{\Ex{ \set{\rv Y_j} }{ \Var{\rv Y_i \mid \rv Y_j } } } \le \Ex{i \sim \Pi_1}{\Var{\rv Y_i}} - c_1 \cdot \frac{\ee^2}{q^2}.\]
\end{lemma}

For our applications, we will be instantiating \cref{lem:brs-dec} with $G_{\ell_1, \ell_2}$ as $G$; and with the
local PSD ensemble $\set{\rv Y_\aye}_{\aye \in X}$ that is obtained from $\set{\rv Y_1, \ldots, \rv Y_n}$ 
(q.v.~\cref{fact:set-ensemble}). For convenience, we will write the concrete instance of the Lemma that we will use,
\begin{corollary}\label{cor:brs-explicit}
    Let $\ell_1 \ge \ell_2 \ge 0$ satisfying $\ell_1 + \ell_2 \le k$ be given parameters, and let
    $G_{\ell_1, \ell_2}$ be the swap graph defined for a $\gamma$-HDX $X(\le k)$. 
    Let $\set{\rv Y_\aye}_{\aye \in X}$ be a local PSD ensemble; satisfying $\supp(Y_\aye) \le q^k$
    for every $\aye \in X(\ell_1) \cup X(\ell_2)$ for some $q \ge 0$. Suppose $\ee \ge 0$ satisfies,
    \[ \frac{\ee}{4k} \le \Ex{\set{\ess, \tee} \in w_{\ell_1, \ell_2}}{\norm{\set{\rv Y_{\ess \sqcup \tee}} - \set{\rv Y_\ess}\set{\rv Y_\tee}}_1}.\]
    There exists absolute constants $c_0 \ge 0$ and $c_2 \ge 0$ that satisfy the following: If $\lambda_2(G) \le c_0 \cdot (\ee/(4k\cdot q^{k}))^2$. Then, conditioning on a random face $\aye \sim \Pi_{\ell_1, \ell_2}$ decreases the variances, i.e.~
    \begin{eqnarray*}
        2 \cdot \Ex{\aye, \bee \sim \Pi_{\ell_1, \ell_2}^2}{\Ex{\set{\rv Y_\aye}}{\Var{\rv Y_\bee \mid \rv Y_\aye}}} & = &
        \Ex{\aye \in \Pi_{\ell_1, \ell_2}}{\Ex{\ess \sim \Pi_{\ell_1}}{\Var{\rv Y_\ess \mid \rv Y_\aye}} 
        + \Ex{\tee \sim \Pi_{\ell_2}}{\Var{\rv Y_\tee \mid \rv Y_\aye}}},\\
        & \le & \Ex{\ess \sim \Pi_{\ell_1}}{\Var{\rv Y_\ess}} + \Ex{\tee \sim \Pi_{\ell_2}}{\Var{\rv Y_\tee}} - c_2 \cdot \frac{\ee^2}{16 \cdot k^2 \cdot q^{2k}}.
    \end{eqnarray*}
\end{corollary}
Here, it can be verified that the expansion criterion presupposed by \cref{lem:brs-dec} is satisfied by 
\cref{cor:brs-explicit} by \cref{thm:swap-eig-bd}. The constant $c_2$ satisfies $c_2 = 2 \cdot c_1$.

\begin{proof}[Proof of \cref{thm:brs_hdx}]
    We will follow the same proof strategy in \cite{BarakRS11}, and extend their arguments for $k$-CSPs.
    
    Write $\Pi_k^m$ for the distribution of the 
    random set that is obtained in steps (2)-(3) of \cref{algo:prop-rd} with $\Pi = \Pi_k$, i.e.~$S \sim \Pi_k^m$ is sampled by
    \begin{enumerate}
        \item independently sampling $m$ $k$-faces $\ess_j \sim \Pi_k$ for $j = 1, \ldots, m$.
        \item outputting $S = \bigcup_{j = 1}^m \ess_j$.
    \end{enumerate}
    First, for $m \in [L/k]$ we will define
    \[ \ee_m = \ExpOp_{S \sim \Pi_k^m}{\ExpOp_{\set{\rv Y_S}}{\Ex{\aye \sim \Pi_k}{\norm{\set{\rv Y_\aye \mid \rv Y_S} - \prod_{j = 1}^k \set{\rv Y_{a_j} \mid \rv Y_S}}_1}}},\]
    which will measure the average correlation along $X(k)$ after conditioning on $m$ $k$-faces.
    Notice that our goal is ensuring,
    \[ \ExpOp_{m \sim [L/k]}{\ee_m} \le \ee\]
    where $m$ is sampled uniformly at random.

    To help us with this goal, we will define a potential function
    \begin{equation}
        \Phi_m = \ExpOp_{i \sim [k]}{\ExpOp_{S \sim \Pi_k^m}{\ExpOp_{\set{\rv Y_S}}{\ExpOp_{\aye \sim \Pi_i}{\Var{\rv Y_\aye \mid \rv Y_S}}}}}.\label{eq:pot-def}
    \end{equation}
    where $i$ is sampled uniformly at random.  Observe that $\Phi_m$ always satisfies $0 \le \Phi_m \le 1$. Using this, we will try to bound the fraction of indices $m \in [L/k]$ such that $\ee_m$ is large, \ie~say $\ee_m \ge \ee/2$. To this end assume $\ee_m \ge \ee/2$, i.e.~we have
    \begin{equation}
        \ExpOp_{S \sim \Pi_k^m}{\ExpOp_{\set{\rv Y_S}}{\Ex{\aye \sim \Pi_k}{\norm{\set{\rv Y_\aye \mid \rv Y_S} - \prod_{i = 1}^k \set{\rv Y_{a_1} | \rv Y_S}}_1}}} \ge \frac{\ee}{2}.\label{eq:assume-large}
    \end{equation}

    We will use \cref{lem:split_rectangular} in the following way: Let $\tree$ be any binary tree with $k$ leaves. We will label each of the vertices $v \in \tree$ with the number of leaves of the subtree rooted at $v$. Notice that this ensures that
    \begin{enumerate}
        \item the root vertex of $\tree$ has the label $k$,
        \item for any vertex $v \in \tree$ with label $\ell$, the label $\ell_1$ of the left child of $v$ and the
            label $\ell_2$ of the right child of $v$ add up to $k$, i.e.~$\ell_1 + \ell_2 = k$,
        \item every vertex $v \in \tree$ with the label 1 is a leaf.
    \end{enumerate}
    We write $J(\tree)$ for the set of labels $\ell$ of the internal nodes of $\tree$, note $|J(\tree)| \le k$. We will use the notation $\ell_1$ (resp.~$\ell_2$) to refer to the label of the left (resp.~right) of a vertex $v \in \tree$ with the label $\ell$.

    By applying \cref{lem:split_rectangular}, we obtain that for any local PSD ensemble $\rv Z$ one has
    \[\Ex{\aye \sim \Pi_k}{ \norm{ \set{\rv Z_\aye} - \prod_{i = 1}^k \set{\rv Z_{\aye_i}}}_1}
        \le \sum_{\ell \in J(\tree)} \Ex{ \set{\tee_1, \tee_2} \in w_{\ell_1, \ell_2}}{
            \norm{\set{\rv Z_{\tee_1 \sqcup \tee_2}} - \set{\rv Z_{\tee_1} } \set{\rv Z_{\tee_2} }}_1}.\]
        Now, by plugging this in \cref{eq:assume-large}, with $\rv Z_\aye = \set{\rv Y_\aye \mid \rv Y_S}$, we obtain
       \begin{equation} 
           \ExpOp_{S \sim \Pi_k^m}{\Ex{\set{\rv Y_S}}{\sum_{\ell \in J(\tree)} \ExpOp_{ \set{\tee_1, \tee_2} \sim w_{\ell_1, \ell_2}}{ 
            \norm{\set{\rv Y_{\tee_1 \sqcup \tee_2} \mid \rv Y_S} - \set{\rv Y_{\tee_1} \mid \rv Y_S } \set{\rv Y_{\tee_2} \mid \rv Y_S }}_1}}} \ge \frac{\ee}{2}.\label{eq:analogy-tree-bd}
       \end{equation}
    In particular, in the sum over $J(\tree)$ there should be some large term corresponding to some $\ell \in J(\tree)$. i.e.~we have,  
    \[ \ExpOp_{S \sim \Pi_k^m}{\Ex{\set{\rv Y_S}}{ \ExpOp_{ \set{\tee_1, \tee_2} \in w_{\ell_1, \ell_2}}{ 
            \norm{\set{\rv Y_{\tee_1 \sqcup \tee_2} \mid \rv Y_S} - \set{\rv Y_{\tee_1} \mid \rv Y_S } \set{\rv Y_{\tee_2} \mid \rv Y_S }}_1}}} \ge \frac{\ee}{2 \cdot |J(\tree)|} \ge \frac{\ee}{2k}.\]
    Now,  we have
    \[ \Pr{S \sim \Pi_k^m\atop \set{\rv Y_S} }{ \ExpOp_{\set{\ess, \tee} \in 
    w_{\ell_1, \ell_2}}{\norm{ \set{\rv Y_{\tee_1 \sqcup \tee_2} \mid \rv Y_S } -  \set{\rv Y_{\tee_1} \mid \rv Y_S} \set{\rv Y_{\tee_2} \mid \rv Y_S}}_1} \ge \frac{\ee}{4k} } \ge \frac{\ee}{4k}.\]
    This together with \cref{cor:brs-explicit} implies,
    \begin{multline}
        \Pr{S \sim \Pi_k^m \atop \set{\rv Y_S} }{ \Ex{\aye \in \Pi_{\ell_1, \ell_2}}{ \Ex{\tee_1 \sim \Pi_{\ell_1}}
        {\Var{ \rv Y_{\tee_1} \mid \rv Y_S } - \Var{\rv Y_{\tee_1}\mid \rv Y_S, \rv Y_\aye}}\atop + \quad  \Ex{\tee_2 \in \Pi_{\ell_2}}{\Var{\rv Y_{\tee_2} | \mid \rv Y_S} - \Var{\rv Y_{\tee_2} \mid \rv Y_S, \rv Y_\aye}} } \ge c_2 \cdot \frac{ \ee^2 }{16 \cdot k^2 \cdot q^{2k}}}\\ \ge \frac{\ee}{4k},\label{eq:swap-dist-drop1}
    \end{multline}
    provided that $\lambda_2(G_{\ell_1, \ell_2}) \le c_0 (\ee / (4k \cdot q^{k}))^2$.

    Now, observe that a sample $\aye \sim \Pi_{\ell_1, \ell_2}$ can be obtained from a sample $\ess_{m + 1} \sim \Pi_k$
    in the following way,
    \begin{enumerate}
        \item with probability $\frac12$ each, pick $j=1$ or $j = 2$.
        \item delete all but $\ell_j$ elements from $\ess_{m+1}$.
    \end{enumerate}
    It is important to note that for the sample $\aye \sim \Pi_{\ell_1, \ell_2}$ obtained this way, we have $\ess_{m + 1} \supseteq \aye$. An application of Jensen's 
    inequality shows that the variance is non-increasing under conditioning, i.e.~for random variables $\rv Z$ and $\rv W$ we have,
    \begin{eqnarray*}
        \Ex{\rv Z}{\Var{\rv W \mid \rv Z}} & = & \Ex{\rv Z}{\Ex{\rv W}{\rv W^2 \mid \rv Z}} - \Ex{\rv Z}{\parens*{\Ex{\rv W}{\rv W \mid \rv Z }^2}},\\
        & \le & \Ex{\rv W^2} - \parens*{\Ex{\rv Z}{\Ex{\rv W \mid \rv Z }}}^2,\\
        & = & \Var{\rv W}.
    \end{eqnarray*}
    This means conditioning on $\ess_{m + 1}$, the drop in variance can only be more, \ie~\cref{eq:swap-dist-drop1} implies
    \[    \Pr{S \sim \Pi_k^m \atop \set{\rv Y_S} }{ \Ex{\ess_{m + 1} \in \Pi_{k}}{ \Ex{\tee_1 \sim \Pi_{\ell_1}}
        {\Var{ \rv Y_{\tee_1} \mid \rv Y_S } - \Var{\rv Y_{\tee_1} \mid \rv Y_S, \rv Y_{\ess_{m + 1}}}}\atop + \quad  \Ex{\tee_2 \in \Pi_{\ell_2}}{\Var{\rv Y_{\tee_2} | \mid \rv Y_S} - \Var{\rv Y_{\tee_2} \mid \rv Y_S, \rv Y_{\ess_{m+1}}}} } \ge c_2 \cdot \frac{ \ee^2 }{16 \cdot k^2 \cdot q^{2k}}} \ge \frac{\ee}{4k}.\]
   By relabeling $\ell_1$ as $\ell_2$ if needed, we can obtain the following inequality from the above
    \begin{equation}
        \Pr{S \sim \Pi_k^m \atop \set{\rv Y_S} }{ \Ex{\ess_{m + 1} \in \Pi_{k}}{ \Ex{\tee_1 \sim \Pi_{\ell_1}}
        {\Var{ \rv Y_{\tee_1} \mid \rv Y_S } - \Var{\rv Y_{\tee_1} \mid \rv Y_S, \rv Y_{\ess_{m + 1}}} }} \ge c_2 \cdot \frac{ \ee^2 }{32 \cdot k^2 \cdot q^{2k}}} \ge \frac{\ee}{4k}.
    \end{equation}
    This implies
    \[ \Phi_m -\Phi_{m + 1}\ge \frac{1}{k} \cdot \frac{\ee}{4k} \cdot \parens*{ c_2 \cdot \frac{\ee^2}{32 \cdot k^2 \cdot q^{2k}}} = c_2 \cdot \frac{\ee^3}{128 \cdot k^4 \cdot q^{2k}},\]
    where the $\frac{1}{k}$ term in the RHS corresponds to $\ell_1 \in [k]$ being chosen in \cref{eq:pot-def}, the
    $\frac{\ee}{4k}$ term in the RHS corresponds to the probability of the variances in $X(\ell_1)$ drop by $\parens*{ c_2 \cdot \frac{\ee^2}{32 \cdot k^2 \cdot q^{k}}}$. Since, the variance is non-increasing under conditioning
    \[ 1 \ge \Phi_1 \ge \cdots \ge \Phi_m \ge 0.\]
    this means there can be at most $128 k^4 \cdot q^{2k}/(c_2 \cdot \ee^3)$ indices $m \in [L/k]$ such that $\ee_m \ge \ee/2$. In particular, since the total number of indices is $(L/k)$ we have,
    \[ \ExpOp_{m \sim [L/k]}{\ee_m } \le \frac{\ee}{2} + \frac{k}{L} \cdot \frac{128 \cdot k^4 \cdot q^{2k}}{c_2 \cdot \ee^3}. \]
    This means that there exists an absolute constant $c' \ge 0$ such that
    \[ L \ge c' \cdot \frac{k^5 \cdot q^{2k}}{\ee^4}\quad \textrm{ ensures } \quad \Ex{m \in [L/k]}{\ee_m} \le \ee.\]
    To finish our proof, we note that to justify our applications of \cref{cor:brs-explicit} it suffices to ensure
    \[ \lambda_2(G_{\ell_1, \ell_2}) \le c_0 \cdot \parens*{\frac{\ee}{4 k \cdot q^{k}}}^2 = c_0 \cdot \frac{\ee^2}{16 \cdot k^2 \cdot q^{2k}}\]
    for all $\ell_1, \ell_2$ occurring in $\tree$ as a label. It can be verified that our choice of $\gamma$ together
    with \cref{thm:swap-eig-bd} (and \cref{rem:sing-to-eig}) satisfies this, where the constant $C' \ge 0$ will account for $c_0$, $c'$, and 
    the constants hidden within the $O$-notation in \cref{thm:swap-eig-bd}.
\end{proof}
\subsection{The Glorified Triangle Inequality: Proof of \cref{lem:split_rectangular}}
In this Section, we will prove \cref{lem:split_rectangular}.
\begin{proposition}\label{prop:triangle_tensor}
    Let $\rv Y, \rv Z, \rv U, \rv W$ be random variables where $\rv Y$ and $\rv Z$; and $\rv U$ and $\rv W$ are on the same support. Then,
    \[ \norm{ \set*{\rv Y}\set*{\rv U} - \set*{\rv Z} \set*{\rv W}}_1 \le \norm{\set*{\rv Y} - \set*{\rv Z}}_1 +
    \norm{\set*{\rv U} - \set*{\rv W}}_1.\]
\end{proposition}
\begin{proof}
    Tensoring with the same probability distribution does not change the total variation distance, i.e.~ 
    \[\norm{\set*{\rv Y} - \set*{\rv Z}}_1 = \norm{\set*{\rv Y}\set*{\rv U} - \set*{\rv Z}\set*{\rv U}}_1 \tand \norm{\set*{\rv U} - \set*{\rv W}}_1 = \norm{\set*{\rv Z}\set*{\rv U} - \set*{\rv Z}\set{\rv W}}_1.\]
    Now, a simple application of the triangle inequality proves the Proposition.
\end{proof}
A straightforward implication of \cref{prop:triangle_tensor} is the following, which will allow us to bound 
the correlation along a face $\aye \in X(k)$, using the correlation along sub-faces $\ess, \tee \subseteq \aye$.
\begin{corollary}\label{cor:split_distance} 
    Let $\aye \in X(\ell)$ and $\ess \in X(\ell_1), \tee \in X(\ell_2)$
    be given such that $\aye = \ess \sqcup \tee$. Then for any
    $k$-local PSD ensemble $\set*{\rv Y_1, \ldots, \rv Y_n}$
    we have
    \begin{multline*}
        \norm{\set*{\rv Y_\aye} - \set*{\rv Y_{a_1}}\cdots \set*{\rv Y_{a_\ell}}}_1 \le \norm{\set*{\rv Y_{\aye}} - \set*{\rv Y_\ess}\set*{\rv Y_\tee}}_1\\
        +\norm{\set*{\rv Y_\ess} - \set*{\rv Y_{s_1}} \cdots \set{\rv Y_{s_{\ell_1}}}}_1 + \norm{\set*{\rv Y_\tee} - \set{\rv Y_{t_1}} \cdots\set{\rv Y_{t_{\ell_2}}}}_1
    \end{multline*}
\end{corollary}

With this, we can go ahead and prove \cref{lem:split_rectangular}
\begin{proof}[Proof of \cref{lem:split_rectangular}]
    Let $\aye \in X(\ell)$ be a fixed face. By \cref{cor:split_distance} and averaging over all the
    $\binom{\ell = \ell_1 + \ell_2}{\ell_1}$ ways of splitting $\aye$ into
    $\set*{\ess, \tee}$ such that $\ess \in X(\ell_1)$ and $\tee \in X(\ell_2)$ we have
    \begin{multline*}
    \norm{\set*{\rv Y_\aye} - \prod_{i = 1}^{\ell = \ell_1 + \ell_2} \set*{\rv Y_{a_i}}}_1\\
        \le \frac{1}{\binom{\ell_1 + \ell_2}{\ell_1}} \sum_{\set*{\ess, \tee}}\parens*{\norm{\set*{\rv Y_\aye} - \set{\rv Y_\ess} \set{\rv Y_\tee} }_1 + \norm{\set*{\rv Y_\ess} - \prod_{i = 1}^{\ell_1} \set*{Y_{s_i}}}_1 + \norm{\set{\rv Y_\tee} - \prod_{i = 1}^{\ell_2} \set*{\rv Y_{t_i}} }}.
    \end{multline*}
    Now, by taking an average over all the edges $\aye \in X(\ell)$ (with respect to the measure $\Pi_{\ell}$) we obtain,
    \begin{multline*}
        \Ex{\aye \sim \Pi_{\ell}}{\norm{\set*{\rv Y_\aye} - \prod_{i = 1}^{\ell} \set{\rv Y_{a_i}}}_1}\\
        \le  \frac{1}{\binom{\ell}{\ell_1}} \cdot \Ex{\aye \in \Pi_{\ell}}{\sum_{\set{\ess, \tee}} \parens*{\norm{\set*{\rv Y_\aye} - \set{\rv Y_\ess} \set{\rv Y_\tee} }_1  + \norm{\set*{\rv Y_\ess} - \prod_{i=1}^{k_1} \set{\rv Y_{s_i}}}_1 + \norm{\set{\rv Y_\tee} - \prod_{i= 1}^{\ell_2} \set{\rv Y_{t_i}}}_1}}
    \end{multline*}
    where the indices $\set{\ess, \tee}$ run over the all the ways of splitting $\aye$ into $\ess$ and $\tee$ as
    before. We can now see that the RHS can be thought as an average over the (weighted) edges in $E(\ell_1, \ell_2)$ (q.v.~\cref{eq:edg_exp}), \ie
    \begin{multline*}
        \Ex{\aye \sim \Pi_{\ell}}{\norm{\set*{\rv Y_\aye} - \prod_{i = 1}^{\ell} \set{\rv Y_{a_i}}}_1}\\
        \le  \Ex{\set{\ess, \tee} \sim w_{\ell_1, \ell_2}}{\norm{\set*{\rv Y_\aye} - \set{\rv Y_\ess} \set{\rv Y_\tee} }_1 
    + \norm{\set*{\rv Y_\ess} - \prod_{i=1}^{\ell_1}\set{\rv Y_{s_i}}}_1 + \norm{\set{\rv Y_\tee} - \prod_{i =1}^{\ell_2} \set{Y_{t_i} }}_1}
    \end{multline*}
    Now, note that since $\Pi_{\ell_1, \ell_2}$ (q.v.~\cref{eq:stat_meas}) is the stationary distribution of the walk defined on $G_{\ell_1, \ell_2}$, \ie~
    \[ 2\Pi_{\ell_1, \ell_2}(\aye) = \sum_{\bee: \set{\aye, \bee} \in E(\ell_1, \ell_2)} w_{\ell_1, \ell_2}(\aye, \bee),\]
    the lemma follows. This is because, we have
    {\small
    \begin{multline*}
        \Ex{\aye \in X(\ell)}{\norm{\set*{\rv Y_\aye} - \prod_{i = 1}^{\ell} \set{\rv Y_{a_i}}}_1}\\
        \le  \Ex{\set{\ess, \tee} \sim w_{\ell_1, \ell_2}}{\norm{\set*{\rv Y_\aye} - \set{\rv Y_\ess} \set{\rv Y_\tee} }_1} 
        + \Ex{\set{\ess, \tee} \sim w_{\ell_1, \ell_2}}{\norm{\set*{\rv Y_\ess} - \prod_{i=1}^{\ell_1}\set{\rv Y_{s_i}}}_1 + \norm{\set{\rv Y_\tee} - \prod_{i =1}^{\ell_2} \set{Y_{t_i} }}_1}\\
         =  \Ex{\set{\ess, \tee} \sim E(\ell_1, \ell_2)}{\norm{\set{\rv Y_{\aye}} - \set{\rv Y_\ess}\set{\rv Y_\tee}}_1} + \Ex{\ess \sim \Pi_{\ell_1}}{\norm{\set{\rv Y_\ess} - \prod_{i=  1}^{\ell_1} \set{\rv Y_{s_i}}}_1} +
        \Ex{\tee \sim \Pi_{\ell_2}}{\norm{\set{\rv Y_\tee} - \prod_{i = 1}^{\ell_2} \set{\rv Y_{t_i}}}_1}
      \end{multline*}
      }
\end{proof}


%% file: trank.tex
In~\cite{BarakRS11},~\cref{thm:brs_exp} was proven for a more general
class of graphs than expander graphs -- namely, the class of low
threshold rank graphs.
\begin{definition}[Threshold Rank of Graphs (from~\cite{BarakRS11})]\label{def:graph_trank}
  Let $G = (V,E,w)$ be a weighted graph on $n$ vertices and $\Aye$ be
  its normalized random walk matrix. Suppose the eigenvalues of $\Aye$
  are $1 = \lambda_1 \ge \cdots \ge \lambda_n$. Given a parameter
  $\tau \in (0,1)$, we denote the threshold rank of $G$ by $\Rank_{\ge
    \tau}(\Aye)$ (or $\Rank_{\ge \tau}(G)$) and define it as
  \[
    \Rank_{\ge \tau}(\Aye) \coloneqq \left\vert \{i \vert \lambda_i \ge \tau \} \right\vert.
  \]
\end{definition}
There~\cite{BarakRS11}, the authors asked for the correct notion of
threshold rank for $k$-CSPs. In this section, we give a candidate
definition of low threshold rank motivated by our techniques.

To break $k$-wise correlations it is sufficient to assume that the
involved swap graphs in the foregoing discussion are low threshold
rank since this is enough to apply a version of \cref{lem:brs-dec},
already described in the work of \cite{BarakRS11}.

Moreover, we have some flexibility as to which swap graphs to consider
as long as they satisfy some splitting conditions. To define a swap
graph it is enough to have a distributions on the hyperedges of a
(constraint) hypergraph. Hence, the notion of swap graph is
independent of high-dimensional expansion. HDXs are just an
interesting family of objects for which the swap graphs are good
expanders.

To capture the many ways of splitting the statistical distance over
hyperedges into the statistical distance over the edges of swap
graphs, we first define the following notion. We say that a binary
tree $\tree$ is a $k$-splitting tree if it has exactly $k$
leaves. Thus, labeling every vertex with the number of leaves on the
subtree rooted at that vertex ensures,
\begin{itemize}
  \item the root of $\tree$ is labeled with $k$ and all other vertices are labeled with
        positive integers,
  \item the leaves are labeled with $1$, and
  \item each non-leaf vertex satisfy the property that its label is the sum of
        the labels of its two children.
\end{itemize}

Note that, we will think of each non-leaf node with left and right
children labeled as $a$ and $b$ as representing the swap graph from
$X(a)$ to $X(b)$ for some simplicial complex $X(\le k)$. Let
$\SwapT(\tree, X)$ be the set of all such swap graphs over $X$ finding
representation in the splitting tree $\tree$. Indeed the tree $\tree$
used in the proof of \cref{thm:brs_hdx} is just one special instance of a
$k$-splitting tree.

Given a threshold parameter $\tau \le 1$ and a set of normalized
adjacency matrices $\mathcal{A} = \{\Aye_1,\dots,\Aye_s\}$, we define
the threshold rank of $\mathcal{A}$ as
$$
\Rank_{\ge \tau}(\mathcal{A}) \coloneqq \max_{\Aye \in \mathcal{A}} ~\Rank_{\ge \tau}(\Aye),
$$
where $\Rank_{\ge \tau}(\Aye)$ is denotes usual threshold rank of
$\Aye$ as in~\cref{def:graph_trank}.

Now, we are ready to define the notion of a $k$-CSP instance being
$(\tree, \tau, r)$-splittable as follows.

\begin{definition}[$(\tree, \tau, r)$-splittability]
    A $k$-CSP instance $\Ins$ with the constraint complex $X(\le k)$ is said to be $(\tree, \tau, r)$-splittable if
    $\tree$ is a $k$-splitting tree and 
    \[ \Rank_{\ge \tau}(\SwapT(\tree, X)) \le r.\]
    If there exists some $k$-splitting tree $\tree$ such that
    $\Ins$ is $(\tree, \tau, r)$-splittable, the instance $\Ins$ will be called a $(\tau, r)$-splittable instance.
\end{definition}

Now, using this definition we can show that whenever
$\Rank_\tau(\Ins)$ is bounded for the appropriate choice of $\tau$,
after conditioning on a random partial assignment as in
\cref{algo:prop-rd} we will have small correlation over the faces
$\aye \in X(k)$, \ie
\begin{theorem}\label{thm:brs_tree_split}
    Suppose a simplicial complex $X(\le k)$ with $X(1) = [n]$ and an $L$-local PSD ensemble
    $\rv Y = \set{\rv Y_1, \ldots, \rv Y_n}$ are given. 
    There exists some universal constants $c_4 \ge 0$ and 
    $C'' \ge 0$ satisfying the following:
    If $L \ge C'' \cdot (q^{4k} \cdot k^7 \cdot r/\ee^5)$, $\supp(\rv Y_j) \le q$
    for all $j \in [n]$, and $\Ins$ is $(c_4 \cdot (\ee/(4 k \cdot q^{k}))^2, r)$-splittable. Then, we have
    \begin{equation}
        \Ex{\aye \in X(k)}{ \norm{\set{\rv Y'_{\aye}} - \set*{\rv Y'_{a_1}}\cdots \set*{\rv Y'_{a_k}}}_1 } \le \ee,
        \label{eq:brs_hdx}
    \end{equation}
    where $\rv Y'$ is as defined in \cref{algo:prop-rd} on the input of $\set{\rv Y_1, \ldots, \rv Y_n}$ and $\Pi_k$. 
\end{theorem}
It is important to note that the specific knowledge of the
$k$-splitting tree $\tree$ that makes $\Ins$ $(\tree, \tau,
r)$-splittable is only needed for the proof of
\cref{thm:brs_tree_split}. The conclusion of \cref{thm:brs_tree_split}
can be used without the knowledge of the specific $k$-splitting tree
$\tree$. The attentive reader might have noticed is
that in the proof of \cref{thm:brs_hdx}, the choice of $\tree$ is not important, as all the splitting tree are guaranteed to have be expanders provided that $X$ is a $\gamma$-HDX. The proof of \cref{thm:brs_tree_split}, in
this light can be thought of an extension of the proof of
\cref{thm:brs_hdx} to the case where not necessarily every tree is
good, and where we can bound the threshold rank instead of the
spectral expansion.

This, will readily imply an algorithm
\begin{corollary}\label{cor:alg-cooked-baked2}
    Suppose $\Ins$ is a $q$-ary $k$-CSP instance 
    whose constraint complex is $X(\le k)$. 
    There exists an absolute constant $C'' \ge 0$ and $c_4 \ge 0$
    that satisfies the following: If $\Ins$ is $(c_4 \cdot (\ee/ (4k \cdot q^{k}))^2, r)$-splittable, then there is an algorithm that runs in time
    $n^{O\parens*{\frac{q^{4k} \cdot k^7 \cdot r}{\ee^5}}}$ and that is based on
    $(\frac{C'' \cdot k^5 \cdot q^{k} \cdot r}{\ee^4})$-levels of SoS-hierarchy and \cref{algo:prop-rd} that outputs a random
    assignment $\assn: [n] \to [q]$ that in expectation ensures $\SAT_\Ins(\assn) =
    \OPT(\Ins) - \ee$. 
\end{corollary}
Since the proof of \cref{cor:alg-cooked-baked2} given \cref{thm:brs_tree_split}, will be almost identical to the proof of \cref{cor:alg-cooked}, given \cref{thm:brs_hdx}, we will omit the proof of this.
\subsection{Breaking Correlations for Splittable CSPs: Proof of \cref{thm:brs_tree_split}}

We will need the more general version of \cref{lem:brs-dec}, already proven in \cite{BarakRS11}.
\begin{lemma}[Lemma 5.4 from \cite{BarakRS11}]\label{lem:brs-dec2}\footnote{We give a derivation of this lemma in \cref{app:appenditis}.}
    Let $G = (V, E, \Pi_2)$ be a weighted graph, $\set{\rv Y_1, \ldots, \rv Y_n}$ a local PSD ensemble, where we have
    $\supp(\rv Y_i) \le q$ for every $i \in V$,  and $q \ge 0$. If $\ee \ge 0$ is a lower bound on 
    the expected statistical difference between independent
    and correlated sampling along the edges,\ie
    \[ \ee \le \Ex{\set{i,j} \sim \Pi_2}{\norm{\set{\rv Y_{ij} } - \set{\rv Y_i}\set{\rv Y_j}  }_1}.\]
    There exists absolute constants $c_3 \ge 0 $ and $c_4 \ge 0$ that satisfy the following:
    Then, conditioning on a random vertex decreases the variances,
    \[ \ExpOp_{i\sim \Pi_1}\ExpOp_{j \sim \Pi_1}{\Ex{ \set{\rv Y_j} }{ \Var{\rv Y_i \mid \rv Y_j } } } \le \Ex{i \sim \Pi_1}{\Var{\rv Y_i}} - c_3 \cdot \frac{\ee^4}{q^4 \cdot \Rank_{\ge c_4 \ee^2/q^2}(G) }.\]
\end{lemma}
Since we will use this lemma, only with the swap graphs $G_{\ell_1, \ell_2}$ and $(L/k)$-local PSD
ensemble $\set{\rv Y_\aye}_{\aye \in X}$ obtained from the $L$-local PSD ensemble $\set{\rv Y_1, \ldots, \rv Y_n}$, for
convenience we will write the corollary we will use more explicitly
\begin{corollary}\label{cor:brs-explicit2}
    Let $\ell_1 \ge \ell_2 \ge 0$ satisfying $\ell_1 + \ell_2 \le k$ be given parameters, and let
    $G_{\ell_1, \ell_2}$ be the swap graph defined for a $\gamma$-HDX $X(\le k)$. 
    Let $\set{\rv Y_\aye}_{\aye \in X}$ be a local PSD ensemble; and suppose we have $\supp(Y_\aye) \le q^{k}$
    for every $\aye \in X(\ell_1) \cup X(\ell_2)$ for some $q \ge 0$. Suppose $\ee > 0$ satisfies,
    \[ \frac{\ee}{4k} \le \Ex{\set{\ess, \tee} \in E(\ell_1, \ell_2)}{\norm{\set{\rv Y_{\ess \cup \tee}} - \set{\rv Y_\ess}\set{\rv Y_\tee}}_1}.\]
    There exists absolute constants $c_3 \ge 0$ and $c_5 \ge 0$ that satisfy the
    following: \\
    If $\Rank_{\ge c_4 \cdot (\ee/(4k \cdot q^{k}))^2}(G_{\ell_1, \ell_2}) \le
    r$, then conditioning on a random face $\aye \sim \Pi_{\ell_1, \ell_2}$
    decreases the variances, i.e.~
    \begin{eqnarray*}
        2 \cdot \Ex{\aye, \bee \sim \Pi_{\ell_1, \ell_2}}{\Ex{\set{\rv Y_\aye}}{\Var{\rv Y_\bee \mid \rv Y_\aye}}} & = &
        \Ex{\aye \in \Pi_{\ell_1, \ell_2}}{\Ex{\ess \sim \Pi_{\ell_1}}{\Var{\rv Y_\ess \mid \rv Y_\aye}} 
        + \Ex{\tee \sim \Pi_{\ell_2}}{\Var{\rv Y_\tee \mid \rv Y_\aye}}},\\
        & \le & \Ex{\ess \sim \Pi_{\ell_1}}{\Var{\rv Y_\ess}} + \Ex{\tee \sim \Pi_{\ell_2}}{\Var{\rv Y_\tee}} - c_5 \cdot \frac{\ee^4}{256 \cdot k^4 \cdot q^{4k} \cdot r}.
    \end{eqnarray*}
\end{corollary}
Here the constant $c_5$ satisfies $c_5 = 2 \cdot c_3$.
\begin{proof}
    As the proof will mostly follow \cref{thm:brs_hdx}, we will only highlight the relevant differences and carry
    out the relevant computations.

    Let $\tau = c_4 \cdot (\ee/ (4k \cdot q^{k}))^2$, and let $\tree$ be the $k$-splitting tree certifying
    that $\Ins$ is $(\tree, \tau, r)$ splittable, \ie~the tree $\tree$ satisfies $\Rank_\tau(\SwapT(\tree, X) ) \le r$. This means that
    all the swap graphs $G_{\ell_1, \ell_2}$ finding representation in $\tree$ satisfy $\Rank_\tau(G_{\ell_1, \ell_2}) \le r$.

    Similarly, as in the proof of we will try to argue that the fraction of indices $m \in [L/k]$ such 
    that $\ee_m$ that is large, say $\ee_m \ge \ee/2$, is small by arguing about the potential $\Phi_m$ with both quantities $\ee_m$ and $\Phi_m$
    as defined as in the Proof of \cref{thm:brs_hdx}. We assume similarly, that $\ee_m \ge \ee/2$ for some $m \in [L/k]$.

    Analogously to \cref{eq:analogy-tree-bd} in the proof of \cref{thm:brs_hdx}, from \cref{cor:brs-explicit2} we 
    obtain
    \begin{equation*} 
        \ExpOp_{S \sim \Pi_k^m}{\Ex{\set{\rv Y_S}}{\sum_{\ell\in J(\tree)} \Ex{ \set{\tee_1, \tee_2} \in E(\ell_1, \ell_2)}{ 
            \norm{\set{\rv Y_{\tee_1 \sqcup \tee_2} \mid \rv Y_S} - \set{\rv Y_{\tee_1} \mid \rv Y_S } \set{\rv Y_{\tee_2} \mid \rv Y_S }}_1}}} \ge \frac{\ee}{2}.
    \end{equation*}
    Notice that the assumption that \cref{eq:analogy-tree-bd} makes on the threshold rank is satisfied by the 
    assumption $\Rank_\tau(\Ins) \le r$ and
    where the set $J(\tree)$ contains all labels $\ell$
    of internal nodes $v \in \tree$, and we write $\ell_1$ (resp.~$\ell_2$) for the label of the left 
    (resp.~right) child of the vertex with the label $\ell$. Similarly, to the proof of \cref{thm:brs_hdx}, there exists some $(\ell_1, \ell_2) \in J(\tree)$ that satisfies
    \begin{equation*} 
        \ExpOp_{S \sim \Pi_k^m}{\ExpOp_{\set{\rv Y_S}}{\ExpOp_{ \set{\tee_1, \tee_2} \sim w_{\ell_1, \ell_2}}{ 
            \norm{\set{\rv Y_{\tee_1 \sqcup \tee_2} \mid \rv Y_S} - \set{\rv Y_{\tee_1} \mid \rv Y_S } \set{\rv Y_{\tee_2} \mid \rv Y_S }}_1}}} \ge \frac{\ee}{2k}.\label{eq:analogy-tree-bd}
    \end{equation*}
    Now, analogously to \cref{eq:swap-dist-drop1}, using $\ell_1 \le k$ using we have
    \begin{equation}
    \Pr{S \sim \Pi_k^m \atop \set{\rv Y_S} }{ \Ex{\aye \in \Pi_{\ell_1, \ell_2}}{ \Ex{\tee_1 \in X(\ell_1)}
        {\Var{ \rv Y_{\tee_1} \mid \rv Y_S } - \Var{\rv Y_{\tee_1}\mid \rv Y_S, \rv Y_\aye}}\atop + \quad  \Ex{\tee_2 \in X(\ell_2)}{\Var{\rv Y_{\tee_2} | \mid \rv Y_S} - \Var{\rv Y_{\tee_2} \mid \rv Y_S, \rv Y_\aye}} } \ge c_5 \cdot \frac{ \ee^4 }{256 \cdot k^4 \cdot q^{4k} \cdot r  }}\\
        \ge \frac{\ee}{4k},
    \end{equation}
    Using the same arguments in the proof of \cref{thm:brs_hdx}, we can get that
    \[ \Phi_{m} - \Phi_{m + 1} \ge \frac{1}{k} \cdot \frac{\ee}{4k} \cdot \frac{c_5 \cdot \ee^4}{512 \cdot k^6 \cdot q^{4k} \cdot r} = c_5 \cdot \frac{\ee^5}{2048 \cdot k^4 \cdot q^{4k} \cdot r}.\]
    Again, this would mean that there can be at most $2048\cdot k^6 \cdot q^{4k} \cdot r/(\ee^5 \cdot c_5)$ indices $m$
    such that $\ee_{m/2} \ge \ee/2$. In particular,
    \[ \Ex{m \in [L/k]}{\ee_m} \le \frac{\ee}{2} + \frac{k}{L} \cdot \frac{2048 \cdot k^6 \cdot q^{4k} \cdot r}{\ee^5 \cdot c_5}.\]
    i.e.~there exists a universal constant $C'' \ge 0$, such that
    \[ L \ge C'' \cdot \frac{k^7 \cdot q^{4k} \cdot r}{\ee^5} \textrm{ ensures } \ExpOp_{m \sim [L/k]}{\ee_m} \le \ee.\]
\end{proof}


%% file: quantum_hamiltonian.tex
Our $k$-CSP results extend to the quantum setting generalizing the
approximation scheme for $2$-local Hamiltonians on bounded degree low
threshold rank graphs from Brand{\~{a}}o and Harrow~\cite{BrandaoH13}
(BH). Before we can make the previous statement more precise we will
need to introduce some notation. A well studied quantum analogue of
classical $k$-CSPs are the so-called quantum $k$-local
Hamiltonians~\cite{ADV13}.
\begin{definition}[$k$-local Hamiltonian]
  We say that $\matr H = \E_{\ess \sim \Pi_k} \matr H_{\ess}$ is an
  instance of the $k$-local Hamiltonian problem over $q$-qudits on
  ground set $[n]$ if there is a distribution $\Pi_k$ on subsets of
  size $k$ of $[n]$ such that for every $\ess \in \supp(\Pi_k)$ there
  is an Hermitian operator $\matr H_{\ess}$ on $\mathbb{C}^{q^n}$ with
  $\norm{\matr H_{\ess}}_{\textup{op}} \le 1$ and acting (possibly)
  non-trivially on the $q$-qudits of $\ess$ and trivially on
  $[n]\setminus \ess$.
\end{definition}
Given an instance $\matr H = \E_{\ess \sim \Pi_k} \matr H_{\ess}$ of
the $k$-local Hamiltonian problem on ground set $[n]$, the goal is to
provide a good (additive) approximation to the \textit{ground state
energy} $e_0(\matr H)$, i.e., the smallest eigenvalue of $\matr H$.
Equivalently, the goal is to approximate
\[
e_0(\matr H) = \min_{\rho \in D\left(\mathbb{C}^{q^n}\right)} \Tr(\matr H \rho),
\]
where $D\left(\mathbb{C}^{q^n}\right)$ is the set of density
operators, PSD operators of trace one, on $\mathbb{C}^{q^n}$. The
eigenspace of $\matr H$ associated to $e_0(\matr H)$ is called the
\textit{ground space} of $\matr H$.
\begin{remark}
  The locality $k$ of a $k$-local Hamiltonian has a similar role as
  the arity of $k$-CPSs whereas the qudit dimension $q$ has the role
  of alphabet size. Observe that for a $k$-CSP the goal is to maximize
  the fraction of satisfied constrains while for a $k$-local
  Hamiltonian the goal is to minimize the energy (constraint
  violations).
\end{remark}

We will need an informationally complete measurement $\Lambda$ modeled
as a channel
\[
\Lambda \colon D\left(\mathbb{C}^q\right) \to D\left(\mathbb{C}^{q^8}\right),
\] and defined as
\[
\Lambda(\rho) \coloneqq \sum_{y \in \mathcal{Y}} \Tr(\matr M_y \rho) \cdot e_y e_y^{\dag},
\]
where $\{\matr M_y\}_{y \in \mathcal{Y}}$ is a POVM~\footnote{A POVM
  is a collection of operators $\{\matr M_y\}_{y \in \mathcal{Y}}$
  such that $\sum_{y \in \mathcal{Y}} M_y = \Ide$ and $(\forall y \in
  \mathcal{Y})(\matr M_y \succeq 0)$.} and $\{ e_y \}_{y \in
  \mathcal{Y}}$ is an orthonormal basis
(see~\cref{lemma:bh_distortion_rec} below for the properties of
$\Lambda$). Recall that an informationally complete measurement is an
injective channel, i.e., the probability outcomes $p(y) = \Tr(\matr M_y
\rho)$ fully determine $\rho$. By definition given this probability
distribution $\{ p(y) \}_{y \in \mathcal{Y}}$ we can uniquely
determine $\rho$. We use the notation $\rho = \Lambda^{-1}\left(\{
p(y) \}_{y \in \mathcal{Y}}\right)$ for the recovered state from
probability outcomes $\{ p(y) \}_{y \in \mathcal{Y}}$.

BH using the informationally complete measurement $\Lambda$ reduced
the quantum $2$-local Hamiltonian problem to a classical problem
involving PSD ensembles of indicator random variables of outcomes
$\mathcal{Y}$ of $\Lambda$. In this reduction, they had to ensure that
the local distributions encoded by these indicators random variables
are indeed consistent with probability distributions of outcomes
arising from actual local density matrices. Note that the channel
$\Lambda$ is only injective, an arbitrary probability distribution on
$\mathcal{Y}$ may not correspond to a valid quantum state. For this
reason, they introduced a new SDP hierarchy to find this special kind
of PSD ensemble, which we refer to as \textit{quantum PSD ensemble},
minimizing the value of the given input $k$-local Hamiltonian
instance.
         
Using our $k$-CSP approximation scheme for low threshold rank
hypergraphs, we show that product state approximations close to the
ground space of $k$-local Hamiltonians on bounded degree low threshold
rank hypergraphs can be computed efficiently in polynomial time
by~\cref{algo:quantum-prop-rd}. Our result is a generalization of the
$k = 2$ case of Brand{\~{a}}o and Harrow~\cite{BrandaoH13} for
$2$-local Hamiltonians on bounded degree low threshold rank
graphs. Their algorithm is based on the $2$-CSP result
from~\cite{BarakRS11}.

\begin{algorithm}{Quantum Propagation Rounding Algorithm}
                 {$L$-local quantum PSD ensemble~\footnote{We define the quantum
                                                              ensemble as the PSD ensemble
                                                              produced by the SDP hierarchy
                                                              of~\cite{BrandaoH13}}
                 ~$\set{\rv Y_1, \ldots, \rv Y_n}$ and distribution $\Pi$ on $X(\le \ell)$.}
                 {A random state $\rho = \rho_1 \otimes \ldots \otimes \rho_n$ where each $\rho_i \in D\left(\mathbb{C}^{q}\right)$.}\label{algo:quantum-prop-rd}
\begin{enumerate}
  \item Choose $m \in \set*{1, \ldots, L/\ell}$ at random.
  \item Independently sample $m$ $\ell$-faces, $\ess_j \sim \Pi$ for $j = 1, \ldots, m$.
  \item Write $S = \bigcup_{j = 1}^m \ess_j$, for the set of the seed vertices.
  \item Sample assignment $\assn_{S}: S \to [q]$ according to the local distribution, $\set{\rv Y_{S}}$.
  \item Set $\rv Y' = \set{\rv Y_1, \ldots \rv Y_n | \rv Y_S = \assn_S}$, i.e.~the local ensemble
        $\rv Y$ conditioned on agreeing with $\assn_S$.
  \item For all $j \in [n]$, set $\rho_j = \Lambda^{-1}(\set{\rv Y'_j})$.
  \item Output $\rho = \rho_1 \otimes \ldots \otimes \rho_n$.
\end{enumerate}
\end{algorithm}

The precise result is given in~\cref{theorem:quantum-alg-cooked-baked}.

\begin{theorem}\label{theorem:quantum-alg-cooked-baked}
    Suppose $\Ins = (\matr H = \E_{\ess \sim \Pi_k} \matr H_{\ess})$
    is a $q$-qudit $k$-local Hamiltonian instance whose constraint
    complex~\footnote{We define the constraint complex of a $k$-local
    Hamiltonian in the same way we define it for $k$-CSPs, namely, by
    taking the downward closure of the support of $\Pi_k$.} is $X(\le k)$ and
    has bounded normalized degree, i.e., $\Pi_1 \le \delta$. Let $\tau
    = c_4 \cdot (\ee^2/ (16k^2 q^{8k}))^2$, for $\ee > 0$. There exists
    an absolute constant $C'$ that satisfies the following:

    Set $L = (\frac{C' \cdot k^5 \cdot q^{8k} \cdot
      \Rank_\tau(\Ins)}{\ee^4})$. Then there is an algorithm based on
    $L$-levels of SoS-hierarchy and \cref{algo:quantum-prop-rd} that
    outputs a random product state $\rho = \rho_1 \otimes \ldots
    \otimes \rho_n$ that in expectation ensures
    $$  
    \Tr(\matr H\rho) \le e_0(\matr H) + (18 q)^{k/2} \cdot \ee +  L \cdot k \cdot \delta,
    $$
    where $e_0(\matr H)$ is the ground state energy of $\matr H$.
\end{theorem}

\begin{remark}
  Similarly to the classical
  case, \cref{theorem:quantum-alg-cooked-baked} serves as a no-go
  barrier (in its parameter regime) to the quantum local-Hamiltonian
  version of the quantum PCP Conjecture~\cite{ADV13}. In particular,
  $k$-local Hamiltonians on bounded degree $\gamma$-HDXs for $\gamma$
  sufficiently small can be efficiently approximated in polynomial
  time.
\end{remark}

Now we sketch a proof of~\cref{theorem:quantum-alg-cooked-baked}. We
provide a sketch rather than a full proof
since~\cref{theorem:quantum-alg-cooked-baked} easily follows from the
BH analysis once the main result used by them, Theorem 5.6
from~\cite{BarakRS11}, is appropriately generalized to ``break''
$k$-wise correlations as accomplished by our~\cref{thm:brs_tree_split}
(restated below for convenience). Furthermore, a full proof would
require introducing more objects and concepts only needed in this
simple derivation (the reader is referred to~\cite{BrandaoH13} for the
quantum terminology and the omitted details).
\begin{theorem}[Adaptation of~\cref{thm:brs_tree_split}]
    Suppose a simplicial complex $X(\le k)$ with $X(1) = [n]$ and an
    $L$-local PSD ensemble $\rv Y = \set{\rv Y_1, \ldots, \rv Y_n}$
    are given.  There exists some universal constants $c_4 \ge 0$ and
    $C'' \ge 0$ satisfying the following: If $L \ge C'' \cdot (q^{4k}
    \cdot k^7 \cdot r/\ee^5)$, $\supp(\rv Y_j) \le q$ for all $j \in
          [n]$, and $\Ins$ is $(c_4 \cdot (\ee/(4 k \cdot q^{8k}))^4,
          r)$-splittable. Then, we have
    \begin{equation}
        \Ex{\aye \in X(k)}{ \norm{\set{\rv Y'_{\aye}} - \set*{\rv Y'_{a_1}}\cdots \set*{\rv Y'_{a_k}}}_1 } \le \ee,
        \label{eq:brs_hdx}
    \end{equation}
    where $\rv Y'$ is as defined in \cref{algo:quantum-prop-rd} on the input of $\set{\rv Y_1, \ldots, \rv Y_n}$ and $\Pi_k$. 
\end{theorem}

Once in possession of the \textit{quantum PSD ensemble} the problem
becomes essentially classical. The key result in the BH approach is
Theorem 5.6 from~\cite{BarakRS11} that brings (in expectation under
conditioning on a random small seed set of qudits) the local
distributions, over the edges of the constraint graph of a $2$-local
Hamiltonian, close to product distributions~\footnote{For this to hold
  we need the underlying constraint graph to be low threshold rank and
  the SoS degree to be sufficiently large}. Now, using the fact that
they have an informationally complete measurement $\Lambda$ they can
``lift'' the conditioned marginal distribution on each qudit $\set{\rv
  Y'_j}$ to an actual quantum state as $\rho_j = \Lambda^{-1}(\set{\rv
  Y'_j})$ (see~\cref{algo:quantum-prop-rd}). In this lifting process,
they pay an \textbf{average} distortion cost of $18q \cdot
\epsilon$ (for using the marginal  over the qudits). For $k$-local Hamiltonians, the distortion of $k$
$q$-qudits is given by~\cref{lemma:quantum_k_wise_distortion} (stated
next without proof).
\begin{lemma}\label{lemma:quantum_k_wise_distortion}
  Let ~$\rv Z_1, \ldots, \rv Z_k$ be random variables in an $L$-local
  quantum PSD ensemble with $L \ge k$. Suppose that
  \[
  \epsilon \coloneqq \norm{\{\rv Z_1, \ldots, \rv Z_k\} - \prod_{i=1}^k \{\rv Z_i\} }_1.
  \]
  Then
  \[
  \norm{\left(\Lambda^{\otimes k}\right)^{-1}\left(\{\rv Z_1, \ldots, \rv Z_k\}\right) - \left(\Lambda^{\otimes k}\right)^{-1}\left(\prod_{i=1}^k \{\rv Z_i\}\right)}_1 \le \left(18q\right)^{k/2} \cdot \epsilon.
  \]
\end{lemma}
Note that~\cref{lemma:quantum_k_wise_distortion} is a direct
consequence of~\cref{lemma:bh_distortion_rec} from~\cite{BrandaoH13}.
\begin{lemma}[Informationally complete measurements (Lemma 16~\cite{BrandaoH13})]\label{lemma:bh_distortion_rec}
  For every positive integer $q$ there exists a measurement $\Lambda$
  with $\le q^8$ outcomes such that for every positive integer $k$ and
  every traceless operator $\xi$, we have
  \[
  \norm{\xi}_1 \le \left(18q\right)^{k/2} \norm{\Lambda^{\otimes k}(\xi)}_1.
  \]
\end{lemma}

BH also pay a full cost for each local term in the Hamiltonian that
involves a seed qudit since its state was not reconstructed using the
full distribution of a qudit given by the quantum ensemble but rather
reconstructed from a single outcome $y \in \mathcal{Y}$ of $\Lambda$.
Naively, this means that the final state of this qudit may be far from
the intended state given by SDP relaxation. In our case, we assume
that the normalized degree satisfies $\Pi_1 \le \delta$. Therefore,
the total error from constraints involving seed qudits is at most
$$
L \cdot k \cdot \delta.
$$

Putting the above pieces together we conclude the proof (sketch)
of~\cref{theorem:quantum-alg-cooked-baked}.

%% file: appendix.tex
We include the key result we use from~\cite{BarakRS11}, namely, their
Lemma 5.4 (below). 
While they proved the lemma for regular graphs, we include the details
in the proof for general weighted graphs, since even for HDXs regular
at the top level, the swap graphs are not necesarily regular.
The extension to general graphs is straighforward (and \cite{BarakRS11}
indicated the same) but we include the details for the sake of completeness
\footnote{For expander graphs it is possible to obtain an improved bound of
  $\Omega((\eps/q)^2)$ instead of $\Omega((\eps/q)^4)$ given by
  \cref{lemma:app_key_brs_result}, simply by using the definition of the second
  smallest eigenvalue of the Laplacian. While BRS analyzed
  $\Ex{i,j}{\ip{v_i}{v_j}^2}$ for low-threshold rank graphs, it is possible
  to directly analyze the quantity $\Ex{i,j}{\ip{v_i}{v_j}}$ for expanders, leading to the
  improved bound.}.
\begin{lemma}[Lemma 5.4 from \cite{BarakRS11} (restatement of~\cref{lem:brs-dec2})]\label{lemma:app_key_brs_result}
    Let $G = (V, E, \Pi_2)$ be a weighted graph, $\set{\rv
    Y_1, \ldots, \rv Y_n}$ a local PSD ensemble, where we have
    $\supp(\rv Y_i) \le q$ for every $i \in V$, and $q \ge 0$. If
    $\ee \ge 0$ is a lower bound on the expected statistical
    difference between independent and correlated sampling along the
    edges,\ie \[ \ee \le \Ex{\set{i,j} \sim \Pi_2}{\norm{\set{\rv
    Y_{ij} } - \set{\rv Y_i}\set{\rv Y_j} }_1}.\] Then, conditioning
    on a random vertex decreases the variances, \[ \Ex{i,
    j \sim \Pi_1}{\Ex{ \set{\rv Y_j} }{ \Var{\rv Y_i \mid \rv Y_j } }
    } \le \Ex{i \sim \Pi_1}{\Var{\rv Y_i}} - \frac{\ee^4}{4
    q^4 \cdot \Rank_{\ee^2/\left(4q^2\right)}(G) }.\]
\end{lemma}

The key ingredient in proving Lemma 5.4 is a ``local to global''
argument generalizing the expander case to low threshold rank
graphs. This new argument is proven in two steps
with~\cref{ap:local_to_glob_inner} being the first one.
\begin{lemma}[Adapted from Lemma 6.1 of~\cite{BarakRS11}]\label{ap:local_to_glob_inner}
  Let $G$ be an undirected weighted graph. Suppose $  v_1,\dots,  v_n \in
  \mathbb{R}^n$ are such that
  \[
  \Ex{i \sim V(G)}{\ip{  v_i}{  v_i}} = 1, \qquad \Ex{ij \sim E(G)}{\ip{  v_i}{  v_j}} \ge 1-\epsilon,
  \]
  but
  \[
  \Ex{i,j \sim V(G)}{\ip{  v_i}{  v_j}^2} \le \frac{1}{m}.
  \]
  Then for $c > 1$, we have
  \[
  \lambda_{\left( 1- \frac{1}{c}\right)^2m} \ge 1 - c\cdot \epsilon.
  \]
  In particular, $\lambda_{m/4} \ge 1 - 2\epsilon$.
\end{lemma}

\begin{proof}
  Let $\matr \Why$ be the Gram matrix defined as $\matr \Why_{i,j} =
  \ip{  v_i}{  v_j}$. Clearly, $\matr \Why$ is positive
  semi-definite. Without loss of generality suppose that the edge
  weights $\{  w(\{i,j\}) ~\vert~ ij \in E(G)\}$ form a
  probability distribution. Set $  w(i) = \sum_{j \sim i}  
  w(\{i,j\})$. Let $\matr D$ to be the diagonal matrix such that
  $\matr D(i,i) =   w(i)$, \ie~the matrix of generalized
  degrees. Let $\Aye$ be such that $\Aye_{i,j} = w(\{i,j\})/2$ and
  $A_G = D^{-1/2} A D^{-1/2}$ be its normalized adjacency matrix.

  Suppose $\Aye_G = \sum_{i=1}^n \lambda_i   u_i   u_i^\top$
  is a spectral decomposition of $\Aye_G$. Set $\matr \Why' =
  \Dee^{1/2} \Why \Dee^{1/2}$. For convenience, define the matrix
  $\matr X$ as $\matr X(i, j)= \ip{  u_i}{\matr \Why'   u_j}$
  and set $p(i) = \matr X(i,i)$. We claim that $p$ is a
  probability distribution. Since $\matr \Why'$ is positive
  semi-definite, we have that $p(i)\ge 0$. Moreover, $\sum_{i=1}^n
  p(i) = 1$ as
  \[
    1 = \Ex{i \sim V(G)}{\ip{  v_i}{  v_i}} = \Tr(\matr \Why') = \Tr(\matr X) = \sum_{i=1}^n \matr X(i,i) 
    = \sum_{i=1}^n p(i).
  \]
  Let $m'$ be the largest value in $[n]$ satisfying $\lambda_{m'}
  \ge 1 - c \cdot \epsilon$. By Cauchy-Schwarz\footnote{In~\cite{BarakRS11},
  there was a minor bug in the application this Cauchy-Schwarz, which led to a
  bound of $(1-1/c)$ instead of $(1-1/c)^2$ in the lemma, leading to a global
  correlation bound of $\Omega(\rho)$ instead of $\Omega(\rho^2)$ as indicated
  in \cref{cor:local_to_global_brs}.},
  \[
  q = \sum_{i=1}^{m'} p(i) \le \sqrt{m'} \sqrt{ \sum_{i=1}^{m'} p(i)^2} \le \sqrt{m'} \sqrt{ \sum_{i,j} (\matr X(i,j))^2} \le \sqrt{\frac{m'}{m}},
  \]
  where the last inequality follows from our assumption that
  \[
  \frac{1}{m} \ge \Ex{i,j \sim V(G)}{\ip{  v_i}{  v_j}^2} = \ip{\matr \Why'}{\matr \Why'} = \ip{\matr X}{\matr X} = \sum_{i,j} \matr X(i,j)^2.
  \]
  Then
  \[
  1 - \epsilon \le  \Ex{ij \sim E(G)}{\ip{  v_i}{  v_j}} = \ip{\matr A}{\matr \Why} = \ip{\matr A_G}{\matr X} = \sum_{i=1}^n \lambda_i \matr X(i,i),
  \]
  implies that
  \[
  1 - \epsilon \le \sum_{i=1}^n \lambda_i \cdot \matr X(i,i) \le \sum_{i=1}^{m'} p(i) + \left(1 - c \cdot \epsilon \right) \sum_{i=m'+1}^n p(i) = 1 - c \cdot \epsilon \left(1-q\right).
  \]
  Finally, using the bound on $q$ we obtain
  \[
  \left(1-\frac{1}{c}\right)\sqrt{m} \le \sqrt{m'},
  \]
  from which the lemma readily follows.
\end{proof}

As a corollary it follows that local correlation implies global
correlation.
\begin{corollary}[Adapted from Lemma 4.1 of~\cite{BarakRS11}]\label{cor:local_to_global_brs}
  Let $G$ be an undirected weighted graph. Suppose $v_1,\dots,v_n \in
  \mathbb{R}^n$ are vectors in the unit ball such that
  \[
  \Ex{ij \sim E(G)}{\ip{  v_i}{  v_j}} \ge \rho,
  \]
  then
  \[
  \Ex{i,j \sim V(G)}{\ip{  v_i}{  v_j}^2} \ge \frac{\rho^2}{4 \cdot \textup{rank}_{\rho/4}(G)}.
  \]
  In particular, we have
  \[
  \Ex{i,j \sim V(G)}{\lvert \ip{  v_i}{  v_j} \rvert} \ge \frac{\rho^2}{4 \cdot \textup{rank}_{\rho/4}(G)}.
  \]
\end{corollary}

\begin{proof}
  If all $  v_1,\dots,  v_n$ are zero, the result trivially follows so
  assume that this is not the case. Then $\alpha = \Ex{i \sim V(G)}{\ip{  v_i}{  v_i}} > 0$.
  Also, $\alpha \le 1$ since the vectors lie in the unit ball. Let
  $v_i' = v_i/\sqrt{\alpha}$. By construction
  \begin{equation}\label{eq:loc_to_glob_assumption}
  \Ex{i \sim V(G)}{\ip{  v_i'}{  v_i'}} = 1, \qquad \Ex{ij \sim E(G)}{\ip{  v_i'}{  v_j'}} \ge \frac{\rho}{\alpha}.
  \end{equation}
  Under these assumptions we want to
  apply~\cref{ap:local_to_glob_inner} in the contra-positive, but
  first we set some parameters. Let $\rho' = \rho/(2\alpha)$,
  $\epsilon=1-\rho'$ and $c = (1-\rho'/2)/(1-\rho')$.
  Then
  \[
  1-\frac{1}{c} = \frac{\rho'/2}{1-\rho'/2} \le \rho',
  \]
  and
  \[
  1- c \cdot \epsilon = \frac{\rho'}{2}.
  \]
  Now, considering the contra-positive of
  the~\cref{ap:local_to_glob_inner} under
  the~\cref{eq:loc_to_glob_assumption} we obtain
  \[
  \Ex{i,j \sim V(G)}{\ip{  v_i'}{  v_j'}^2} > \frac{1}{m} \ge \frac{(\rho')^2}{\textup{rank}_{\rho'/2}(G)},
  \]  
  since $\textup{rank}_{\rho'/2}(G) < (\rho')^2m$ as
  $\lambda_{(\rho')^2m} < \rho'/2$. Or equivalently
  \[
  \Ex{i,j \sim V(G)}{\frac{\ip{  v_i}{  v_j}^2}{\alpha^2}} = \Ex{i,j \sim V(G)}{\ip{  v_i'}{  v_j'}^2} \ge \frac{\rho^2}{4 \alpha^2 \cdot \textup{rank}_{\rho/(4\alpha)}(G)} \ge \frac{\rho^2}{4 \alpha^2 \cdot \textup{rank}_{\rho/4}(G)},
  \]
  where the last inequality follows form the fact that $\alpha \le 1$.
\end{proof}

To finish the proof of Lemma 5.4, we state the
following~\cref{fact:vectorization_byproduct} extracted
from~\cite{BarakRS11}.
\begin{fact}[Adapted from \cite{BarakRS11}]\label{fact:vectorization_byproduct}
  Let $\set{\rv Y_1, \ldots, \rv Y_n}$ be a 2-local PSD ensemble where
  each $\rv Y_i$ can take at most $q$ values. Suppose
  \[ \ee = \Ex{\set{i,j} \sim \Pi_2}{\norm{\set{\rv Y_{ij} } - \set{\rv Y_i}\set{\rv Y_j}  }_1}.\]  
  Then there exist vectors $v_1,\dots,v_n$ in the unit ball such that
  \begin{equation}\label{fact:app_vec_lower_bound}
    \Ex{ij \sim E(G)}{\ip{ v_i}{ v_j}} \ge \frac{1}{q^2} \cdot \Ex{ij \sim \Pi_2}{\norm{\set{\rv Y_{ij} } - \set{\rv Y_i}\set{\rv Y_j}  }_1^2} \ge \frac{\ee^2}{q^2},
  \end{equation}
  and
  \begin{equation}\label{fact:app_vec_upper_bound}
    \Ex{i,j \sim V(G)}{\Var{\rv Y_i} - \Ex{\set{\rv Y_j}}{\Var{\rv Y_i \mid \rv Y_j }}} \ge \Ex{i,j \sim V(G)}{\abs{\ip{v_i}{v_j}}}.
  \end{equation}
\end{fact}

Now we are ready to prove the key result from~\cite{BarakRS11} used in
our proof.
\begin{lemma}[Lemma 5.4 from \cite{BarakRS11} (restatement of~\cref{lem:brs-dec2})]
    Let $G = (V, E, \Pi_2)$ be a weighted graph, $\set{\rv Y_1, \ldots, \rv Y_n}$ a local PSD ensemble, where we have
    $\supp(\rv Y_i) \le q$ for every $i \in V$,  and $q \ge 0$. If $\ee \ge 0$ is a lower bound on 
    the expected statistical difference between independent
    and correlated sampling along the edges,\ie
    \[ \ee \le \Ex{\set{i,j} \sim \Pi_2}{\norm{\set{\rv Y_{ij} } - \set{\rv Y_i}\set{\rv Y_j}  }_1}.\]
    Then, conditioning on a random vertex decreases the variances,
    \[ \Ex{i, j \sim \Pi_1}{\Ex{ \set{\rv Y_j} }{ \Var{\rv Y_i \mid \rv Y_j } } } \le \Ex{i \sim \Pi_1}{\Var{\rv Y_i}} - \frac{\ee^4}{4 q^4 \cdot \Rank_{\ee^2/\left(4q^2\right)}(G) }.\]
\end{lemma}

\begin{proof}
  Using~\cref{fact:app_vec_lower_bound} there exist vectors
  $v_1,\dots,v_n$ such that~\cref{fact:vectorization_byproduct}
  implies
  \[
  \Ex{ij \sim E(G)}{\ip{ v_i}{ v_j}} \ge \frac{\ee^2}{q^2}.
  \]
  From~\cref{cor:local_to_global_brs} we obtain
  \[
  \Ex{i,j \sim V(G)}{\abs{\ip{ v_i}{ v_j}}} \ge \frac{\ee^4}{4q^4 \cdot \textup{rank}_{\ee^2/\left(4q^2\right)}(G)}.
  \]
  Finally, using~\cref{fact:app_vec_upper_bound} we get
  \[
  \Ex{i,j \sim V(G)}{\Var{\rv Y_i} - \Ex{\set{\rv Y_j}}{\Var{\rv Y_i \mid \rv Y_j }}} \ge \Ex{i,j \sim V(G)}{\abs{\ip{v_i}{v_j}}} \ge \frac{\ee^4}{4q^4\cdot \textup{rank}_{\ee^2/\left(4q^2\right)}(G)},
  \]
  as claimed.
\end{proof}

\section{Harmonic Analysis on HDXs}\label{app:harmonic_hdx}

We provide the proofs of known facts used
in~\cref{subsec:hdx_harmonic_analysis}.

\begin{definition}[From~\cite{DiksteinDFH18}]
  We say that $d$-sized complex $X$ is \textit{proper} provided
  $\ker{\left(\Up_{i}\right)}$ is trivial for $1 \le i < d$.
\end{definition}

We will need the following decomposition.
\begin{claim}
Let $\matr A\colon V \to W$ where $V$ and $W$ are finite dimensional inner product spaces. Then
\[
    V = \ker{\matr A} \oplus \im{\matr A^{\dag}}.
\]
\end{claim}
\begin{proof}
    We show that $\ker{\matr A} = \left(\im{\matr
      A^{\dag}}\right)^{\perp}$.  Recall that $ v \in \left(\im{\matr
      A^{\dag}}\right)^{\perp}$ if and only if $\ip{\matr A^{\dag} w}{
      v} = 0$ for every $ w \in W$. This is equivalent to $\ip{
      w}{\matr A v} = 0$ for every $ w \in W$, implying $\matr A v=0$.
\end{proof}

\begin{lemma}[From~\cite{DiksteinDFH18}]\label{lemma:hdx_space_decomp}
  We have
  \[
  C^k = \sum_{i=0}^{k} C^{k}_i.
  \]
  Moreover, if X is proper then
  \[
  C^k = \bigoplus_{i=0}^{k} C^{k}_i,
  \]
  and $\dim{C^k_i} = \lvert X(i) \rvert - \lvert X(i-1) \rvert$.
\end{lemma}

\begin{proof}
  We induct on $k$. For $k=0$, $X(0) = \{ \emptyset\}$ and $C^{0} =
  C^{0}_{0}$. Now suppose $k > 0$. Since $\Dee_k$ and $\Up_{k-1}$ are
  adjoints, we have $C^k = \ker{\Dee_k} \oplus \im{\Up_{k-1}}$ or
  equivalently
  \begin{equation}\label{eq:c_k_space_decomp}
    C^k = \ker{\Dee_k} \oplus \Up_{k-1} C^{k-1}.
  \end{equation}
  Using the induction hypothesis $C^{k-1} = \sum_{i=0}^{k-1}
  C^{k-1}_i$. Note that
  \[
  \Up_{k-1} C^{k-1}_i = \set*{\Up_{k-1} \Up^{k-1-i} h_i~\vert~h_i \in H_i } = C^{k}_i.
  \]
  Thus $C^k = C^k_k + \sum_{i=0}^{k-1} C^{k}_i$. Assuming
  $\ker{\left(\Up_{i}\right)}$ is trivial for $0 \le i < k$ we obtain
  \[
  \dim{C^k_i} = \dim{H_i} = \dim{C^{i}} - \dim{C^{i-1}} = \lvert X(i) \rvert - \lvert X(i-1) \rvert,
  \]
  where the second equality follows from~\cref{eq:c_k_space_decomp}. Hence
  $\dim{C^k} = \sum_{i=0}^k \dim{C^k_i}$. This implies that each $C^k_i \cap \sum_{j\ne i} C^k_j$
  is trivial and so we have a direct sum as claimed.
\end{proof}

\begin{corollary}[From~\cite{DiksteinDFH18}]
  Let $f \in C^k$. If $X$ is proper, then $f$ can be written
  uniquely as
  \[
  f = f_{0} + \cdots + f_{k},
  \]
  where $f_i \in C^k_i$.
\end{corollary}



%% file: main.bbl
\newcommand{\etalchar}[1]{$^{#1}$}
\begin{thebibliography}{KMOW17}

\bibitem[AAV13]{ADV13}
Dorit Aharonov, Itai Arad, and Thomas Vidick.
\newblock Guest column: The quantum {PCP} conjecture.
\newblock {\em SIGACT News}, 44(2), 2013.

\bibitem[AJQ{\etalchar{+}}19]{AJQST19}
Vedat~Levi Alev, Fernando~Granha Jeronimo, Dylan Quintana, Shashank Srivastava,
  and Madhur Tulsiani.
\newblock List decoding of direct sum codes.
\newblock Manuscript, 2019.

\bibitem[AKK{\etalchar{+}}08]{AKKSTV08}
Sanjeev Arora, Subhash Khot, Alexandra Kolla, David Steurer, Madhur Tulsiani,
  and Nisheeth Vishnoi.
\newblock Unique games on expanding constraint graphs are easy.
\newblock In {\em Proceedings of the 40th ACM Symposium on Theory of
  Computing}, 2008.

\bibitem[ALGV18]{ALOV18:log}
Nima Anari, Kuikui Liu, Shayan~Oveis Gharan, and Cynthia Vinzant.
\newblock Log-concave polynomials {II}: High-dimensional walks and an {FPRAS}
  for counting bases of a matroid.
\newblock {\em arXiv preprint arXiv:1811.01816}, 2018.

\bibitem[BH13]{BrandaoH13}
Fernando G. S.~L. Brand{\~{a}}o and Aram~Wettroth Harrow.
\newblock Product-state approximations to quantum ground states.
\newblock In {\em Proceedings of the 45th ACM Symposium on Theory of
  Computing}, pages 871--880, 2013.

\bibitem[BRS11]{BarakRS11}
Boaz Barak, Prasad Raghavendra, and David Steurer.
\newblock Rounding semidefinite programming hierarchies via global correlation.
\newblock In {\em FOCS}, pages 472--481, 2011.

\bibitem[BS02]{Bernstein02}
M.~{Bernstein} and N.~J.~A. {Sloane}.
\newblock {Some Canonical Sequences of Integers}.
\newblock {\em arXiv Mathematics e-prints}, page math/0205301, May 2002.
\newblock \href {http://arxiv.org/abs/math/0205301}
  {\path{arXiv:math/0205301}}.

\bibitem[BSHR05]{BHR05}
Eli Ben-Sasson, Prahladh Harsha, and Sofya Raskhodnikova.
\newblock Some {3CNF} properties are hard to test.
\newblock {\em SIAM Journal on Computing}, 35(1):1--21, 2005.

\bibitem[CTZ18]{ConlonTZ18}
David {Conlon}, Jonathan {Tidor}, and Yufei {Zhao}.
\newblock {Hypergraph expanders of all uniformities from Cayley graphs}.
\newblock {\em arXiv e-prints}, page arXiv:1809.06342, September 2018.
\newblock \href {http://arxiv.org/abs/1809.06342} {\path{arXiv:1809.06342}}.

\bibitem[DD19]{DD19}
Yotam Dikstein and Irit Dinur.
\newblock Agreement testing theorems on layered set systems.
\newblock In {\em Proceedings of the 60th IEEE Symposium on Foundations of
  Computer Science}, 2019.

\bibitem[DDFH18]{DiksteinDFH18}
Yotam Dikstein, Irit Dinur, Yuval Filmus, and Prahladh Harsha.
\newblock Boolean function analysis on high-dimensional expanders.
\newblock In {\em Approximation, Randomization, and Combinatorial Optimization.
  Algorithms and Techniques, {APPROX/RANDOM} 2018, August 20-22, 2018 -
  Princeton, NJ, {USA}}, pages 38:1--38:20, 2018.

\bibitem[DHK{\etalchar{+}}18]{DinurHKNT18}
Irit Dinur, Prahladh Harsha, Tali Kaufman, Inbal~Livni Navon, and Amnon
  Ta{-}Shma.
\newblock List decoding with double samplers.
\newblock {\em Electronic Colloquium on Computational Complexity {(ECCC)}},
  25:136, 2018.

\bibitem[DK12]{DinurK12}
Irit Dinur and Tali Kaufman.
\newblock Locally testable codes and expanders.
\newblock Manuscript, 2012.

\bibitem[DK17]{DinurK17}
Irit Dinur and Tali Kaufman.
\newblock High dimensional expanders imply agreement expanders.
\newblock In {\em Proceedings of the 58th IEEE Symposium on Foundations of
  Computer Science}, pages 974--985, 2017.

\bibitem[Fil16]{Filmus16}
Yuval Filmus.
\newblock An orthogonal basis for functions over a slice of the boolean
  hypercube.
\newblock {\em Electr. J. Comb.}, 23(1):P1.23, 2016.

\bibitem[FK96]{FK96:focs}
A.~Frieze and R.~Kannan.
\newblock The regularity lemma and approximation schemes for dense problems.
\newblock In {\em Proceedings of the 37th IEEE Symposium on Foundations of
  Computer Science}, 1996.

\bibitem[GM15]{GodsilM15}
Christopher Godsil and Karen Meagher.
\newblock {\em Erd\H{o}s-Ko-Rado Theorems: Algebraic Approaches}.
\newblock Cambridge Studies in Advanced Mathematics. Cambridge University
  Press, 2015.

\bibitem[GS11]{GuruswamiS11}
Venkatesan Guruswami and Ali~Kemal Sinop.
\newblock Lasserre hierarchy, higher eigenvalues, and approximation schemes for
  graph partitioning and quadratic integer programming with {PSD} objectives.
\newblock In {\em Proceedings of the 52nd IEEE Symposium on Foundations of
  Computer Science}, pages 482--491, 2011.

\bibitem[GS12]{GS12:faster}
Venkatesan Guruswami and Ali~Kemal Sinop.
\newblock Faster {SDP} hierarchy solvers for local rounding algorithms.
\newblock In {\em Proceedings of the 53rd IEEE Symposium on Foundations of
  Computer Science}, pages 197--206. IEEE, 2012.

\bibitem[Kho02]{Khot02:unique}
Subhash Khot.
\newblock On the power of unique 2-prover 1-round games.
\newblock In {\em Proceedings of the 34th ACM Symposium on Theory of
  Computing}, pages 767--775, 2002.

\bibitem[KKL16]{KaufmanKL16}
Tali Kaufman, David Kazhdan, and Alexander Lubotzky.
\newblock Isoperimetric inequalities for ramanujan complexes and topological
  expanders.
\newblock {\em Geometric and Functional Analysis}, 26(1):250--287, Feb 2016.

\bibitem[KM17]{KaufmanM17}
Tali Kaufman and David Mass.
\newblock High dimensional random walks and colorful expansion.
\newblock In {\em Proceedings of the 8th Conference on Innovations in
  Theoretical Computer Science}, pages 4:1--4:27, 2017.

\bibitem[KM18]{KaufmanM18}
Tali Kaufman and David Mass.
\newblock Good distance lattices from high dimensional expanders.
\newblock {\em CoRR}, abs/1803.02849, 2018.
\newblock URL: \url{http://arxiv.org/abs/1803.02849}, \href
  {http://arxiv.org/abs/1803.02849} {\path{arXiv:1803.02849}}.

\bibitem[KMOW17]{KothariMOW17}
Pravesh Kothari, Ryuhei Mori, Ryan O'Donnell, and David Witmer.
\newblock Sum of squares lower bounds for refuting any {CSP}.
\newblock In {\em Proceedings of the 49th ACM Symposium on Theory of
  Computing}, 2017.

\bibitem[KO18a]{KaufmanO18Constr}
Tali Kaufman and Izhar Oppenheim.
\newblock Construction of new local spectral high dimensional expanders.
\newblock In {\em Proceedings of the 50th ACM Symposium on Theory of
  Computing}, STOC 2018, pages 773--786. ACM, 2018.

\bibitem[KO18b]{KaufmanO18Walk}
Tali Kaufman and Izhar Oppenheim.
\newblock High order random walks: Beyond spectral gap.
\newblock In {\em Approximation, Randomization, and Combinatorial Optimization.
  Algorithms and Techniques, {APPROX/RANDOM} 2018, August 20-22, 2018 -
  Princeton, NJ, {USA}}, pages 47:1--47:17, 2018.

\bibitem[LLP17]{LLP17}
Eyal {Lubetzky}, Alex {Lubotzky}, and Ori {Parzanchevski}.
\newblock {Random walks on Ramanujan complexes and digraphs}.
\newblock {\em arXiv e-prints}, page arXiv:1702.05452, Feb 2017.
\newblock \href {http://arxiv.org/abs/1702.05452} {\path{arXiv:1702.05452}}.

\bibitem[LSV05a]{LubotzkySV05b}
Alexander Lubotzky, Beth Samuels, and Uzi Vishne.
\newblock Explicit constructions of ramanujan complexes of type ad.
\newblock {\em Eur. J. Comb.}, 26(6):965--993, August 2005.

\bibitem[LSV05b]{LubotzkySV05a}
Alexander Lubotzky, Beth Samuels, and Uzi Vishne.
\newblock Ramanujan complexes of type{\~a}d.
\newblock {\em Israel Journal of Mathematics}, 149(1):267--299, Dec 2005.

\bibitem[MM11]{MakarychevM11}
Konstantin Makarychev and Yury Makarychev.
\newblock How to play unique games on expanders.
\newblock In {\em Approximation and Online Algorithms}, pages 190--200.
  Springer Berlin Heidelberg, 2011.

\bibitem[MR17]{MR17}
Pasin Manurangsi and Prasad Raghavendra.
\newblock A birthday repetition theorem and complexity of approximating dense
  csps.
\newblock In {\em Proceedings of the 44th International Colloquium on Automata,
  Languages and Programming}. Schloss Dagstuhl-Leibniz-Zentrum fuer Informatik,
  2017.

\bibitem[OGT15]{OGT15}
Shayan Oveis~Gharan and Luca Trevisan.
\newblock A new regularity lemma and faster approximation algorithms for low
  threshold rank graphs.
\newblock {\em Theory of Computing}, 11(9):241--256, 2015.
\newblock URL: \url{http://www.theoryofcomputing.org/articles/v011a009}, \href
  {http://dx.doi.org/10.4086/toc.2015.v011a009}
  {\path{doi:10.4086/toc.2015.v011a009}}.

\bibitem[PRT16]{ParzanchevskiRT2016}
Ori Parzanchevski, Ron Rosenthal, and Ran~J. Tessler.
\newblock Isoperimetric inequalities in simplicial complexes.
\newblock {\em Combinatorica}, 36(2):195--227, Apr 2016.

\bibitem[RW17]{RW17:sos}
Prasad Raghavendra and Benjamin Weitz.
\newblock On the bit complexity of sum-of-squares proofs.
\newblock In {\em Proceedings of the 44th International Colloquium on Automata,
  Languages and Programming}. Schloss Dagstuhl-Leibniz-Zentrum fuer Informatik,
  2017.

\bibitem[Sag13]{sagan13}
B.E. Sagan.
\newblock {\em The Symmetric Group: Representations, Combinatorial Algorithms,
  and Symmetric Functions}.
\newblock Graduate Texts in Mathematics. Springer New York, 2013.

\bibitem[WJ04]{WJ04:treewidth}
Martin~J Wainwright and Michael~I Jordan.
\newblock Treewidth-based conditions for exactness of the {Sherali-Adams} and
  {Lasserre} relaxations.
\newblock Technical report, Technical Report 671, University of California,
  Berkeley, 2004.

\bibitem[YZ14]{YZ14:dense}
Yuichi Yoshida and Yuan Zhou.
\newblock Approximation schemes via {Sherali-Adams} hierarchy for dense
  constraint satisfaction problems and assignment problems.
\newblock In {\em Proceedings of the 5th Conference on Innovations in
  Theoretical Computer Science}, pages 423--438. ACM, 2014.

\end{thebibliography}
